\documentclass{article}
\usepackage{amssymb}
\usepackage{amsfonts}
\usepackage{amsmath}

\newtheorem{theorem}{Theorem}

\newtheorem{corollary}[theorem]{Corollary}

\newtheorem{definition}[theorem]{Definition}
\newtheorem{example}[theorem]{Example}

\newtheorem{proposition}[theorem]{Proposition}
\newtheorem{remark}[theorem]{Remark}

\newenvironment{proof}[1][Proof]{\noindent\textbf{#1.} }{\ \rule{0.5em}{0.5em}}
\input{tcilatex}
\begin{document}

\title{Mechanical Systems in the Generalized Lie Algebroids Framework }
\author{Constantin M. ARCU\c{S} \\
\ \ \\
[0pt] \ \ \\
[0pt] 
\begin{tabular}{c}
Secondary School \textquotedblleft CORNELIUS RADU\textquotedblright , \\ 
R\u{a}dine\c{s}ti Village, 217196, Gorj County, Rom\^{a}nia \\ 
e-mail: constantin\_arcus@yahoo.com, c\_arcus@radinesti.ro%
\end{tabular}%
}
\maketitle

\begin{abstract}
\emph{Mechanical systems} called by use, \emph{mechanical }$\left( \rho
,\eta\right) $\emph{-systems, Lagrange mechanical }$\left( \rho ,\eta
\right) $\emph{-systems} or \emph{Finsler mechanical }$\left( \rho ,\eta
\right) $\emph{-systems} are presented. The canonical $\left( \rho ,\eta
\right) $\emph{-}semi(spray) associated to a mechanical $\left( \rho ,\eta
\right) $-system is obtained. New and important results are obtained in the
particular case of Lie algebroids. The Lagrange mechanical $(\rho ,\eta)$%
-systems are the spaces necessary to develop a new Lagrangian formalism. We
obtain the $(\rho ,\eta)$-semispray associated to a regular Lagrangian $L$
and external force $F_{e}$ and we derive the equations of Euler-Lagrange
type. So, a new solution for the Weinstein's Problem in the general
framework of generalized Lie algebroids is presented. \ \ \bigskip\newline
\end{abstract}

\tableofcontents

\section{Introduction}

\ \ \ \ \ 

In general, if $\mathcal{C}$ is a category, then we denote $\left\vert 
\mathcal{C}\right\vert $ the class of objects and for any $A,B{\in }%
\left\vert \mathcal{C}\right\vert $, we denote $\mathcal{C}\left( A,B\right) 
$ the set of morphisms of $A$ source and $B$ target and $Iso_{\mathcal{C}%
}\left( A,B\right) $ the set of $\mathcal{C}$-isomorphisms of $A$ source and 
$B$ target. Let $\mathbf{LieAlg},~\mathbf{Mod,}$ $\mathbf{Man}$ and $\mathbf{%
B}^{\mathbf{v}}$ be the category of Lie algebras, modules, manifolds and
vector bundles respectively.

We know that if 
\begin{equation*}
\left( E,\pi ,M\right) \in \left\vert \mathbf{B}^{\mathbf{v}}\right\vert ,
\end{equation*}%
\begin{equation*}
\Gamma \left( E,\pi ,M\right) =\left\{ u\in \mathbf{Man}\left( M,E\right)
:u\circ \pi =Id_{M}\right\}
\end{equation*}%
and 
\begin{equation*}
\mathcal{F}\left( M\right) =\mathbf{Man}\left( M,\mathbb{R}\right) ,
\end{equation*}%
then $\left( \Gamma \left( E,\pi ,M\right) ,+,\cdot \right) $ is a $\mathcal{%
F}\left( M\right) $-module.

If \ $\left( \varphi ,\varphi _{0}\right) \in \mathbf{B}^{\mathbf{v}}\left(
\left( E,\pi ,M\right) ,\left( E^{\prime },\pi ^{\prime },M^{\prime }\right)
\right) $ such that $\varphi _{0}\in Iso_{\mathbf{Man}}\left( M,M^{\prime
}\right) ,$ then, using the operation 
\begin{equation*}
\begin{array}{ccc}
\mathcal{F}\left( M\right) \times \Gamma \left( E^{\prime },\pi ^{\prime
},M^{\prime }\right) & ^{\underrightarrow{~\ \ \cdot ~\ \ }} & \Gamma \left(
E^{\prime },\pi ^{\prime },M^{\prime }\right) \\ 
\left( f,u^{\prime }\right) & \longmapsto & f\circ \varphi _{0}^{-1}\cdot
u^{\prime }%
\end{array}%
\end{equation*}%
it results that $\left( \Gamma \left( E^{\prime },\pi ^{\prime },M^{\prime
}\right) ,+,\cdot \right) $ is a $\mathcal{F}\left( M\right) $-module and we
obtain the $\mathbf{Mod}$-morphism%
\begin{equation*}
\begin{array}{ccc}
\Gamma \left( E,\pi ,M\right) & ^{\underrightarrow{~\ \ \Gamma \left(
\varphi ,\varphi _{0}\right) ~\ \ }} & \Gamma \left( E^{\prime },\pi
^{\prime },M^{\prime }\right) \\ 
u & \longmapsto & \Gamma \left( \varphi ,\varphi _{0}\right) u%
\end{array}%
\end{equation*}%
defined by 
\begin{equation*}
\begin{array}{c}
\Gamma \left( \varphi ,\varphi _{0}\right) u\left( y\right) =\varphi \left(
u_{\varphi _{0}^{-1}\left( y\right) }\right) =\left( \varphi \circ u\circ
\varphi _{0}^{-1}\right) \left( y\right) ,%
\end{array}%
\end{equation*}%
for any $y\in M^{\prime }.$

If $\left( F,\nu ,M\right) \in \left\vert \mathbf{B}^{\mathbf{v}}\right\vert 
$ so that there exists 
\begin{equation*}
\begin{array}{c}
\left( \rho ,Id_{M}\right) \in \mathbf{B}^{\mathbf{v}}\left( \left( F,\nu
,M\right) ,\left( TM,\tau _{M},M\right) \right)%
\end{array}%
\end{equation*}%
and an operation 
\begin{equation*}
\begin{array}{ccc}
\Gamma \left( F,\nu ,M\right) \times \Gamma \left( F,\nu ,M\right) & ^{%
\underrightarrow{~\ \left[ ,\right] _{F}~\ }} & \Gamma \left( F,\nu ,M\right)
\\ 
\left( u,v\right) & \longmapsto & \left[ u,v\right] _{F}%
\end{array}%
\end{equation*}%
with the following properties:\bigskip

\noindent $\qquad LA_{1}$. \emph{the equality holds good }%
\begin{equation*}
\begin{array}{c}
\left[ u,f\cdot v\right] _{F}=f\left[ u,v\right] _{F}+\Gamma \left( \rho
,Id_{M}\right) \left( u\right) f\cdot v,%
\end{array}%
\end{equation*}%
\qquad \quad\ \ \emph{for all }$u,v\in \Gamma \left( F,\nu ,M\right) $\emph{%
\ and} $f\in \mathcal{F}\left( M\right) .$

\medskip $LA_{2}$. \emph{the }$4$\emph{-tuple} $\left( \Gamma \left( F,\nu
,M\right) ,+,\cdot ,\left[ ,\right] _{F}\right) $ \emph{is a Lie} $\mathcal{F%
}\left( M\right) $\emph{-algebra,}

$LA_{3}$. \emph{the }$\mathbf{Mod}$\emph{-morphism }$\Gamma \left( \rho
,Id_{M}\right) $\emph{\ is a }$\mathbf{LieAlg}$\emph{-morphism of }%
\begin{equation*}
\left( \Gamma \left( F,\nu ,M\right) ,+,\cdot ,\left[ ,\right] _{F}\right)
\end{equation*}%
\emph{\ source and }%
\begin{equation*}
\left( \Gamma \left( TN,\tau _{N},M\right) ,+,\cdot ,\left[ ,\right]
_{TM}\right)
\end{equation*}%
\emph{target, then the triple }%
\begin{equation}
\begin{array}{c}
\left( \left( F,\nu ,M\right) ,\left[ ,\right] _{F},\left( \rho
,Id_{M}\right) \right)%
\end{array}
\label{eq1}
\end{equation}%
\emph{is an object of the category }$\mathbf{LA}$\emph{\ of Lie algebroids}$%
\mathbf{.}$ The couple $\left( \left[ ,\right] _{F},\left( \rho
,Id_{M}\right) \right) $ is called\emph{\ Lie algebroid structure.}

Used earlier by many authors, the study of Lie algebroids were considerably
improved by J. Pradines in \cite{30}, who noticed that the Lie algebroids
are infinitesimal versions of Lie groupoids in a functorial manner.

A generalization of a Lie algebroid is the \emph{Lie bialgebroid} defined in 
\cite{24} by K.C.H. Mackenzie and P. Xu.

A Lie algebroid $\left( \Gamma \left( F,\nu ,M\right) ,\left[ ,\right]
_{F},\left( \rho ,Id_{M}\right) \right) $ is called \emph{Lie bialgebroid}
if there exists a Lie algebroid structure $\left( \left[ ,\right] _{\overset{%
\ast }{F}},\left( \overset{\ast }{\rho },Id_{M}\right) \right) $ for its
dual $\left( \overset{\ast }{F},\overset{\ast }{\nu },M\right) $ such that 
\begin{equation*}
\begin{array}{c}
d_{\overset{\ast }{F}}\left( \left[ u,v\right] _{F}\right) =\left[ u,d_{%
\overset{\ast }{F}}v\right] _{F}-\left[ v,d_{\overset{\ast }{F}}u\right]
_{F},~\forall u,v\in \Gamma \left( F,\nu ,M\right)%
\end{array}%
\end{equation*}%
where $d_{\overset{\ast }{F}}$ is the exterior differentiation operator for
the exterior differential algebra of the Lie algebroid $\left( \left( 
\overset{\ast }{F},\overset{\ast }{\nu },M\right) ,\left[ ,\right] _{\overset%
{\ast }{F}},\left( \overset{\ast }{\rho },Id_{M}\right) \right) .$

The Courant algebroids defined in \cite{23} by Z.J. Liu, A. Weinstein and P.
Xu are new generalizations of Lie algebroids. This notion is the result of
an effort to unify the Courant bracket and the Manin bracket. (see: \cite{29}%
)

If $\left( F,\nu ,M\right) \in \left\vert \mathbf{B}^{\mathbf{v}}\right\vert 
$ so that there exists 
\begin{equation*}
\begin{array}{c}
\left( \rho ,Id_{M}\right) \in \mathbf{B}^{\mathbf{v}}\left( \left( F,\nu
,M\right) ,\left( TM,\tau _{M},M\right) \right) ,%
\end{array}%
\end{equation*}%
a skewsymmetric and $\mathbb{R}$-linear bracket 
\begin{equation*}
\begin{array}{ccc}
\Gamma \left( F,\nu ,M\right) \times \Gamma \left( F,\nu ,M\right) & ^{%
\underrightarrow{~\ \left[ ,\right] _{F}~\ }} & \Gamma \left( F,\nu ,M\right)
\\ 
\left( u,v\right) & \longmapsto & \left[ u,v\right] _{F}%
\end{array}%
,
\end{equation*}%
a nondegenerate symmetric and bilinear form 
\begin{equation*}
\begin{array}{ccc}
\Gamma \left( F,\nu ,M\right) \times \Gamma \left( F,\nu ,M\right) & ^{%
\underrightarrow{~\ \left\langle ,\right\rangle _{F}~\ }} & \mathcal{F}%
\left( M\right) \\ 
\left( u,v\right) & \longmapsto & \left\langle u,v\right\rangle _{F}%
\end{array}%
\end{equation*}%
and we have the section application 
\begin{equation*}
\begin{array}{ccc}
\mathcal{F}\left( M\right) & ^{\underrightarrow{~\ \mathcal{S}~\ }} & \Gamma
\left( F,\nu ,M\right) \\ 
f & \longmapsto & \mathcal{S}\left( f\right)%
\end{array}%
\end{equation*}%
defined by%
\begin{equation*}
\begin{array}{c}
\left\langle \mathcal{S}\left( f\right) ,u\right\rangle _{F}=\frac{1}{2}%
\cdot \Gamma \left( \rho ,Id_{M}\right) u\left( f\right) ,~\forall u\in
\Gamma \left( F,\nu ,M\right)%
\end{array}%
\end{equation*}%
such that the following properties are satisfied:

\noindent $\qquad CA_{1}$. the equality holds good\emph{\ }%
\begin{equation*}
\begin{array}{c}
\left[ u,f\cdot v\right] _{F}=f\left[ u,v\right] _{F}+\Gamma \left( \rho
,Id_{M}\right) u\left( f\right) \cdot v-\left\langle u,v\right\rangle
_{F}\cdot \mathcal{S}\left( f\right) ,%
\end{array}%
\end{equation*}%
\qquad \quad\ \ for all $u,v\in \Gamma \left( F,\nu ,M\right) $\ and $f\in 
\mathcal{F}\left( M\right) .$

\medskip $CA_{2}$. $\underset{cycl}{\sum }\left[ \left[ u,v\right] _{F},z%
\right] _{F}=\frac{1}{3}\cdot \mathcal{S}\left( \underset{cycl}{\sum }%
\left\langle \left[ u,v\right] _{F},z\right\rangle _{F}\right) ,$

$CA_{3}$. $\Gamma \left( \rho ,Id_{M}\right) \left[ u,v\right] _{F}=\left[
\Gamma \left( \rho ,Id_{M}\right) u,\Gamma \left( \rho ,Id_{M}\right) v%
\right] _{TM},$ for all $u,v\in \Gamma \left( F,\nu ,M\right) \emph{,}$

$CA_{4}$. $\left\langle \mathcal{S}\left( f\right) ,\mathcal{S}\left(
g\right) \right\rangle _{F}=0,$ for all $f,g\in \mathcal{F}\left( M\right) ,$

$CA_{5}$. $\Gamma \left( \rho ,Id_{M}\right) u\left( \left\langle
v,z\right\rangle _{F}\right) =\left\langle \left[ u,v\right] _{F}+S\left(
\left\langle u,v\right\rangle _{F}\right) ,z\right\rangle _{F}+\left\langle
u,\left[ v,z\right] _{F}+\mathcal{S}\left( \left\langle v,z\right\rangle
_{F}\right) \right\rangle _{F},$ for all $u,v,z\in \Gamma \left( F,\nu
,M\right) \emph{,}$

then the $5$-tuple\emph{\ }%
\begin{equation}
\begin{array}{c}
\left( \left( F,\nu ,M\right) ,\left[ ,\right] _{F},\left( \rho
,Id_{M}\right) ,\left\langle ,\right\rangle _{F},\mathcal{S}\right)%
\end{array}
\label{eq2}
\end{equation}%
\emph{\ is a Courant algebroid}\textbf{.} The 4-tuple $\left( \left[ ,\right]
_{F},\left( \rho ,Id_{M}\right) ,\left\langle ,\right\rangle _{F},\mathcal{S}%
\right) $ is called\emph{\ Courant algebroid structure.}

Z.J.Liu, A.Weinstein and P.Xu \cite{23}\ proved that any Lie bialgebroid is
a Courant algebroid and some general new problems are proposed. Trying to
give some possible proofs for these problems, a new class of \emph{%
generalized algebroids} are discovered by Paul Popescu in \cite{35}.

If $\left( F,\nu ,M\right) \in \left\vert \mathbf{B}^{\mathbf{v}}\right\vert 
$ so that there exists 
\begin{equation*}
\begin{array}{c}
\left( \rho ,Id_{M}\right) \in \mathbf{B}^{\mathbf{v}}\left( \left( F,\nu
,M\right) ,\left( TM,\tau _{M},M\right) \right) ,%
\end{array}%
\end{equation*}%
a skewsymmetric and $\mathbb{R}$-linear bracket 
\begin{equation*}
\begin{array}{ccc}
\Gamma \left( F,\nu ,M\right) \times \Gamma \left( F,\nu ,M\right) & ^{%
\underrightarrow{~\ \left[ ,\right] _{F}~\ }} & \Gamma \left( F,\nu ,M\right)
\\ 
\left( u,v\right) & \longmapsto & \left[ u,v\right] _{F}%
\end{array}%
\end{equation*}%
and a submodul $\left( \mathcal{M},+,\cdot \right) $ of $\left( \Gamma
\left( F,\nu ,M\right) ,+,\cdot \right) $ such that $\Gamma \left( \rho
,Id_{M}\right) u=0,$ for all $u\in \mathcal{M}$ and the following properties
are satisfied:

\noindent $\qquad GA_{1}$. $\left[ u,f\cdot v\right] _{F}-f\cdot \left[ u,v%
\right] _{F}-\Gamma \left( \rho ,Id_{M}\right) u\left( f\right) \cdot v\in 
\mathcal{M}$ and $\left[ f\cdot u,v\right] _{F}-f\cdot \left[ u,v\right]
_{F}+\Gamma \left( \rho ,Id_{M}\right) u\left( f\right) \cdot v\in \mathcal{M%
},$ for all\emph{\ }$u,v\in \Gamma \left( F,\nu ,M\right) $\emph{\ }and $%
f\in \mathcal{F}\left( M\right) .$

\medskip $GA_{2}$. $\underset{cycl}{\sum }\left[ \left[ u,v\right] _{F},z%
\right] _{F}\in \mathcal{M},$ for all $u,v,z\in \Gamma \left( F,\nu
,M\right) ,$

$GA_{3}$. $\Gamma \left( \rho ,Id_{M}\right) \left[ u,v\right] _{F}=\left[
\Gamma \left( \rho ,Id_{M}\right) u,\Gamma \left( \rho ,Id_{M}\right) v%
\right] _{TM},$ for all\emph{\ }$u,v\in \Gamma \left( F,\nu ,M\right) \emph{,%
}$

$GA_{4}$. $\left[ u,v\right] _{F}\in \mathcal{M},$ whenever $u$ or $v$ are
in $\mathcal{M},$

then the $4$-tuple%
\begin{equation}
\begin{array}{c}
\left( \left( F,\nu ,M\right) ,\left[ ,\right] _{F},\left( \rho
,Id_{M}\right) ,\mathcal{M}\right)%
\end{array}
\label{eq3}
\end{equation}%
\emph{\ is a generalized algebroid}$\mathbf{.}$ The triple $\left( \left[ ,%
\right] _{F},\left( \rho ,Id_{M}\right) ,\mathcal{M}\right) $ is called\emph{%
\ generalized algebroid structure.}

We know that the secret of the Ehresmann connection is given by the diagrams%
\begin{equation}
\begin{array}{c}
\begin{array}[b]{ccccc}
E &  & \left( TM,\left[ ,\right] _{TM}\right) & ^{\underrightarrow{~\ \ \
Id_{TM}~\ \ }} & \left( TM,\left[ ,\right] _{TM}\right) \\ 
~\downarrow \pi &  & ~\ \ \downarrow \tau _{M} &  & ~\ \ \downarrow \tau _{M}
\\ 
M & ^{\underrightarrow{~\ \ \ Id_{M}~\ \ }} & M & ^{\underrightarrow{~\ \ \
Id_{M}~\ \ }} & M%
\end{array}%
\end{array}
\label{eq4}
\end{equation}%
where $\left( E,\pi ,M\right) $ is a fiber bundle and $\left( \left( TM,\tau
_{M},M\right) ,\left[ ,\right] _{TM},\left( Id_{TM},Id_{M}\right) \right) $
is the standard Lie algebroid.

First, there appeared the idea of changing the standard Lie algebroid with
an arbitrary Lie algebroid as in the diagrams%
\begin{equation}
\begin{array}{c}
\begin{array}[b]{ccccc}
E &  & \left( F,\left[ ,\right] _{F}\right) & ^{\underrightarrow{~\ \ \ \rho
\ \ \ \ }} & \left( TM,\left[ ,\right] _{TM}\right) \\ 
~\downarrow \pi &  & ~\downarrow \nu &  & ~\ \ \downarrow \tau _{M} \\ 
M & ^{\underrightarrow{~\ \ \ Id_{M}~\ \ }} & M & ^{\underrightarrow{~\ \ \
Id_{M}~\ \ }} & M%
\end{array}%
\end{array}
\label{eq5}
\end{equation}

Secondly, there appeared the idea of changing in the previous diagrams the
identities morphisms with arbitrary $\mathbf{Man}$-isomorphisms $h$ and $%
\eta $ as in the diagrams%
\begin{equation}
\begin{array}{c}
\begin{array}[b]{ccccccc}
E &  & \left( F,\left[ ,\right] _{F,h}\right) & ^{\underrightarrow{~\ \ \
\rho \ \ \ \ }} & \left( TM,\left[ ,\right] _{TM}\right) & ^{%
\underrightarrow{~\ \ \ Th\ \ \ \ }} & \left( TN,\left[ ,\right] _{TN}\right)
\\ 
~\downarrow \pi &  & ~\downarrow \nu &  & ~\ \ \downarrow \tau _{M} &  & ~\
\ \downarrow \tau _{N} \\ 
M & ^{\underrightarrow{~\ \ \ h~\ \ }} & N & ^{\underrightarrow{~\ \ \ \eta
~\ \ }} & M & ^{\underrightarrow{~\ \ \ h~\ \ }} & N%
\end{array}%
\end{array}
\label{eq6}
\end{equation}%
where 
\begin{equation*}
\left( \rho ,\eta \right) \in \mathbf{B}^{\mathbf{v}}\left( \left( F,\nu
,M\right) ,\left( TM,\tau _{M},M\right) \right)
\end{equation*}%
and 
\begin{equation*}
\begin{array}{ccc}
\Gamma \left( F,\nu ,N\right) \times \Gamma \left( F,\nu ,N\right) & ^{%
\underrightarrow{~\ \ \left[ ,\right] _{F,h}~\ \ }} & \Gamma \left( F,\nu
,N\right) \\ 
\left( u,v\right) & \longmapsto & \ \left[ u,v\right] _{F,h}%
\end{array}%
\end{equation*}%
is an operation with the following properties:

$GLA_{1}$. \emph{the equality holds good }%
\begin{equation*}
\begin{array}{c}
\left[ u,f\cdot v\right] _{F,h}=f\left[ u,v\right] _{F,h}+\Gamma \left(
Th\circ \rho ,h\circ \eta \right) \left( u\right) f\cdot v,%
\end{array}%
\end{equation*}%
\qquad \quad\ \ \emph{for all }$u,v\in \Gamma \left( F,\nu ,N\right) $\emph{%
\ and} $f\in \mathcal{F}\left( N\right) .$

\medskip $GLA_{2}$. \emph{the }$4$\emph{-tuple} $\left( \Gamma \left( F,\nu
,N\right) ,+,\cdot ,\left[ ,\right] _{F,h}\right) $ \emph{is a Lie} $%
\mathcal{F}\left( N\right) $\emph{-algebra,}

$GLA_{3}$. \emph{the }$\mathbf{Mod}$\emph{-morphism }$\Gamma \left( Th\circ
\rho ,h\circ \eta \right) $\emph{\ is a }$\mathbf{LieAlg}$\emph{-morphism of 
}%
\begin{equation*}
\left( \Gamma \left( F,\nu ,N\right) ,+,\cdot ,\left[ ,\right] _{F,h}\right)
\end{equation*}%
\emph{\ source and }%
\begin{equation*}
\left( \Gamma \left( TN,\tau _{N},N\right) ,+,\cdot ,\left[ ,\right]
_{TN}\right)
\end{equation*}%
\emph{target.}

We say that \emph{the triple}%
\begin{equation}
\begin{array}{c}
\left( \left( F,\nu ,N\right) ,\left[ ,\right] _{F,h},\left( \rho ,\eta
\right) \right)%
\end{array}
\label{eq7}
\end{equation}%
\emph{is a generalized Lie algebroid. }The couple $\left( \left[ ,\right]
_{F,h},\left( \rho ,\eta \right) \right) $ is called \emph{generalized Lie
algebroid structure.}

\begin{quotation}
{\small So, we extend the notion of Lie algebroid from one base manifold to
a pair of diffeomorphic base manifolds and we obtain the notion of
generalized Lie algebroid via Ehresmann connections theory, independent of
the pevious generalizations of Lie algebroids appeared in literature. In the
following particular case, }$\left( \eta ,h\right) =\left(
Id_{M},Id_{M}\right) ,${\small \ we obtain the definition of Lie algebroid.
(see also }\cite{1}{\small )}
\end{quotation}

We can define the set of morphisms of 
\begin{equation*}
\left( \left( F,\nu ,N\right) ,\left[ ,\right] _{F,h},\left( \rho ,\eta
\right) \right)
\end{equation*}%
source and 
\begin{equation*}
\left( \left( F^{\prime },\nu ^{\prime },N^{\prime }\right) ,\left[ ,\right]
_{F^{\prime },h^{\prime }},\left( \rho ^{\prime },\eta ^{\prime }\right)
\right)
\end{equation*}%
target as being the set 
\begin{equation*}
\begin{array}{c}
\left\{ \left( \varphi ,\varphi _{0}\right) \in \mathbf{B}^{\mathbf{v}%
}\left( \left( F,\nu ,N\right) ,\left( F^{\prime },\nu ^{\prime },N^{\prime
}\right) \right) \right\}%
\end{array}%
\end{equation*}%
such that $\varphi _{0}\in Iso_{\mathbf{Man}}\left( N,N^{\prime }\right) $
and the $\mathbf{Mod}$-morphism $\Gamma \left( \varphi ,\varphi _{0}\right) $
is a $\mathbf{LieAlg}$-morphism of 
\begin{equation*}
\left( \Gamma \left( F,\nu ,N\right) ,+,\cdot ,\left[ ,\right] _{F,h}\right)
\end{equation*}%
source and 
\begin{equation*}
\left( \Gamma \left( F^{\prime },\nu ^{\prime },N^{\prime }\right) ,+,\cdot
, \left[ ,\right] _{F^{\prime },h^{\prime }}\right)
\end{equation*}%
target.

So, we can discuss about \emph{the category }$\mathbf{GLA}$\emph{\ of
generalized Lie algebroids.} We remark that $\mathbf{GLA}$ is a subcategory
of the category $\mathbf{B}^{\mathbf{v}}.$

\begin{example}
\label{e1}\textrm{Let }$M,N\in \left\vert \mathbf{Man}\right\vert ,$\textrm{%
\ }$h\in Iso_{\mathbf{Man}}\left( M,N\right) $\textrm{\ and }$\eta \in Iso_{%
\mathbf{Man}}\left( N,M\right) $\textrm{\ be. Using the tangent }$\mathbf{B}%
^{\mathbf{v}}$\textrm{-morphism }$\left( T\eta ,\eta \right) $\textrm{\ and
the operation }%
\begin{equation*}
\begin{array}{ccc}
\Gamma \left( TN,\tau _{N},N\right) \times \Gamma \left( TN,\tau
_{N},N\right) & ^{\underrightarrow{~\ \ \left[ ,\right] _{TN,h}~\ \ }} & 
\Gamma \left( TN,\tau _{N},N\right) \\ 
\left( u,v\right) & \longmapsto & \ \left[ u,v\right] _{TN,h}%
\end{array}%
\end{equation*}%
\textrm{where }%
\begin{equation*}
\left[ u,v\right] _{TN,h}=\Gamma \left( T\left( h\circ \eta \right)
^{-1},\left( h\circ \eta \right) ^{-1}\right) \left( \left[ \Gamma \left(
T\left( h\circ \eta \right) ,h\circ \eta \right) u,\Gamma \left( T\left(
h\circ \eta \right) ,h\circ \eta \right) v\right] _{TN}\right) ,
\end{equation*}%
\textrm{for any }$u,v\in \Gamma \left( TN,\tau _{N},N\right) $\textrm{, we
obtain that}%
\begin{equation*}
\begin{array}{c}
\left( \left( TN,\tau _{N},N\right) ,\left( T\eta ,\eta \right) ,\left[ ,%
\right] _{TN,h}\right)%
\end{array}%
\end{equation*}%
\textrm{is a generalized Lie algebroid.} (see \cite{1})
\end{example}

For any $\mathbf{Man}$-isomorphisms $\eta $ and $h,$ new and interesting
generalized Lie algebroid structures for the tangent vector bundle $\left(
TN,\tau _{N},N\right) $ are obtained$.$ For any base $\left\{ t_{\alpha
},~\alpha \in \overline{1,m}\right\} $ of the module of sections $\left(
\Gamma \left( TN,\tau _{N},N\right) ,+,\cdot \right) $ we obtain the
structure functions%
\begin{equation*}
\begin{array}{c}
L_{\alpha \beta }^{\gamma }=\left( \theta _{\alpha }^{i}\frac{\partial
\theta _{\beta }^{j}}{\partial x^{i}}-\theta _{\beta }^{i}\frac{\partial
\theta _{\alpha }^{j}}{\partial x^{i}}\right) \tilde{\theta}_{j}^{\gamma
},~\alpha ,\beta ,\gamma \in \overline{1,m}%
\end{array}%
\end{equation*}%
where 
\begin{equation*}
\theta _{\alpha }^{i},~i,\alpha \in \overline{1,m}
\end{equation*}%
are real local functions so that 
\begin{equation*}
\begin{array}{c}
\Gamma \left( T\left( h\circ \eta \right) ,h\circ \eta \right) \left(
t_{\alpha }\right) =\theta _{\alpha }^{i}\frac{\partial }{\partial x^{i}}%
\end{array}%
\end{equation*}%
and 
\begin{equation*}
\tilde{\theta}_{j}^{\gamma },~i,\gamma \in \overline{1,m}
\end{equation*}%
are real local functions so that 
\begin{equation*}
\begin{array}{c}
\Gamma \left( T\left( h\circ \eta \right) ^{-1},\left( h\circ \eta \right)
^{-1}\right) \left( \frac{\partial }{\partial x^{j}}\right) =\tilde{\theta}%
_{j}^{\gamma }t_{\gamma }.%
\end{array}%
\end{equation*}

In particular, using arbitrary isometries (symmetries, translations,
rotations,...) for the Euclidean $3$-dimensional space $\Sigma ,$ and
arbitrary basis for the module of sections we obtain a lot of generalized
Lie algebroid structures for the tangent vector bundle $\left( T\Sigma ,\tau
_{\Sigma },\Sigma \right) $.

The problem to develop a Lagrangian formalism directly on a Lie algebroid
similar to Klein's formalism for ordinary Lagrangian Mechanics \cite{19} was
proposed by Alan Weinstein in his work \cite{40}.

In that work, the author gave a theory of Lagrangian systems on Lie
algebroids and obtained the Euler-Lagrange equations using the dual of a Lie
algebroid and the Legendre transformation defined by a regular Lagrangian.

Paulette Liberman \cite{22}\ showed later that such a formalism is not
possible if one considers the tangent bundle of a Lie algebroid as space for
developing the theory. Eduardo Martinez \cite{25}\ gave a full description
using the notion of prolongation of a Lie algebroid.

Recently, there has been a lot of activity around mechanical systems on Lie
algebroids (due to its unifying view many different problems) and some
extensions to more general structures \cite{7,8,21,26,33,34,38,39}.

Recently, a unified approach of constraining implicit Lagrangian and
Hamiltonian systems on \emph{Dirac algebroid} was presented in \cite{17}. A
Dirac algebroid on a vector bundle $\left( E,\pi ,M\right) $ has been viewed
as a double vector bundle morphism 
\begin{equation*}
\varepsilon :T^{\ast }E~\ ^{\underrightarrow{~\ \ \ \ }}~\ T\overset{\ast }{E%
}
\end{equation*}%
covering the identity on $\left( \overset{\ast }{E},\overset{\ast }{\pi }%
,M\right) .$ (see \cite{15,16})

The double vector bundles presented in \cite{20,31,32} has been applied to
geometric formalism of Analytical Mechanics, including nonholonomic
constrains in \cite{13,14}.

There are obtained $\varepsilon $ as the composition of the canonical
isomorphism of double vector bundles 
\begin{equation*}
R_{\tau }:T^{\ast }E~\ ^{\underrightarrow{~\ \ \ \ }}~\ T^{\ast }\overset{%
\ast }{E}
\end{equation*}%
and 
\begin{equation*}
\widetilde{\Pi }_{\varepsilon }:T^{\ast }\overset{\ast }{E}~\ ^{%
\underrightarrow{~\ \ \ \ }}~\ T\overset{\ast }{E}.
\end{equation*}

An application of this aproach to Analytical Mechanics so that $\left( E,\pi
,M\right) $ plays the role of kinematic configurations, is based on some
ideas of Tulkzyjew presented in \cite{36,37}.

We know that the use of connections for the geometry of systems of second
order differential equations has been proposed by M. Crampin \cite{10} and
C. Griffone \cite{18}.

Using the generalized Lie algebroids, a new class of (linear) connections in
Ehresmann sense on the base of the \emph{Lie algebroid generalized tangent
bundle }is presented in Section $2$\emph{. }(see also \cite{1}) Using this $%
\left( \rho ,\eta \right) $-connection theory, the\emph{\ adapted basis} are
presented$.$

The lift $%
\begin{array}{ccc}
I & ^{\underrightarrow{~\ \dot{c}~\ }} & E%
\end{array}%
$of a curve $%
\begin{array}{ccc}
I & ^{\underrightarrow{~\ c~\ }} & M%
\end{array}%
$ is presented and studied in the Section $3.$ Section $4$ studies
remarkable endomorphisms of the Lie algebra of sections of the Lie algebroid
generalized tangent bundle and in Section $5$ a new class of linear
connections, called by use \emph{distinguished linear }$\left( \rho ,\eta
\right) $\emph{-connections, }is presented.

We know that the geometry of Lagrange space is the geometry of its canonical
semispray and the associated nonlinear connection. It has been developed by
R. Miron and M. Anastasiei \cite{28}$.$ Using techniques that are specific
to Lagrange geometry, R. Miron \cite{27} introduces and investigates some
geometric aspects of nonconservative mechanical systems by means of the
corresponding semispray and nonlinear connection.

When the external force field depends on both position and velocity, the
geometry of nonconservative mechanical systems was rigorously ivestigated by
Klein \cite{19} and Godbillon \cite{12}. In paper \cite{19} Klein introduces
a second rank skew symmetric tensor as external force tensor of a
nonconservative mechanical system.

Some aspects regarding first integrals for nonconservative mechanical
systems were investigated by Djukic and Vijanovic \cite{11} and Cantrijn 
\cite{9}. I. Buc\u{a}taru and R. Miron \cite{6} extend the geometric
investigations of nonconservative mechanical systems using the evolution
nonlinear connection and the almost syplectic structure of the
nonconservative mechanical system.

Finding a new solution for the Weinstein's Problem in the general framework
of generalized Lie algebroids, a new class of \emph{mechanical systems}
called by use, \emph{mechanical }$\left( \rho ,\eta \right) $\emph{-systems,
Lagrange mechanical }$\left( \rho ,\eta \right) $\emph{-systems} or \emph{%
Finsler mechanical }$\left( \rho ,\eta \right) $\emph{-systems} is presented
in Section $6$. In Section $7$ we study the \emph{\ canonical }$\left( \rho
,\eta \right) $\emph{-semispray associated to a mechanical }$\left( \rho
,\eta \right) $\emph{-system.} Finally, an important result about the \emph{%
canonical }$\left( \rho ,\eta \right) $\emph{-spray associated to a
mechanical }$\left( \rho ,\eta \right) $\emph{-system }is presented\emph{. }%
In particular, we obtain similar results with I. Buc\u{a}taru and R. Miron 
\cite{5,6}.

The Section $8$ of this paper is dedicated to study the geometry of \emph{%
Lagrange mechanical }$\left( \rho ,\eta \right) $\emph{-systems.} We
determine and we study the $\left( \rho ,\eta \right) $\emph{-semispray
associated to a regular Lagrangian }$L$\emph{\ and external force} $F_{e}$
which are applied on the total space of a generalized Lie algebroid.

The equations of Euler-Lagrange type are derived. In particular, using the
Lie algebroid generalized tangent bundle of a Lie algebroid, we obtain a new
solution for the \textbf{Weinstein's Problem, }different by the Mart\'{\i}%
nez's solution\textbf{\ }\cite{25}. (see also \cite{3,4})

Finally, we obtain that the\ integral curves of the canonical $\left( \rho
,\eta \right) $-semispray associated to Lagrange mechanical $\left( \rho
,\eta \right) $-system $\left( \left( E,\pi ,M\right) ,F_{e},L\right) $\ and
from locally invertible\emph{\ }$\mathbf{B}^{\mathbf{v}}$-morphism $\left(
g,h\right) $\ are the $\left( g,h\right) $-lifts solutions of the \emph{%
equations of Euler-Lagrange type }\eqref{eq107}.

In particular, If $F$\ is a Finsler fundamental function, then the\
geodesics on the manifold $M$\ are the curves such that the components of
their $\left( g,h\right) $-lifts are\ solutions for \emph{equations of
Euler-Lagrange type} \eqref{eq107}.

Notice that our theory of \emph{mechanical systems} is a progress because,
in particular, we obtain all previous mechanical systems presented in
literature.

As there exists (see: \cite{1}) a construction of the generalized tangent
bundle of a dual vector bundle and a connection theory in the dual case, we
ask:

\begin{quotation}
{\small - Can we develop an alternative approach for the dual case and can
we find a transformation of Legendre type with whom to show the equivalence
between these Mechanical Systems?}
\end{quotation}

The answer, in the next paper.

\section{$\left( \protect\rho ,\protect\eta \right) $-connections and
adapted basis\label{sfive}}

\ \ 

We consider the diagram:%
\begin{equation}
\begin{array}{rcl}
E &  & \left( F,\left[ ,\right] _{F,h},\left( \rho ,\eta \right) \right) \\ 
\pi \downarrow &  & ~\downarrow \nu \\ 
M & ^{\underrightarrow{~\ \ \ \ h~\ \ \ \ }} & ~\ N%
\end{array}
\label{eq8}
\end{equation}%
where $\left( E,\pi ,M\right) \in \left\vert \mathbf{B}^{\mathbf{v}%
}\right\vert $ and $\left( \left( F,\nu ,N\right) ,\left[ ,\right]
_{F,h},\left( \rho ,\eta \right) \right) \in \left\vert \mathbf{GLA}%
\right\vert $.

We obtain the $\mathbf{B}^{\mathbf{v}}$-morphism%
\begin{equation}
\begin{array}{ccc}
~\ \ \ \ \ \ \ \ \ \ \ \ \ \pi ^{\ast }\ \left( h^{\ast }F\right) & 
\hookrightarrow & F \\ 
\pi ^{\ast }\left( h^{\ast }\nu \right) \downarrow &  & ~\downarrow \nu \\ 
~\ \ \ \ \ \ \ \ \ \ \ \ \ \ E & ^{\underrightarrow{~\ \ h\circ \pi ~\ \ }}
& M%
\end{array}
\label{eq9}
\end{equation}

We take $\left( x^{i},y^{a}\right) $ as canonical local coordinates on $%
\left( E,\pi ,M\right) ,$ where $i\in \overline{1,m}$ and $a\in \overline{1,r%
}.$ Let 
\begin{equation*}
\left( x^{i},y^{a}\right) \longrightarrow \left( x^{i%
{\acute{}}%
}\left( x^{i}\right) ,y^{a%
{\acute{}}%
}\left( x^{i},y^{a}\right) \right)
\end{equation*}%
be a change of coordinates on $\left( E,\pi ,M\right) $. Then the
coordinates $y^{a}$ change to $y^{a%
{\acute{}}%
}$ according to the rule:%
\begin{equation}
\begin{array}{c}
y^{a%
{\acute{}}%
}=\displaystyle M_{a}^{a%
{\acute{}}%
}y^{a}.%
\end{array}
\label{eq10}
\end{equation}

\begin{theorem}
\label{t2}\emph{Let} $\Big({\overset{\pi ^{\ast }\ \left( h^{\ast }F\right) }%
{\rho }},Id_{E}\Big)$ \emph{be the }$\mathbf{B}^{\mathbf{v}}$\emph{-morphism
of }$\left( \pi ^{\ast }\ \left( h^{\ast }F\right) ,\pi ^{\ast }\left(
h^{\ast }\nu \right) ,E\right) $\ \emph{source and} $\left( TE,\tau
_{E},E\right) $\ \emph{target, where}%
\begin{equation}
\begin{array}{rcl}
\ \pi ^{\ast }\ \left( h^{\ast }F\right) & ^{\underrightarrow{\overset{\pi
^{\ast }\ \left( h^{\ast }F\right) }{\rho }}} & TE \\ 
\displaystyle Z^{\alpha }T_{\alpha }\left( u_{x}\right) & \longmapsto & %
\displaystyle\left( Z^{\alpha }\cdot \rho _{\alpha }^{i}\circ h\circ \pi
\right) \frac{\partial }{\partial x^{i}}\left( u_{x}\right)%
\end{array}
\label{eq11}
\end{equation}%
\emph{Using the operation} 
\begin{equation*}
\begin{array}{ccc}
\Gamma \left( \pi ^{\ast }\ \left( h^{\ast }F\right) ,\pi ^{\ast }\left(
h^{\ast }\nu \right) ,E\right) ^{2} & ^{\underrightarrow{~\ \ \left[ ,\right]
_{\pi ^{\ast }\ \left( h^{\ast }F\right) }~\ \ }} & \Gamma \left( \pi ^{\ast
}\ \left( h^{\ast }F\right) ,\pi ^{\ast }\left( h^{\ast }\nu \right)
,E\right)%
\end{array}%
\end{equation*}%
\emph{defined by}%
\begin{equation}
\begin{array}{ll}
\left[ T_{\alpha },T_{\beta }\right] _{\pi ^{\ast }\ \left( h^{\ast
}F\right) } & =L_{\alpha \beta }^{\gamma }\circ h\circ \pi \cdot T_{\gamma },%
\vspace*{1mm} \\ 
\left[ T_{\alpha },fT_{\beta }\right] _{\pi ^{\ast }\ \left( h^{\ast
}F\right) } & \displaystyle=fL_{\alpha \beta }^{\gamma }\circ h\circ \pi
T_{\gamma }+\rho _{\alpha }^{i}\circ h\circ \pi \frac{\partial f}{\partial
x^{i}}T_{\beta },\vspace*{1mm} \\ 
\left[ fT_{\alpha },T_{\beta }\right] _{\pi ^{\ast }\ \left( h^{\ast
}F\right) } & =-\left[ T_{\beta },fT_{\alpha }\right] _{\pi ^{\ast }\ \left(
h^{\ast }F\right) },%
\end{array}
\label{eq12}
\end{equation}%
\emph{for any} $f\in \mathcal{F}\left( E\right) ,$ \emph{it results that} 
\begin{equation*}
\begin{array}{c}
\left( \left( \pi ^{\ast }\ \left( h^{\ast }F\right) ,\pi ^{\ast }\left(
h^{\ast }\nu \right) ,E\right) ,\left[ ,\right] _{\pi ^{\ast }\ \left(
h^{\ast }F\right) },\left( \overset{\pi ^{\ast }\ \left( h^{\ast }F\right) }{%
\rho },Id_{E}\right) \right)%
\end{array}%
\end{equation*}%
\emph{is a Lie algebroid which is called the pull-back Lie algebroid of the
generalized Lie algebroid }$\left( \left( F,\nu ,N\right) ,\left[ ,\right]
_{F,h},\left( \rho ,\eta \right) \right) .$
\end{theorem}

If $z=z^{\alpha }t_{\alpha }\in \Gamma \left( F,\nu ,N\right) ,$ then we
obtain the section%
\begin{equation*}
Z=\left( z^{\alpha }\circ h\circ \pi \right) T_{\alpha }\in \Gamma \left(
\pi ^{\ast }\left( h^{\ast }F\right) ,\pi ^{\ast }\left( h^{\ast }\nu
\right) ,E\right)
\end{equation*}%
so that $Z\left( u_{x}\right) =z\left( h\left( x\right) \right) ,$ for any $%
u_{x}\in \pi ^{-1}\left( U{\cap h}^{-1}V\right) .$

Let 
\begin{equation*}
\begin{array}[b]{c}
\left( \partial _{i},\dot{\partial}_{a}\right) \overset{put}{=}\left( \frac{%
\partial }{\partial x^{i}},\frac{\partial }{\partial y^{a}}\right)%
\end{array}%
\end{equation*}
be the base sections for the Lie $\mathcal{F}\left( E\right) $-algebra 
\begin{equation*}
\left( \Gamma \left( TE,\tau _{E},E\right) ,+,\cdot ,\left[ ,\right]
_{TE}\right) .
\end{equation*}

For any sections%
\begin{equation*}
\begin{array}{c}
Z^{\alpha }T_{\alpha }\in \Gamma \left( \pi ^{\ast }\left( h^{\ast }F\right)
,\pi ^{\ast }\left( h^{\ast }F\right) ,E\right)%
\end{array}%
\end{equation*}%
and%
\begin{equation*}
\begin{array}{c}
Y^{a}\dot{\partial}_{a}\in \Gamma \left( VTE,\tau _{E},E\right)%
\end{array}%
\end{equation*}%
we obtain the section 
\begin{equation*}
\begin{array}{c}
Z^{\alpha }\tilde{\partial}_{\alpha }+Y^{a}\overset{\cdot }{\tilde{\partial}}%
_{a}=:Z^{\alpha }\left( T_{\alpha }\oplus \left( \rho _{\alpha }^{i}\circ
h\circ \pi \right) \partial _{i}\right) +Y^{a}\left( 0_{\pi ^{\ast }\left(
h^{\ast }F\right) }\oplus \dot{\partial}_{a}\right) \vspace*{1mm} \\ 
=Z^{\alpha }T_{\alpha }\oplus \left( Z^{\alpha }\left( \rho _{\alpha
}^{i}\circ h\circ \pi \right) \partial _{i}+Y^{a}\dot{\partial}_{a}\right)
\in \Gamma \left( \pi ^{\ast }\left( h^{\ast }F\right) \oplus TE,\overset{%
\oplus }{\pi },E\right) .%
\end{array}%
\end{equation*}

Since we have 
\begin{equation*}
\begin{array}{c}
Z^{\alpha }\displaystyle\tilde{\partial}_{\alpha }+Y^{a}\overset{\cdot }{%
\tilde{\partial}}_{a}=0 \\ 
\Updownarrow \\ 
Z^{\alpha }T_{\alpha }=0~\wedge Z^{\alpha }\left( \rho _{\alpha }^{i}\circ
h\circ \pi \right) \partial _{i}+Y^{a}\dot{\partial}_{a}=0,%
\end{array}%
\end{equation*}%
it implies $Z^{\alpha }=0,~\alpha \in \overline{1,p}$ and $Y^{a}=0,~a\in 
\overline{1,r}.$

Therefore, the sections $\tilde{\partial}_{1},...,\tilde{\partial}_{p},%
\overset{\cdot }{\tilde{\partial}}_{1},...,\overset{\cdot }{\tilde{\partial}}%
_{r}$ are linearly independent.\smallskip

We consider the vector subbundle $\left( \left( \rho ,\eta \right) TE,\left(
\rho ,\eta \right) \tau _{E},E\right) $ of the vector bundle\break $\left(
\pi ^{\ast }\left( h^{\ast }F\right) \oplus TE,\overset{\oplus }{\pi }%
,E\right) ,$ for which the $\mathcal{F}\left( E\right) $-module of sections
is the $\mathcal{F}\left( E\right) $-submodule of $\left( \Gamma \left( \pi
^{\ast }\left( h^{\ast }F\right) \oplus TE,\overset{\oplus }{\pi },E\right)
,+,\cdot \right) ,$ generated by the set of sections $\left( \tilde{\partial}%
_{\alpha },\overset{\cdot }{\tilde{\partial}}_{a}\right) .$

The base sections $\left( \tilde{\partial}_{\alpha },\overset{\cdot }{\tilde{%
\partial}}_{a}\right) $ will be called the \emph{natural }$\left( \rho ,\eta
\right) $\emph{-base.}

The matrix of coordinate transformation on $\left( \left( \rho ,\eta \right)
TE,\left( \rho ,\eta \right) \tau _{E},E\right) $ at a change of fibred
charts is%
\begin{equation}
\left\Vert 
\begin{array}{cc}
\Lambda _{\alpha }^{\alpha 
{\acute{}}%
}\circ h\circ \pi & 0\vspace*{1mm} \\ 
\left( \rho _{a}^{i}\circ h\circ \pi \right) \displaystyle\frac{\partial
M_{b}^{a%
{\acute{}}%
}\circ \pi }{\partial x_{i}}y^{b} & M_{a}^{a%
{\acute{}}%
}\circ \pi%
\end{array}%
\right\Vert .  \label{eq13}
\end{equation}

Easily we obtain:

\begin{theorem}
\label{t3}\emph{Let} $\left( \tilde{\rho},Id_{E}\right) $\ \emph{be the} $%
\mathbf{B}^{\mathbf{v}}$\emph{-morphism of }$\left( \left( \rho ,\eta
\right) TE,\left( \rho ,\eta \right) \tau _{E},E\right) $\ \emph{source and }%
$\left( TE,\tau _{E},E\right) $\ \emph{target, where}%
\begin{equation}
\begin{array}{rcl}
\left( \rho ,\eta \right) TE\!\!\! & \!\!^{\underrightarrow{\tilde{\ \ \rho
\ \ }}}\!\!\! & \!\!TE\vspace*{2mm} \\ 
\left( Z^{\alpha }\tilde{\partial}_{\alpha }+Y^{a}\overset{\cdot }{\tilde{%
\partial}}_{a}\right) \!(u_{x})\!\!\!\! & \!\!\longmapsto \!\!\! & 
\!\!\left( \!Z^{\alpha }\!\left( \rho _{\alpha }^{i}{\circ }h{\circ }\pi
\!\right) \!\partial _{i}{+}Y^{a}\dot{\partial}_{a}\right) \!(u_{x})\!\!.%
\end{array}%
.  \label{eq14}
\end{equation}%
\emph{Using the operation} 
\begin{equation*}
\begin{array}{ccc}
\Gamma \left( \left( \rho ,\eta \right) TE,\left( \rho ,\eta \right) \tau
_{E},E\right) ^{2} & ^{\underrightarrow{~\ \ \left[ ,\right] _{\left( \rho
,\eta \right) TE}~\ \ }} & \Gamma \left( \left( \rho ,\eta \right) TE,\left(
\rho ,\eta \right) \tau _{E},E\right) 
\end{array}%
\end{equation*}%
\emph{defined by}%
\begin{equation}
\begin{array}{l}
\left[ Z_{1}^{\alpha }\tilde{\partial}_{\alpha }+Y_{1}^{a}\overset{\cdot }{%
\tilde{\partial}}_{a},Z_{2}^{\beta }\tilde{\partial}_{\beta }+Y_{2}^{b}%
\overset{\cdot }{\tilde{\partial}}_{b}\right] _{\left( \rho ,\eta \right) TE}%
\vspace*{1mm} \\ 
\displaystyle=\left[ Z_{1}^{\alpha }T_{\alpha },Z_{2}^{\beta }T_{\beta }%
\right] _{\pi ^{\ast }\left( h^{\ast }F\right) }\oplus \left[ Z_{1}^{\alpha
}\left( \rho _{\alpha }^{i}\circ h\circ \pi \right) \partial _{i}+Y_{1}^{a}%
\dot{\partial}_{a},\right. \vspace*{1mm} \\ 
\hfill \displaystyle\left. Z_{2}^{\beta }\left( \rho _{\beta }^{j}\circ
h\circ \pi \right) \partial _{j}+Y_{2}^{b}\dot{\partial}_{b}\right] _{TE},%
\end{array}%
.  \label{eq15}
\end{equation}%
\emph{for any} $Z_{1}^{\alpha }\tilde{\partial}_{\alpha }+Y_{1}^{a}\overset{%
\cdot }{\tilde{\partial}}_{a}$\emph{\ and }$Z_{2}^{\beta }\tilde{\partial}%
_{\beta }+Y_{2}^{b}\overset{\cdot }{\tilde{\partial}}_{b},$ \emph{we obtain
that the couple }%
\begin{equation*}
\left( \left[ ,\right] _{\left( \rho ,\eta \right) TE},\left( \tilde{\rho}%
,Id_{E}\right) \right) 
\end{equation*}%
\emph{\ is a Lie algebroid structure for the vector bundle }$\left( \left(
\rho ,\eta \right) TE,\left( \rho ,\eta \right) \tau _{E},E\right) .$
\end{theorem}

\begin{remark}
\label{r4}\textrm{In particular, if }$h=Id_{M},$\textrm{\ then the Lie
algebroid }%
\begin{equation*}
\begin{array}{c}
\left( \left( \left( Id_{TM},Id_{M}\right) TE,\left( Id_{TM},Id_{M}\right)
\tau _{E},E\right) ,\left[ ,\right] _{\left( Id_{TM},Id_{M}\right)
TE},\left( \widetilde{Id_{TM}},Id_{E}\right) \right)%
\end{array}%
\end{equation*}%
\textrm{is isomorphic with the usual Lie algebroid }%
\begin{equation*}
\begin{array}{c}
\left( \left( TE,\tau _{E},E\right) ,\left[ ,\right] _{TE},\left(
Id_{TE},Id_{E}\right) \right) .%
\end{array}%
\end{equation*}
\end{remark}

\textrm{This is a reason for which the Lie algebroid }%
\begin{equation*}
\begin{array}{c}
\left( \left( \left( \rho ,\eta \right) TE,\left( \rho ,\eta \right) \tau
_{E},E\right) ,\left[ ,\right] _{\left( \rho ,\eta \right) TE},\left( \tilde{%
\rho},Id_{E}\right) \right)%
\end{array}%
,
\end{equation*}%
\textrm{will be called the Lie algebroid generalized tangent bundle.}\emph{\ 
}(see also \cite{1})

We consider the $\mathbf{B}^{\mathbf{v}}$-morphism $\left( \left( \rho ,\eta
\right) \pi !,Id_{E}\right) $ given by the commutative diagram%
\begin{equation}
\begin{array}{rcl}
\left( \rho ,\eta \right) TE & ^{\underrightarrow{~\ \left( \rho ,\eta
\right) \pi !~\ }} & \pi ^{\ast }\left( h^{\ast }F\right) \\ 
\left( \rho ,\eta \right) \tau _{E}\downarrow ~ &  & ~\downarrow \pi ^{\ast
}\left( h^{\ast }\nu \right) \\ 
E~\  & ^{\underrightarrow{~Id_{E}~}} & ~\ E%
\end{array}
\label{eq16}
\end{equation}%
This is defined as:%
\begin{equation}
\begin{array}{c}
\left( \rho ,\eta \right) \pi !\left( \left( Z^{\alpha }\tilde{\partial}%
_{\alpha }+Y^{a}\overset{\cdot }{\tilde{\partial}}_{a}\right) \left(
u_{x}\right) \right) =\left( Z^{\alpha }T_{\alpha }\right) \left(
u_{x}\right) ,%
\end{array}
\label{eq17}
\end{equation}%
for any $Z^{\alpha }\tilde{\partial}_{\alpha }+Y^{a}\overset{\cdot }{\tilde{%
\partial}}_{a}\in \Gamma \left( \left( \rho ,\eta \right) TE,\left( \rho
,\eta \right) \tau _{E},E\right) .$\medskip

Using the $\mathbf{B}^{\mathbf{v}}$-morphism $\left( \left( \rho ,\eta
\right) \pi !,Id_{E}\right) $ we obtain the \emph{tangent }$\left( \rho
,\eta \right) $\emph{-application }$\left( \left( \rho ,\eta \right) T\pi
,h\circ \pi \right) $ of $\left( \left( \rho ,\eta \right) TE,\left( \rho
,\eta \right) \tau _{E},E\right) $ source and $\left( F,\nu ,N\right) $
target.

\begin{definition}
\label{d5}The kernel of the tangent $\left( \rho ,\eta \right) $%
-application\ is writen 
\begin{equation*}
\left( V\left( \rho ,\eta \right) TE,\left( \rho ,\eta \right) \tau
_{E},E\right)
\end{equation*}%
and it is called \emph{the vertical interior differential system}.\bigskip\
(see \cite{2})
\end{definition}

We remark that the set $\left\{ \overset{\cdot }{\tilde{\partial}}_{a},~a\in 
\overline{1,r}\right\} $ is a base of the $\mathcal{F}\left( E\right) $%
-module 
\begin{equation*}
\left( \Gamma \left( V\left( \rho ,\eta \right) TE,\left( \rho ,\eta \right)
\tau _{E},E\right) ,+,\cdot \right) .
\end{equation*}

\begin{proposition}
\label{p6}\emph{The short sequence of vector bundles}%
\begin{equation}
\begin{array}{ccccccccc}
0 & \hookrightarrow & V\left( \rho ,\eta \right) TE & \hookrightarrow & 
\left( \rho ,\eta \right) TE & ^{\underrightarrow{~\ \left( \rho ,\eta
\right) \pi !~\ }} & \pi ^{\ast }\left( h^{\ast }F\right) & ^{%
\underrightarrow{}} & 0 \\ 
\downarrow &  & \downarrow &  & \downarrow &  & \downarrow &  & \downarrow
\\ 
E & ^{\underrightarrow{~Id_{E}~}} & E & ^{\underrightarrow{~Id_{E}~}} & E & 
^{\underrightarrow{~Id_{E}~}} & E & ^{\underrightarrow{~Id_{E}~}} & E%
\end{array}
\label{eq18}
\end{equation}%
\emph{is exact.}
\end{proposition}

\begin{definition}
\label{d7}\textit{A }$\mathbf{Man}$-morphism $\left( \rho ,\eta \right)
\Gamma $ of $\left( \rho ,\eta \right) TE$ source and $V\left( \rho ,\eta
\right) TE$ target defined by%
\begin{equation}
\begin{array}{c}
\left( \rho ,\eta \right) \Gamma \left( Z^{\gamma }\tilde{\partial}_{\gamma
}+Y^{a}\overset{\cdot }{\tilde{\partial}}_{a}\right) \left( u_{x}\right)
=\left( Y^{a}+\left( \rho ,\eta \right) \Gamma _{\gamma }^{a}Z^{\gamma
}\right) \overset{\cdot }{\tilde{\partial}}_{a}\left( u_{x}\right) ,%
\end{array}
\label{eq19}
\end{equation}%
so that the $\mathbf{B}^{\mathbf{v}}$-morphism $\left( \left( \rho ,\eta
\right) \Gamma ,Id_{E}\right) $ is a split to the left in the previous exact
sequence, will be called $\left( \rho ,\eta \right) $\emph{-connection for
the vector bundle }$\left( E,\pi ,M\right) $.
\end{definition}

The $\left( \rho ,Id_{M}\right) $-connection is called $\rho $\emph{%
-connection }and is denoted $\rho \Gamma $\emph{\ }and the $\left(
Id_{TM},Id_{M}\right) $-connection is called \emph{connection }and is
denoted $\Gamma $\emph{.}

\begin{definition}
\label{d8}If $\left( \rho ,\eta \right) \Gamma $ is a $\left( \rho ,\eta
\right) $-connection for the vector bundle $\left( E,\pi ,M\right) $, then
the kernel of the $\mathbf{B}^{\mathbf{v}}$-morphism $\left( \left( \rho
,\eta \right) \Gamma ,Id_{E}\right) $\ is written 
\begin{equation*}
\left( H\left( \rho ,\eta \right) TE,\left( \rho ,\eta \right) \tau
_{E},E\right) 
\end{equation*}
and is called the \emph{horizontal interior differential system}. (see \cite%
{2})
\end{definition}

\begin{definition}
\label{d9}If $\left( E,\pi ,M\right) \in \left\vert \mathbf{B}^{\mathbf{v}%
}\right\vert $ and $\left\{ s_{a},a\in \overline{1,r}\right\} $ is a base of
the $\mathcal{F}\left( M\right) $-module of sections $\left( \Gamma \left(
E,\pi ,M\right) ,+,\cdot \right) $, then we obtain the $\mathbf{B}^{\mathbf{v%
}}$-morphism $\left( \Pi ,\pi \right) $ defined by the commutative diagram%
\begin{equation}
\begin{array}{rcl}
V\left( \rho ,\eta \right) TE & ^{\underrightarrow{~\ \Pi ~\ }} & ~\ E \\ 
\left( \rho ,\eta \right) \tau _{E}\downarrow ~ &  & ~\downarrow \pi  \\ 
E~\  & ^{\underrightarrow{~~~\pi ~~}} & ~\ M%
\end{array}
\label{eq20}
\end{equation}%
so that%
\begin{equation}
\begin{array}{c}
\Pi \left( Y^{a}\overset{\cdot }{\tilde{\partial}}_{a}\left( u_{x}\right)
\right) =Y^{a}\left( u_{x}\right) s_{a}\left( x\right) .%
\end{array}
\label{eq21}
\end{equation}
\end{definition}

\begin{theorem}
\label{t10}(see \cite{1}) \emph{If }$\left( \rho ,\eta \right) \Gamma $\emph{%
\ is a }$\left( \rho ,\eta \right) $\emph{-connection for the vector bundle }%
$\left( E,\pi ,M\right) ,$\emph{\ then its components satisfy the law of
transformation}%
\begin{equation}
\begin{array}{c}
\left( \rho ,\eta \right) \Gamma _{\gamma 
{\acute{}}%
}^{a%
{\acute{}}%
}{=}M_{a}^{a%
{\acute{}}%
}{\circ }\pi \!\!\left[ \rho _{\gamma }^{k}{\circ }h{\circ }\pi \frac{%
\partial M_{b%
{\acute{}}%
}^{a}\circ \pi }{\partial x^{k}}y^{b%
{\acute{}}%
}{+}\left( \rho ,\eta \right) \!\Gamma _{\gamma }^{a}\right] \!\!\Lambda
_{\gamma 
{\acute{}}%
}^{\gamma }{\circ }h{\circ }\pi .%
\end{array}
\label{eq22}
\end{equation}%
\emph{\ }
\end{theorem}

\emph{In the particular case of Lie algebroids (see also \cite{3}), }$\left(
\eta ,h\right) =\left( Id_{M},Id_{M}\right) ,$\emph{\ the relations %
\eqref{eq22}\ become}%
\begin{equation}
\begin{array}{c}
\rho \Gamma _{\gamma 
{\acute{}}%
}^{a%
{\acute{}}%
}=M_{a}^{a%
{\acute{}}%
}\circ \pi \left[ \rho _{\gamma }^{k}\circ \pi \frac{\partial M_{b%
{\acute{}}%
}^{a}\circ \pi }{\partial x^{k}}y^{b%
{\acute{}}%
}+\rho \Gamma _{\gamma }^{a}\right] \Lambda _{\gamma 
{\acute{}}%
}^{\gamma }\circ \pi .%
\end{array}
\label{eq23}
\end{equation}%
\emph{In the classical case, }$\left( \rho ,\eta ,h\right) =\left(
Id_{TM},Id_{M},Id_{M}\right) ,$\emph{\ the relations \eqref{eq23}\ become}%
\begin{equation}
\begin{array}{c}
\Gamma _{k%
{\acute{}}%
}^{i%
{\acute{}}%
}=\frac{\partial x^{i%
{\acute{}}%
}}{\partial x^{i}}\circ \tau _{M}\left[ \frac{\partial }{\partial x^{k}}%
\left( \frac{\partial x^{i}}{\partial x^{j%
{\acute{}}%
}}\circ \tau _{M}\right) y^{j%
{\acute{}}%
}+\Gamma _{k}^{i}\right] \frac{\partial x^{k}}{\partial x^{k%
{\acute{}}%
}}\circ \tau _{M}.%
\end{array}
\label{eq24}
\end{equation}

\begin{remark}
\label{r11}\textrm{If we have a set of real local functions }$\left( \rho
,\eta \right) \Gamma _{\gamma }^{a}$\textrm{\ which satisfies the relations
of passing \emph{\eqref{eq22}},\ then we have a }$\left( \rho ,\eta \right) $%
\textrm{-connection }$\left( \rho ,\eta \right) \Gamma $\textrm{\ for the
vector bundle }$\left( E,\pi ,M\right) .$
\end{remark}

\begin{example}
\label{e12}\textrm{If }$\Gamma $\textrm{\ is an Ehresmann connection for the
vector bundle }$\left( E,\pi ,M\right) $\textrm{\ on components }$\Gamma
_{k}^{a},$\textrm{\ then the differentiable real local functions }%
\begin{equation*}
\left( \rho ,\eta \right) \Gamma _{\gamma }^{a}=\left( \rho _{\gamma
}^{k}\circ h\circ \pi \right) \Gamma _{k}^{a}
\end{equation*}%
\textrm{\ are the components of a }$\left( \rho ,\eta \right) $\textrm{%
-connection }$\left( \rho ,\eta \right) \Gamma $\textrm{\ for the vector
bundle }$\left( E,\pi ,M\right) .$\textrm{\ This }$\left( \rho ,\eta \right) 
$\textrm{-connection will be called the }$\left( \rho ,\eta \right) $\textrm{%
-connection associated to the connection }$\Gamma .$
\end{example}

We put the problem of finding a base for the $\mathcal{F}\left( E\right) $%
-module 
\begin{equation*}
\left( \Gamma \left( H\left( \rho ,\eta \right) TE,\left( \rho ,\eta \right)
\tau _{E},E\right) ,+,\cdot \right)
\end{equation*}%
of the type\textbf{\ } 
\begin{equation*}
\begin{array}[t]{l}
\tilde{\delta}_{\alpha }=Z_{\alpha }^{\beta }\tilde{\partial}_{\beta
}+Y_{\alpha }^{a}\overset{\cdot }{\tilde{\partial}}_{a},\alpha \in \overline{%
1,r}%
\end{array}%
\end{equation*}%
which satisfies the following conditions:%
\begin{equation}
\begin{array}{rcl}
\displaystyle\Gamma \left( \left( \rho ,\eta \right) \pi !,Id_{E}\right)
\left( \tilde{\delta}_{\alpha }\right) & = & T_{\alpha }\vspace*{2mm}, \\ 
\displaystyle\Gamma \left( \left( \rho ,\eta \right) \Gamma ,Id_{E}\right)
\left( \tilde{\delta}_{\alpha }\right) & = & 0.%
\end{array}
\label{eq25}
\end{equation}

Then we obtain the sections%
\begin{equation}
\begin{array}[t]{l}
\frac{\delta }{\delta \tilde{z}^{\alpha }}=\tilde{\partial}_{\alpha }-\left(
\rho ,\eta \right) \Gamma _{\alpha }^{a}\overset{\cdot }{\tilde{\partial}}%
_{a}=T_{\alpha }\oplus \left( \left( \rho _{\alpha }^{i}\circ h\circ \pi
\right) \partial _{i}-\left( \rho ,\eta \right) \Gamma _{\alpha }^{a}\dot{%
\partial}_{a}\right) .%
\end{array}
\label{eq26}
\end{equation}%
such that their law of change is a tensorial law under a change of vector
fiber charts.

The base $\left( \tilde{\delta}_{\alpha },\overset{\cdot }{\tilde{\partial}}%
_{a}\right) $ will be called the \emph{adapted }$\left( \rho ,\eta \right) $%
\emph{-base.}

\begin{remark}
\label{r13}\textrm{The following equality holds good}%
\begin{equation}
\begin{array}{l}
\Gamma \left( \tilde{\rho},Id_{E}\right) \left( \tilde{\delta}_{\alpha
}\right) =\left( \rho _{\alpha }^{i}\circ h\circ \pi \right) \partial
_{i}-\left( \rho ,\eta \right) \Gamma _{\alpha }^{a}\dot{\partial}_{a}.%
\end{array}
\label{eq27}
\end{equation}
\end{remark}

\textrm{Moreover, if }$\left( \rho ,\eta \right) \Gamma $\textrm{\ is the }$%
\left( \rho ,\eta \right) $\textrm{-connection associated to a connection }$%
\Gamma $\textrm{, then we obtain}%
\begin{equation}
\begin{array}{l}
\Gamma \left( \tilde{\rho},Id_{E}\right) \left( \tilde{\delta}_{\alpha
}\right) =\left( \rho _{\alpha }^{i}\circ h\circ \pi \right) \delta _{i},%
\end{array}
\label{eq28}
\end{equation}%
\textrm{where }$\left( \delta _{i},\dot{\partial}_{a}\right) $\textrm{\ is
the adapted base for the }$F\left( E\right) $\textrm{-module }$\left( \Gamma
\left( TE,\tau _{E},E\right) ,+,\cdot \right) .$

\begin{theorem}
\label{t14}\emph{The following equality holds good}%
\begin{equation}
\begin{array}{c}
\left[ \tilde{\delta}_{\alpha },\tilde{\delta}_{\beta }\right] _{\left( \rho
,\eta \right) TE}=L_{\alpha \beta }^{\gamma }\circ \left( h\circ \pi \right) 
\tilde{\delta}_{\gamma }+\left( \rho ,\eta ,h\right) \mathbb{R}_{\,\ \alpha
\beta }^{a}\overset{\cdot }{\tilde{\partial}}_{a},%
\end{array}
\label{eq29}
\end{equation}%
\emph{\ where}%
\begin{equation}
\begin{array}{l}
\left( \rho ,\eta ,h\right) \mathbb{R}_{\,\ \alpha \beta }^{a}=\Gamma \left( 
\tilde{\rho},Id_{E}\right) \left( \tilde{\delta}_{\beta }\right) \left(
\left( \rho ,\eta \right) \Gamma _{\alpha }^{a}\right) \vspace*{1mm} \\ 
\qquad -\Gamma \left( \tilde{\rho},Id_{E}\right) \left( \tilde{\delta}%
_{\alpha }\right) \left( \left( \rho ,\eta \right) \Gamma _{\beta
}^{a}\right) +\left( L_{\alpha \beta }^{\gamma }\circ h\circ \pi \right)
\left( \rho ,\eta \right) \Gamma _{\gamma }^{a},%
\end{array}
\label{eq30}
\end{equation}
\end{theorem}

\emph{Moreover, we have:}%
\begin{equation}
\begin{array}{c}
\left[ \tilde{\delta}_{\alpha },\overset{\cdot }{\tilde{\partial}}_{b}\right]
_{\left( \rho ,\eta \right) TE}=\Gamma \left( \tilde{\rho},Id_{E}\right)
\left( \overset{\cdot }{\tilde{\partial}}_{b}\right) \left( \left( \rho
,\eta \right) \Gamma _{\alpha }^{a}\right) \overset{\cdot }{\tilde{\partial}}%
_{a},%
\end{array}
\label{eq31}
\end{equation}%
\emph{and}%
\begin{equation}
\begin{array}{c}
\Gamma \left( \tilde{\rho},Id_{E}\right) \left[ \tilde{\delta}_{\alpha },%
\tilde{\delta}_{\beta }\right] _{\left( \rho ,\eta \right) TE}=\left[ \Gamma
\left( \tilde{\rho},Id_{E}\right) \left( \tilde{\delta}_{\alpha }\right)
,\Gamma \left( \tilde{\rho},Id_{E}\right) \left( \tilde{\delta}_{\beta
}\right) \right] _{TE}.%
\end{array}
\label{eq32}
\end{equation}

Let $\left( d\tilde{z}^{\alpha },d\tilde{y}^{b}\right) $ be the natural dual 
$\left( \rho ,\eta \right) $-base of natural $\left( \rho ,\eta \right) $%
-base $\left( \displaystyle\tilde{\partial}_{\alpha },\displaystyle\overset{%
\cdot }{\tilde{\partial}}_{a}\right) .$

This is determined by the equations 
\begin{equation*}
\begin{array}{c}
\left\{ 
\begin{array}{cc}
\displaystyle\left\langle d\tilde{z}^{\alpha },\tilde{\partial}_{\beta
}\right\rangle =\delta _{\beta }^{\alpha }, & \displaystyle\left\langle d%
\tilde{z}^{\alpha },\overset{\cdot }{\tilde{\partial}}_{a}\right\rangle =0,%
\vspace*{2mm} \\ 
\displaystyle\left\langle d\tilde{y}^{a},\tilde{\partial}_{\beta
}\right\rangle =0, & \displaystyle\left\langle d\tilde{y}^{a},\overset{\cdot 
}{\tilde{\partial}}_{b}\right\rangle =\delta _{b}^{a}.%
\end{array}%
\right.%
\end{array}%
\end{equation*}

We consider the problem of finding a base for the $\mathcal{F}\left(
E\right) $-module 
\begin{equation*}
\left( \Gamma \left( \left( V\left( \rho ,\eta \right) TE\right) ^{\ast
},\left( \left( \rho ,\eta \right) \tau _{E}\right) ^{\ast },E\right)
,+,\cdot \right)
\end{equation*}%
of the type 
\begin{equation*}
\begin{array}{c}
\delta \tilde{y}^{a}=\theta _{\alpha }^{a}d\tilde{z}^{\alpha }+\omega
_{b}^{a}d\tilde{y}^{b},~a\in \overline{1,n}%
\end{array}%
\end{equation*}%
which satisfies the following conditions:%
\begin{equation}
\begin{array}{c}
\left\langle \delta \tilde{y}^{a},\overset{\cdot }{\tilde{\partial}}%
_{a}\right\rangle =1\wedge \left\langle \delta \tilde{y}^{a},\tilde{\delta}%
_{\alpha }\right\rangle =0.%
\end{array}
\label{eq33}
\end{equation}

We obtain the sections%
\begin{equation}
\begin{array}{l}
\delta \tilde{y}^{a}=\left( \rho ,\eta \right) \Gamma _{\alpha }^{a}d\tilde{z%
}^{\alpha }+d\tilde{y}^{a},a\in \overline{1,n}.%
\end{array}
\label{eq34}
\end{equation}
such that their changing rule is tensorial under a change of vector fiber
charts. The base $\left( d\tilde{z}^{\alpha },\delta \tilde{y}^{a}\right) $
will be called the \emph{adapted dual }$\left( \rho ,\eta \right) $\emph{%
-base.}

\bigskip

\section{The $\left( g,h\right) $-lift of a differentiable curve}

\ \ 

We consider the following diagram:%
\begin{equation}
\begin{array}{rcl}
E &  & \left( F,\left[ ,\right] _{F,h},\left( \rho ,\eta \right) \right) \\ 
\pi \downarrow &  & ~\downarrow \nu \\ 
M & ^{\underrightarrow{~\ \ \ \ h~\ \ \ \ }} & ~\ N%
\end{array}
\label{eq35}
\end{equation}%
where $\left( E,\pi ,M\right) \in \left\vert \mathbf{B}^{\mathbf{v}%
}\right\vert $ and $\left( \left( F,\nu ,N\right) ,\left[ ,\right]
_{F,h},\left( \rho ,\eta \right) \right) $ is a generalized Lie algebroid.

We admit that $\left( \rho ,\eta \right) \Gamma $ is a $\left( \rho ,\eta
\right) $-connection for the vector bundle $\left( E,\pi ,M\right) $ and $%
\begin{array}[b]{ccc}
I & ^{\underrightarrow{~c\ }} & M%
\end{array}%
$ is a differentiable curve. We know that 
\begin{equation*}
\begin{array}{c}
\left( E_{|\func{Im}\left( \eta \circ h\circ c\right) },\pi _{|\func{Im}%
\left( \eta \circ h\circ c\right) },\func{Im}\left( \eta \circ h\circ
c\right) \right)%
\end{array}%
\end{equation*}%
is a vector subbundle of the vector bundle $\left( E,\pi ,M\right) .$

\begin{definition}
\label{d15}If%
\begin{equation}
\begin{array}{ccc}
I & ^{\underrightarrow{\ \ \dot{c}\ \ }} & E_{|\func{Im}\left( \eta \circ
h\circ c\right) } \\ 
t & \longmapsto & y^{a}\left( t\right) s_{a}\left( \eta \circ h\circ c\left(
t\right) \right)%
\end{array}
\label{eq36}
\end{equation}%
is a differentiable curve such that there exists $g\in \mathbf{Man}\left(
E,F\right) $ such that the following conditions are satisfied:
\end{definition}

\begin{itemize}
\item[1.] $\left( g,h\right) \in \mathbf{B}^{v}\left( \left( E,\pi ,M\right)
,\left( F,\nu ,N\right) \right) $ and

\item[2.] $\rho \circ g\circ \dot{c}\left( t\right) =\displaystyle\frac{%
d\left( \eta \circ h\circ c\right) ^{i}\left( t\right) }{dt}\frac{\partial }{%
\partial x^{i}}\left( \left( \eta \circ h\circ c\right) \left( t\right)
\right) ,$ for any $t\in I,$ \smallskip

then we will say that $\dot{c}$ \emph{\ is the }$\left( g,h\right) $\emph{%
-lift of the differentiable curve }$c.$
\end{itemize}

\begin{remark}
\label{r16}\textrm{The second condition is equivalent with the following
affirmation:}%
\begin{equation}
\begin{array}[b]{c}
\rho _{\alpha }^{i}\left( \eta \circ h\circ c\left( t\right) \right) \cdot
g_{a}^{\alpha }\left( h\circ c\left( t\right) \right) \cdot y^{a}\left(
t\right) =\frac{d\left( \eta \circ h\circ c\right) ^{i}\left( t\right) }{dt}%
,~i\in \overline{1,m}.%
\end{array}
\label{eq37}
\end{equation}
\end{remark}

\begin{definition}
\label{d17}If $%
\begin{array}{ccc}
I & ^{\underrightarrow{\dot{c}}} & E_{|\func{Im}\left( \eta \circ h\circ
c\right) }%
\end{array}%
$ is a differentiable $\left( g,h\right) $-lift of the differentiable curve $%
c,$ then the section%
\begin{equation}
\begin{array}{ccc}
\func{Im}\left( \eta \circ h\circ c\right) & ^{\underrightarrow{u\left( c,%
\dot{c}\right) }} & E_{|\func{Im}\left( \eta \circ h\circ c\right) }\vspace*{%
1mm} \\ 
\eta \circ h\circ c\left( t\right) & \longmapsto & \dot{c}\left( t\right)%
\end{array}
\label{eq38}
\end{equation}%
will be called the\emph{\ canonical section associated to the couple }$%
\left( c,\dot{c}\right) .$
\end{definition}

\begin{definition}
\label{d18}If $\left( g,h\right) \in \mathbf{B}^{\mathbf{v}}\left( \left(
E,\pi ,M\right) ,\left( F,\nu ,N\right) \right) $ has the components 
\begin{equation*}
\begin{array}{c}
g_{a}^{\alpha };a\in \overline{1,r},~\alpha \in \overline{1,p}%
\end{array}%
\end{equation*}%
such that for any vector local $\left( n+p\right) $-chart $\left(
V,t_{V}\right) $ of $\left( F,\nu ,N\right) $ there exists the real
functions 
\begin{equation*}
\begin{array}{ccc}
V & ^{\underrightarrow{~\ \ \ \tilde{g}_{\alpha }^{a}~\ \ }} & \mathbb{R}%
\end{array}%
;~a\in \overline{1,r},~\alpha \in \overline{1,p}
\end{equation*}%
such that%
\begin{equation}
\begin{array}{c}
\tilde{g}_{\alpha }^{b}\left( \varkappa \right) \cdot g_{a}^{\alpha }\left(
\varkappa \right) =\delta _{a}^{b},%
\end{array}
\label{eq39}
\end{equation}%
for any $\varkappa \in V,$ then we will say that \emph{the }$\mathbf{B}^{%
\mathbf{v}}$\emph{-morphism }$\left( g,h\right) $\emph{\ is locally
invertible.}
\end{definition}

\begin{remark}
\label{r19}\textrm{In particular, if }$\left( Id_{TM},Id_{M},Id_{M}\right)
=\left( \rho ,\eta ,h\right) $\textrm{\ and the }$\mathbf{B}^{\mathbf{v}}$%
\textrm{\ morphism }$\left( g,Id_{M}\right) $\textrm{\ is locally
invertible, then we have the differentiable }$\left( g,Id_{M}\right) $%
\textrm{-lift}%
\begin{equation}
\begin{array}{ccl}
I & ^{\underrightarrow{\ \ \dot{c}\ \ }} & TM \\ 
t & \longmapsto & \displaystyle\tilde{g}_{j}^{i}\left( c\left( t\right)
\right) \frac{dc^{j}\left( t\right) }{dt}\frac{\partial }{\partial x^{i}}%
\left( c\left( t\right) \right) .%
\end{array}
\label{eq40}
\end{equation}
\end{remark}

\textrm{Moreover, if }$g=Id_{TM}$\textrm{, then we obtain the usual lift of
tangent vectors}%
\begin{equation}
\begin{array}{ccl}
I & ^{\underrightarrow{\ \ \dot{c}\ \ }} & TM\vspace*{1mm} \\ 
t & \longmapsto & \displaystyle\frac{dc^{i}\left( t\right) }{dt}\frac{%
\partial }{\partial x^{i}}\left( c\left( t\right) \right) .%
\end{array}
\label{eq41}
\end{equation}

\begin{definition}
\label{d20}If $%
\begin{array}{ccl}
I & ^{\underrightarrow{\ \ \dot{c}\ \ }} & E_{|\func{Im}\left( \eta \circ
h\circ c\right) }%
\end{array}%
$ is a differentiable $\left( g,h\right) $-lift of differentiable curve $c,$
such that its components functions $\left( y^{a},~a\in \overline{1,n}\right) 
$ are solutions for the differentiable system of equations:%
\begin{equation}
\begin{array}[b]{c}
\frac{du^{a}}{dt}+\left( \rho ,\eta \right) \Gamma _{\alpha }^{a}\circ
u\left( c,\dot{c}\right) \circ \left( \eta \circ h\circ c\right) \cdot
g_{b}^{\alpha }\circ h\circ c\cdot u^{b}=0,%
\end{array}
\label{eq42}
\end{equation}%
then we will say that \emph{the }$\left( g,h\right) $\emph{-lift }$\dot{c}$%
\emph{\ is parallel with respect to the }$\left( \rho ,\eta \right) $\emph{%
-connection }$\left( \rho ,\eta \right) \Gamma .$

\begin{remark}
\label{r21}\textrm{In particular, if }$\left( \rho ,\eta ,h\right) =\left(
Id_{TM},Id_{M},Id_{M}\right) $\textrm{\ and the\ }$\mathbf{B}^{\mathbf{v}}$%
\textrm{-morphism }$\left( g,Id_{M}\right) $\textrm{\ is locally invertible,
then the differentiable }$\left( g,Id_{TM}\right) $\textrm{-lift}%
\begin{equation}
\begin{array}{ccl}
I & ^{\underrightarrow{\ \ \dot{c}\ \ }} & TM\vspace*{1mm} \\ 
t & \longmapsto & \displaystyle\left( \tilde{g}_{j}^{i}\circ c\cdot \frac{%
dc^{j}}{dt}\right) \frac{\partial }{\partial x^{i}}\left( c\left( t\right)
\right) ,%
\end{array}
\label{eq43}
\end{equation}%
\textrm{is parallel with respect to the connection }$\Gamma $\textrm{\ if
the component functions }%
\begin{equation*}
\begin{array}[b]{c}
\left( \tilde{g}_{j}^{i}\circ c\cdot \frac{dc^{j}}{dt},~i\in \overline{1,n}%
\right)%
\end{array}%
\end{equation*}%
\textrm{are solutions for the differentiable system of equations}%
\begin{equation}
\begin{array}[b]{c}
\frac{du^{i}}{dt}+\Gamma _{k}^{i}\circ u\left( c,\dot{c}\right) \circ c\cdot
g_{h}^{k}\circ c\cdot u^{h}=0,%
\end{array}
\label{eq44}
\end{equation}%
\textrm{namely}%
\begin{equation}
\begin{array}{l}
\displaystyle\frac{d}{dt}\left( \tilde{g}_{j}^{i}\left( c\left( t\right)
\right) \cdot \frac{dc^{j}\left( t\right) }{dt}\right) \vspace*{1mm} \\ 
\qquad \displaystyle+\Gamma _{k}^{i}\left( \left( \tilde{g}_{j}^{i}\left(
c\left( t\right) \right) \cdot \frac{dc^{j}\left( t\right) }{dt}\right)
\cdot \frac{\partial }{\partial x^{i}}\left( c\left( t\right) \right)
\right) \cdot \frac{dc^{k}\left( t\right) }{dt}=0.%
\end{array}
\label{eq45}
\end{equation}%
\textrm{Moreover, if }$g=Id_{TM}$\textrm{, then the usual lift of tangent
vectors \eqref{eq41} is parallel with respect to the connection }$\Gamma $%
\textrm{\ if the component functions }$\left( \frac{dc^{j}}{dt},~j\in 
\overline{1,n}\right) $\textrm{\ are solutions for the differentiable system
of equations}%
\begin{equation}
\begin{array}[b]{c}
\frac{du^{i}}{dt}+\Gamma _{k}^{i}\circ u\left( c,\dot{c}\right) \circ c\cdot
u^{k}=0,%
\end{array}
\label{eq46}
\end{equation}%
\textrm{namely}%
\begin{equation}
\begin{array}[b]{c}
\frac{d}{dt}\left( \frac{dc^{j}\left( t\right) }{dt}\right) +\Gamma
_{k}^{i}\left( \frac{dc^{j}\left( t\right) }{dt}\cdot \frac{\partial }{%
\partial x^{i}}\left( c\left( t\right) \right) \right) \cdot \frac{%
dc^{k}\left( t\right) }{dt}=0.%
\end{array}
\label{eq47}
\end{equation}
\end{remark}
\end{definition}

\textrm{\bigskip }

\section{Remarkable $\mathbf{Mod}$-endomorphisms}

\ \ \ 

We consider the following diagram:%
\begin{equation}
\begin{array}{rcl}
E &  & \left( F,\left[ ,\right] _{F,h},\left( \rho ,\eta \right) \right) \\ 
\pi \downarrow &  & ~\downarrow \nu \\ 
M & ^{\underrightarrow{~\ \ \ \ h~\ \ \ \ }} & ~\ N%
\end{array}
\label{eq48}
\end{equation}%
where $\left( E,\pi ,M\right) \in \left\vert \mathbf{B}^{\mathbf{v}%
}\right\vert $ and $\left( \left( F,\nu ,N\right) ,\left[ ,\right]
_{F,h},\left( \rho ,\eta \right) \right) $ is a generalized Lie algebroid.

\begin{definition}
\label{d22}For any $\mathbf{Mod}$-endomorphism $e$ of $\Gamma \left( \left(
\rho ,\eta \right) TE,\left( \rho ,\eta \right) \tau _{E},E\right) $ we
define the application of Nijenhuis \ type 
\begin{equation*}
\!\!\Gamma \!((\rho ,\!\eta )TE,\!(\rho ,\eta )\tau _{E},E)^{2~\ 
\underrightarrow{~\ \ N_{e}~\ \ }}~\ \Gamma \!((\rho ,\eta )TE,\!(\rho ,\eta
)\tau _{E},\!E)\vspace*{1mm}
\end{equation*}%
defined by 
\begin{equation*}
\begin{array}{c}
N_{e}\left( X,Y\right) =\left[ eX,eY\right] _{(\rho ,\!\eta )TE}+e^{2}\left[
X,Y\right] _{(\rho ,\!\eta )TE}-e\left[ eX,Y\right] _{(\rho ,\!\eta )TE}-e%
\left[ X,eY\right] _{(\rho ,\!\eta )\rho TE},%
\end{array}%
\end{equation*}%
for any $X,Y\in \Gamma \!((\rho ,\eta )TE,\!(\rho ,\eta )\tau _{E},\!E)%
\vspace*{1mm}.$
\end{definition}

\bigskip

\subsection{Projectors}

\ \ 

\begin{definition}
\label{d23}Any $\mathbf{Mod}$-endomorphism $e$ of $\Gamma \left( (\rho ,\eta
)TE,\break (\rho ,\eta \ )\tau _{E},E\right) $ with the property%
\begin{equation}
\begin{array}{l}
e^{2}=e%
\end{array}
\label{eq49}
\end{equation}%
will be called \emph{projector.}
\end{definition}

\begin{example}
\label{e24}\textrm{The }$\mathbf{Mod}$\textrm{-endomorphism }%
\begin{equation*}
\begin{array}{rcl}
\Gamma \left( \left( \rho ,\eta \right) TE,\left( \rho ,\eta \right) \tau
_{E},E\right) & ^{\underrightarrow{\ \ \mathcal{V}\ \ }} & \Gamma \left(
\left( \rho ,\eta \right) TE,\left( \rho ,\eta \right) \tau _{E},E\right) \\ 
Z^{\alpha }\tilde{\delta}_{\alpha }+Y^{a}\overset{\cdot }{\tilde{\partial}}%
_{a} & \longmapsto & Y^{a}\overset{\cdot }{\tilde{\partial}}_{a}%
\end{array}%
\end{equation*}%
\textrm{is a projector which will be called\ vertical projector.}
\end{example}

\begin{remark}
\label{r25}\textrm{We have }$V\left( \tilde{\delta}_{\alpha }\right) =0$%
\textrm{\ and }$V\left( \overset{\cdot }{\tilde{\partial}}_{a}\right) =%
\overset{\cdot }{\tilde{\partial}}_{a}.$\textrm{\ Therefore, it follows%
\vspace*{-2mm} }%
\begin{equation*}
\mathcal{V}\left( \tilde{\partial}_{\alpha }\right) =\left( \rho ,\eta
\right) \Gamma _{\alpha }^{a}\overset{\cdot }{\tilde{\partial}}_{a}.
\end{equation*}
\end{remark}

\textrm{In addition, we obtain the equality}%
\begin{equation}
\begin{array}[b]{c}
\Gamma \left( \left( \rho ,\eta \right) \Gamma ,Id_{E}\right) \left(
Z^{\alpha }\tilde{\partial}_{\alpha }+Y^{a}\overset{\cdot }{\tilde{\partial}}%
_{a}\right) =\mathcal{V}\left( Z^{\alpha }\tilde{\partial}_{\alpha }+Y^{a}%
\overset{\cdot }{\tilde{\partial}}_{a}\right) ,%
\end{array}
\label{eq50}
\end{equation}%
\textrm{for any} $Z^{\alpha }\tilde{\partial}_{\alpha }+Y^{a}\overset{\cdot }%
{\tilde{\partial}}_{a}\in \Gamma \left( \left( \rho ,\eta \right) TE,\left(
\rho ,\eta \right) \tau _{E},E\right) .$

\begin{theorem}
\label{t26}\emph{A }$(\rho ,\eta )$\emph{-connection for the vector bundle} $%
(E,\pi ,M)$ \emph{is characterized by the existence of a} $\mathbf{Mod}$%
\emph{-endomorphism} $\mathcal{V}$ \emph{of }$\Gamma \left( \left( \rho
,\eta \right) TE,\left( \rho ,\eta \right) \tau _{E},E\right) $ \emph{with
the properties:}%
\begin{equation}
\begin{array}{c}
\mathcal{V}\left( \Gamma \left( \left( \rho ,\eta \right) TE,\left( \rho
,\eta \right) \tau _{E},E\right) \right) \subset \Gamma \left( V\left( \rho
,\eta \right) TE,\left( \rho ,\eta \right) \tau _{E},E\right) \vspace*{1mm}
\\ 
\mathcal{V}\left( X\right) =X\Longleftrightarrow X\in \Gamma \left( V\left(
\rho ,\eta \right) TE,\left( \rho ,\eta \right) \tau _{E},E\right)%
\end{array}
\label{eq51}
\end{equation}
\end{theorem}

\begin{example}
\label{e27}\textrm{The }$\mathbf{Mod}$\textrm{-endomorphism }%
\begin{equation*}
\begin{array}{rcl}
\Gamma \left( \left( \rho ,\eta \right) TE,\left( \rho ,\eta \right) \tau
_{E},E\right) & ^{\underrightarrow{\ \ \mathcal{H}\ \ }} & \Gamma \left(
\left( \rho ,\eta \right) TE,\left( \rho ,\eta \right) \tau _{E},E\right) \\ 
Z^{\alpha }\tilde{\delta}_{\alpha }+Y^{a}\overset{\cdot }{\tilde{\partial}}%
_{a} & \longmapsto & Z^{\alpha }\tilde{\delta}_{\alpha }%
\end{array}%
\end{equation*}%
\textrm{is a projector which will be called horizontal projector.}
\end{example}

\begin{remark}
\label{r28}\textrm{We have }$H\left( \tilde{\delta}_{\alpha }\right) =\tilde{%
\delta}_{\alpha }$\textrm{\ and }$H\big(\overset{\cdot }{\tilde{\partial}}%
_{a}\big)=0.$\textrm{\ Therefore, we obtain }$H\left( \tilde{\partial}%
_{\alpha }\right) =\tilde{\delta}_{\alpha }.$
\end{remark}

\begin{theorem}
\label{t29}\emph{A }$\left( \rho ,\eta \right) $\emph{-connection for the
vector bundle} $\left( E,\pi ,M\right) $ \emph{\ is characterized by the
existence of a }$\mathbf{Mod}$\emph{-endomorphism }$\mathcal{H}$\emph{\ of }$%
\Gamma \left( \left( \rho ,\eta \right) TE,\left( \rho ,\eta \right) \tau
_{E},E\right) $ \emph{with the properties:}%
\begin{equation}
\begin{array}{c}
\mathcal{H}\left( \Gamma \left( \left( \rho ,\eta \right) TE,\left( \rho
,\eta \right) \tau _{E},E\right) \right) \subset \Gamma \left( H\left( \rho
,\eta \right) TE,\left( \rho ,\eta \right) \tau _{E},E\right) \vspace*{1mm}
\\ 
\mathcal{H}\left( X\right) =X\Longleftrightarrow X\in \Gamma \left( H\left(
\rho ,\eta \right) TE,\left( \rho ,\eta \right) \tau _{E},E\right) .%
\end{array}
\label{eq52}
\end{equation}
\end{theorem}

\begin{corollary}
\label{c30}\emph{A }$\left( \rho ,\eta \right) $\emph{-connection for the
vector bundle} $\left( E,\pi ,M\right) $ \emph{is characterized by the
existence of a }$\mathbf{Mod}$\emph{-endomorphism }$\mathcal{H}$\emph{\ of }$%
\Gamma \left( \left( \rho ,\eta \right) TE,\left( \rho ,\eta \right) \tau
_{E},E\right) $\emph{\ with the properties:}%
\begin{equation}
\begin{array}{c}
\mathcal{H}^{2}=\mathcal{H}\vspace*{1mm}, \\ 
Ker\left( \mathcal{H}\right) =\left( \Gamma \left( V\left( \rho ,\eta
\right) TE,\left( \rho ,\eta \right) \tau _{E},E\right) ,+,\cdot \right) .%
\end{array}
\label{eq53}
\end{equation}
\end{corollary}

\begin{remark}
\label{r31}\textrm{For any }$X\in \Gamma \left( \left( \rho ,\eta \right)
TE,\left( \rho ,\eta \right) \tau _{E},E\right) $\textrm{\ we obtain the
unique decomposition} 
\begin{equation*}
X=\mathcal{H}X+\mathcal{V}X.
\end{equation*}
\end{remark}

\begin{proposition}
\label{p32}\emph{After some calculations we obtain}%
\begin{equation}
\begin{array}{c}
N_{\mathcal{V}}\left( X,Y\right) =\mathcal{V}\left[ \mathcal{H}X,\mathcal{H}Y%
\right] _{\left( \rho ,\eta \right) TE}=N_{\mathcal{H}}\left( X,Y\right) ,%
\end{array}
\label{eq54}
\end{equation}%
\emph{\ }\textit{for any }$X,Y\in \Gamma \left( \left( \rho ,\eta \right)
TE,\left( \rho ,\eta \right) \tau _{E},E\right) .$
\end{proposition}

\begin{corollary}
\label{c33}\emph{The horizontal interior differential system }%
\begin{equation*}
\left( H\left( \rho ,\eta \right) TE,\left( \rho ,\eta \right) \tau
_{E},E\right)
\end{equation*}%
\emph{is involutive if and only if }$N_{\mathcal{V}}=0$\emph{\ or }$N_{%
\mathcal{H}}=0.$
\end{corollary}

\subsection{The almost product structure}

\ \ 

\begin{definition}
\label{d34}Any $\mathbf{Mod}$-endomorphism $e$ of $\Gamma \left( (\rho ,\eta
)TE,\break (\rho ,\eta \ )\tau _{E},E\right) $ with the property%
\begin{equation}
\begin{array}{c}
e^{2}=Id%
\end{array}
\label{eq55}
\end{equation}%
will be called the \emph{almost product structure}.
\end{definition}

\begin{example}
\label{e35}\textrm{The }$\mathbf{Mod}$\textrm{-endomorphism }%
\begin{equation*}
\begin{array}{rcl}
\Gamma \left( \left( \rho ,\eta \right) TE,\left( \rho ,\eta \right) \tau
_{E},E\right) & ^{\underrightarrow{\ \ \mathcal{P}\ \ }} & \Gamma \left(
\left( \rho ,\eta \right) TE,\left( \rho ,\eta \right) \tau _{E},E\right) \\ 
Z^{\alpha }\tilde{\delta}_{\alpha }+Y^{a}\overset{\cdot }{\tilde{\partial}}%
_{a} & \longmapsto & Z^{\alpha }\tilde{\delta}_{\alpha }-Y^{a}\overset{\cdot 
}{\tilde{\partial}}_{a}%
\end{array}%
\end{equation*}%
\textrm{is an almost product structure.}
\end{example}

\begin{remark}
\label{r36}\textrm{The previous almost product structure has the properties:}%
\begin{equation}
\begin{array}{l}
\mathcal{P}=\left( 2\mathcal{H}-Id\right) ; \\ 
\mathcal{P}=\left( Id-2\mathcal{V}\right) ; \\ 
\mathcal{P}=\left( \mathcal{H}-\mathcal{V}\right) .%
\end{array}
\label{eq56}
\end{equation}
\end{remark}

\begin{remark}
\label{r37}\textrm{We obtain that }$P\left( \tilde{\delta}_{\alpha }\right) =%
\tilde{\delta}_{\alpha }$\textrm{\ and }$P\left( \overset{\cdot }{\tilde{%
\partial}}_{a}\right) =-\overset{\cdot }{\tilde{\partial}}_{a}.$\textrm{\
Therefore, it follows \vspace*{-2mm} }%
\begin{equation*}
\mathcal{P}\left( \tilde{\partial}_{\alpha }\right) =\tilde{\delta}_{\alpha
}-\rho \Gamma _{\alpha }^{a}\overset{\cdot }{\tilde{\partial}}_{a}.
\end{equation*}
\end{remark}

\begin{theorem}
\label{t38}\emph{A }$\left( \rho ,\eta \right) $\emph{-connection for the
vector bundle }$\left( E,\pi ,M\right) $\emph{\ is characterized by the
existence of a }$\mathbf{Mod}$\emph{-endomorphism }$\mathcal{P}$\emph{\ of }$%
\Gamma \left( \left( \rho ,\eta \right) TE,\left( \rho ,\eta \right) \tau
_{E},E\right) $ \emph{with the following property:}%
\begin{equation}
\begin{array}{c}
\mathcal{P}\left( X\right) =-X\Longleftrightarrow X\in \Gamma \left( V\left(
\rho ,\eta \right) TE,\left( \rho ,\eta \right) \tau _{E},E\right) .%
\end{array}
\label{eq57}
\end{equation}
\end{theorem}

\begin{proposition}
\label{p39}\emph{After some calculations, we obtain }%
\begin{equation*}
N_{\mathcal{P}}\left( X,Y\right) =4\mathcal{V}\left[ \mathcal{H}X,\mathcal{H}%
Y\right] ,
\end{equation*}%
\emph{for any }$X,Y\in \Gamma \left( \left( \rho ,\eta \right) TE,\left(
\rho ,\eta \right) \tau _{E},E\right) .$
\end{proposition}

\begin{corollary}
\label{c40}\emph{The horizontal interior differential system }$\left(
H\left( \rho ,\eta \right) TE,\left( \rho ,\eta \right) \tau _{E},E\right) $ 
\emph{is involutive if and only if }$N_{\mathcal{P}}=0.$
\end{corollary}

\subsection{\noindent The almost tangent structure}

\ \ 

\begin{definition}
\label{d41}Any $\mathbf{Mod}$-endomorphism $e$ of $\left( \Gamma \!((\rho
,\eta )TE,\break (\rho ,\eta )\tau _{E},E\right) $ with the property%
\begin{equation}
\begin{array}{c}
e^{2}=0%
\end{array}
\label{eq58}
\end{equation}%
will be called the \emph{almost tangent structure.}
\end{definition}

\begin{example}
\label{e42}\textrm{If }$\left( E,\pi ,M\right) =\left( F,\nu ,N\right) $%
\textrm{, }$g\in \mathbf{Man}\left( E,E\right) $\textrm{\ such that }$\left(
g,h\right) $\textrm{\ is a locally invertible }$\mathbf{B}^{\mathbf{v}}$%
\textrm{-morphism, then the }$\mathbf{Mod}$\textrm{-endomorphism }%
\begin{equation*}
\begin{array}{rcl}
\Gamma \left( \left( \rho ,\eta \right) TE,\left( \rho ,\eta \right) \tau
_{E},E\right) & ^{\underrightarrow{\mathcal{J}_{\left( g,h\right) }}} & 
\Gamma \left( \left( \rho ,\eta \right) TE,\left( \rho ,\eta \right) \tau
_{E},E\right) \\ 
Z^{a}\tilde{\delta}_{a}+Y^{b}\overset{\cdot }{\tilde{\partial}}_{b} & 
\longmapsto & \left( \tilde{g}_{a}^{b}\circ h\circ \pi \right) Z^{a}\overset{%
\cdot }{\tilde{\partial}}_{b}%
\end{array}%
\end{equation*}%
\textrm{is an almost tangent structure which will be called the almost
tangent structure associated to the }$\mathbf{B}^{\mathbf{v}}$\textrm{%
-morphism }$\left( g,h\right) $\textrm{. (see: \ref{d18} )}
\end{example}

\begin{example}
\label{e43}\emph{\ }\textrm{We obtain that }%
\begin{equation*}
\mbox{$\mathcal{J}_{\left( g,h\right)
}\left( \tilde{\delta}_{a}\right) =\mathcal{J}_{\left( g,h\right) }\left(
\tilde{\partial}_{a}\right) =\left( \tilde{g}_{a}^{b}\circ h\circ \pi
\right) \overset{\cdot }{\tilde{\partial}}_{b}$ and $\mathcal{J}_{\left(
g,h\right) }\left( \overset{\cdot }{\tilde{\partial}}_{b}\right) =0.$}
\end{equation*}%
\textrm{and we have the following properties:}%
\begin{equation}
\begin{array}{rcl}
\mathcal{J}_{\left( g,h\right) }\circ \mathcal{P} & = & \mathcal{J}_{\left(
g,h\right) };\vspace*{1mm} \\ 
\mathcal{P}\circ \mathcal{J}_{\left( g,h\right) } & = & -\mathcal{J}_{\left(
g,h\right) };\vspace*{1mm} \\ 
\mathcal{J}_{\left( g,h\right) }\circ \mathcal{H} & = & \mathcal{J}_{\left(
g,h\right) };\vspace*{1mm} \\ 
\mathcal{H}\circ \mathcal{J}_{\left( g,h\right) } & = & 0;\vspace*{1mm} \\ 
\mathcal{J}_{\left( g,h\right) }\circ \mathcal{V} & = & 0;\vspace*{1mm} \\ 
\mathcal{V}\circ \mathcal{J}_{\left( g,h\right) } & = & \mathcal{J}_{\left(
g,h\right) };\vspace*{1mm} \\ 
N_{\mathcal{J}_{\left( g,h\right) }} & = & 0.%
\end{array}
\label{eq59}
\end{equation}
\end{example}

\section{Distinguished linear $\left( \protect\rho ,\protect\eta \right) $%
-connections}

\ \ \ 

We consider the following diagram:%
\begin{equation}
\begin{array}{rcl}
E &  & \left( F,\left[ ,\right] _{F,h},\left( \rho ,\eta \right) \right) \\ 
\pi \downarrow &  & ~\downarrow \nu \\ 
M & ^{\underrightarrow{~\ \ \ \ h~\ \ \ \ }} & ~\ N%
\end{array}
\label{eq60}
\end{equation}
where $\left( E,\pi ,M\right) \in \left\vert \mathbf{B}^{\mathbf{v}%
}\right\vert $ and $\left( \left( F,\nu ,N\right) ,\left[ ,\right]
_{F,h},\left( \rho ,\eta \right) \right) $ is a generalized Lie algebroid.

Let 
\begin{equation*}
\left( \mathcal{T}~_{q,s}^{p,r}\left( \left( \rho ,\eta \right) TE,\left(
\rho ,\eta \right) \tau _{E},E\right) ,+,\cdot \right)
\end{equation*}%
be the $\mathcal{F}\left( E\right) $-module of tensor fields by $\left(
_{q,s}^{p,r}\right) $-type from the generalized tangent bundle 
\begin{equation*}
\left( H\left( \rho ,\eta \right) TE,\left( \rho ,\eta \right) \tau
_{E},E\right) \oplus \left( V\left( \rho ,\eta \right) TE,\left( \rho ,\eta
\right) \tau _{E},E\right) .
\end{equation*}

An arbitrary tensor field $T$ is written as 
\begin{equation*}
\begin{array}{c}
T=T_{\beta _{1}...\beta _{q}b_{1}...b_{s}}^{\alpha _{1}...\alpha
_{p}a_{1}...a_{r}}\tilde{\delta}_{\alpha _{1}}\otimes ...\otimes \tilde{%
\delta}_{\alpha _{p}}\otimes d\tilde{z}^{\beta _{1}}\otimes ...\otimes d%
\tilde{z}^{\beta _{q}}\otimes \\ 
\overset{\cdot }{\tilde{\partial}}_{a_{1}}\otimes ...\otimes \overset{\cdot }%
{\tilde{\partial}}_{a_{r}}\otimes \delta \tilde{y}^{b_{1}}\otimes ...\otimes
\delta \tilde{y}^{b_{s}}.%
\end{array}%
\end{equation*}

Let 
\begin{equation*}
\left( \mathcal{T}~\left( \left( \rho ,\eta \right) TE,\left( \rho ,\eta
\right) \tau _{E},E\right) ,+,\cdot ,\otimes \right)
\end{equation*}%
be the tensor fields algebra of generalized tangent bundle $\left( \left(
\rho ,\eta \right) TE,\left( \rho ,\eta \right) \tau _{E},E\right) $.

\begin{definition}
\label{d44}Let $\left( E,\pi ,M\right) $ be a vector bundle endowed with a $%
\left( \rho ,\eta \right) $-connection $\left( \rho ,\eta \right) \Gamma $
and let 
\begin{equation*}
\begin{array}{l}
\left( X,T\right) ^{\ \underrightarrow{\left( \rho ,\eta \right) D}\,}%
\vspace*{1mm}\left( \rho ,\eta \right) D_{X}T%
\end{array}%
\end{equation*}%
be a covariant $\left( \rho ,\eta \right) $-derivative for the tensor
algebra 
\begin{equation*}
\left( \mathcal{T}~\left( \left( \rho ,\eta \right) TE,\left( \rho ,\eta
\right) \tau _{E},E\right) ,+,\cdot ,\otimes \right)
\end{equation*}%
of the generalized tangent bundle 
\begin{equation*}
\left( \left( \rho ,\eta \right) TE,\left( \rho ,\eta \right) \tau
_{E},E\right)
\end{equation*}%
which preserves the horizontal and vertical interior differential systems by
parallelism. (see \cite{2})
\end{definition}

The real local functions 
\begin{equation*}
\left( \left( \rho ,\eta \right) H_{\beta \gamma }^{\alpha },\left( \rho
,\eta \right) H_{b\gamma }^{a},\left( \rho ,\eta \right) V_{\beta c}^{\alpha
},\left( \rho ,\eta \right) V_{bc}^{a}\right)
\end{equation*}%
defined by the following equalities:%
\begin{equation}
\begin{array}{ll}
\left( \rho ,\eta \right) D_{\tilde{\delta}_{\gamma }}\tilde{\delta}_{\beta
}=\left( \rho ,\eta \right) H_{\beta \gamma }^{\alpha }\tilde{\delta}%
_{\alpha }, & \left( \rho ,\eta \right) D_{\tilde{\delta}_{\gamma }}\overset{%
\cdot }{\tilde{\partial}}_{b}=\left( \rho ,\eta \right) H_{b\gamma }^{a}%
\overset{\cdot }{\tilde{\partial}}_{a} \\ 
\left( \rho ,\eta \right) D_{\overset{\cdot }{\tilde{\partial}}_{c}}\tilde{%
\delta}_{\beta }=\left( \rho ,\eta \right) V_{\beta c}^{\alpha }\tilde{\delta%
}_{\alpha }, & \left( \rho ,\eta \right) D_{\overset{\cdot }{\tilde{\partial}%
}_{c}}\overset{\cdot }{\tilde{\partial}}_{b}=\left( \rho ,\eta \right)
V_{bc}^{a}\overset{\cdot }{\tilde{\partial}}_{a}%
\end{array}
\label{eq61}
\end{equation}%
are the components of a linear $\left( \rho ,\eta \right) $-connection $%
\left( \left( \rho ,\eta \right) H,\left( \rho ,\eta \right) V\right) $ for
the generalized tangent bundle $\left( \left( \rho ,\eta \right) TE,\left(
\rho ,\eta \right) \tau _{E},E\right) $ which will be called the \emph{%
distinguished linear }$\left( \rho ,\eta \right) $\emph{-connection.}

If $h=Id_{M},$ then the distinguished linear $\left( Id_{TM},Id_{M}\right) $%
-connection is the classical \emph{distinguished linear connection.}

The components of a distinguished linear connection $\left( H,V\right) $
will be denoted 
\begin{equation*}
\left( H_{jk}^{i},H_{bk}^{a},V_{jc}^{i},V_{bc}^{a}\right) .
\end{equation*}

\begin{theorem}
\label{t45}\emph{If }$((\rho ,\eta )H,(\rho ,\eta )V)$ \emph{is a
distinguished linear} $(\rho ,\eta )$-\emph{connection for the generalized
tangent bundle }$\left( \left( \rho ,\eta \right) TE,\left( \rho ,\eta
\right) \tau _{E},E\right) $\emph{, then its components satisfy the change
relations: }%
\begin{equation}
\begin{array}{ll}
\left( \rho ,\eta \right) H_{\beta 
{\acute{}}%
\gamma 
{\acute{}}%
}^{\alpha 
{\acute{}}%
}\!\! & =\Lambda _{\alpha }^{\alpha 
{\acute{}}%
}\circ h\circ \pi \cdot \left[ \Gamma \left( \tilde{\rho},Id_{E}\right)
\left( \tilde{\delta}_{\gamma }\right) \left( \Lambda _{\beta 
{\acute{}}%
}^{\alpha }\circ h\circ \pi \right) +\right. \vspace*{1mm} \\ 
& +\left. \left( \rho ,\eta \right) H_{\beta \gamma }^{\alpha }\cdot \Lambda
_{\beta 
{\acute{}}%
}^{\beta }\circ h\circ \pi \right] \cdot \Lambda _{\gamma 
{\acute{}}%
}^{\gamma }\circ h\circ \pi ,\vspace*{2mm} \\ 
\left( \rho ,\eta \right) H_{b%
{\acute{}}%
\gamma 
{\acute{}}%
}^{a%
{\acute{}}%
}\!\! & =M_{a}^{a%
{\acute{}}%
}\circ \pi \cdot \left[ \Gamma \left( \tilde{\rho},Id_{E}\right) \left( 
\tilde{\delta}_{\gamma }\right) \left( M_{b%
{\acute{}}%
}^{a}\circ \pi \right) +\right. \vspace*{1mm} \\ 
& \left. +\left( \rho ,\eta \right) H_{b\gamma }^{a}\cdot M_{b%
{\acute{}}%
}^{b}\circ \pi \right] \cdot \Lambda _{\gamma 
{\acute{}}%
}^{\gamma }\circ h\circ \pi ,\vspace*{2mm} \\ 
\left( \rho ,\eta \right) V_{\beta 
{\acute{}}%
c%
{\acute{}}%
}^{\alpha 
{\acute{}}%
}\!\! & =\Lambda _{\alpha 
{\acute{}}%
}^{\alpha }\circ h\circ \pi \cdot \left( \rho ,\eta \right) V_{\beta
c}^{\alpha }\cdot \Lambda _{\beta 
{\acute{}}%
}^{\beta }\circ h\circ \pi \cdot M_{c%
{\acute{}}%
}^{c}\circ \pi ,\vspace*{2mm} \\ 
\left( \rho ,\eta \right) V_{b%
{\acute{}}%
c%
{\acute{}}%
}^{a%
{\acute{}}%
}\!\! & =M_{a}^{a%
{\acute{}}%
}\circ \pi \cdot \left( \rho ,\eta \right) V_{bc}^{a}\cdot M_{b%
{\acute{}}%
}^{b}\circ \pi \cdot M_{c%
{\acute{}}%
}^{c}\circ \pi .%
\end{array}
\label{eq62}
\end{equation}
\end{theorem}

\begin{corollary}
\label{c46}\emph{In the particular case of Lie algebroids (see \cite{3}), }$%
\left( \eta ,h\right) =\left( Id_{M},Id_{M}\right) ,$\emph{\ we obtain}%
\begin{equation}
\begin{array}{ll}
\rho H_{\beta 
{\acute{}}%
\gamma 
{\acute{}}%
}^{\alpha 
{\acute{}}%
}\!\! & =\Lambda _{\alpha }^{\alpha 
{\acute{}}%
}\circ \pi \cdot \left[ \Gamma \left( \tilde{\rho},Id_{E}\right) \left( 
\tilde{\delta}_{\gamma }\right) \left( \Lambda _{\beta 
{\acute{}}%
}^{\alpha }\circ \pi \right) +\rho H_{\beta \gamma }^{\alpha }\cdot \Lambda
_{\beta 
{\acute{}}%
}^{\beta }\circ \pi \right] \cdot \Lambda _{\gamma 
{\acute{}}%
}^{\gamma }\circ \pi \\ 
\rho H_{b%
{\acute{}}%
\gamma 
{\acute{}}%
}^{a%
{\acute{}}%
}\!\! & =M_{a}^{a%
{\acute{}}%
}\circ \pi \cdot \left[ \Gamma \left( \tilde{\rho},Id_{E}\right) \left( 
\tilde{\delta}_{\gamma }\right) \left( M_{b%
{\acute{}}%
}^{a}\circ \pi \right) +\rho H_{b\gamma }^{a}\cdot M_{b%
{\acute{}}%
}^{b}\circ \pi \right] \cdot \Lambda _{\gamma 
{\acute{}}%
}^{\gamma }\circ \pi , \\ 
\rho V_{\beta 
{\acute{}}%
c%
{\acute{}}%
}^{\alpha 
{\acute{}}%
}\!\! & =\Lambda _{\alpha 
{\acute{}}%
}^{\alpha }\circ \pi \cdot \rho V_{\beta c}^{\alpha }\cdot \Lambda _{\beta 
{\acute{}}%
}^{\beta }\circ \pi \cdot M_{c%
{\acute{}}%
}^{c}\circ \pi ,\vspace*{2mm} \\ 
\rho V_{b%
{\acute{}}%
c%
{\acute{}}%
}^{a%
{\acute{}}%
}\!\! & =M_{a}^{a%
{\acute{}}%
}\circ \pi \cdot \rho V_{bc}^{a}\cdot M_{b%
{\acute{}}%
}^{b}\circ \pi \cdot M_{c%
{\acute{}}%
}^{c}\circ \pi .%
\end{array}
\label{eq63}
\end{equation}%
\emph{\ }
\end{corollary}

\emph{In the classical case, }$\left( \rho ,\eta ,h\right) =\left(
Id_{TE},Id_{M},Id_{M}\right) ,$\emph{\ we obtain that the components of a
distinguished linear connection }$\left( H,V\right) $\emph{\ verify the
change relations:}%
\begin{equation}
\begin{array}{cl}
H_{j%
{\acute{}}%
k%
{\acute{}}%
}^{i%
{\acute{}}%
} & =\frac{\partial x^{i%
{\acute{}}%
}}{\partial x^{i}}\circ \pi \cdot \left[ \frac{\delta }{\delta x^{k}}\left( 
\frac{\partial x^{i}}{\partial x^{j%
{\acute{}}%
}}\circ \pi \right) +H_{jk}^{i}\cdot \frac{\partial x^{j}}{\partial x^{j%
{\acute{}}%
}}\circ \pi \right] \cdot \frac{\partial x^{k}}{\partial x^{k%
{\acute{}}%
}}\circ \pi ,\vspace*{2mm} \\ 
H_{b%
{\acute{}}%
k%
{\acute{}}%
}^{a%
{\acute{}}%
} & =M_{a}^{a%
{\acute{}}%
}\circ \pi \cdot \left[ \frac{\delta }{\delta x^{k}}\left( M_{b%
{\acute{}}%
}^{a}\circ \pi \right) +H_{bk}^{a}\cdot M_{b%
{\acute{}}%
}^{b}\circ \pi \right] \cdot \frac{\partial x^{k}}{\partial x^{k%
{\acute{}}%
}}\circ \pi ,\vspace*{2mm} \\ 
V_{j%
{\acute{}}%
c%
{\acute{}}%
}^{i%
{\acute{}}%
} & =\frac{\partial x^{i%
{\acute{}}%
}}{\partial x^{i}}\circ \pi \cdot V_{jc}^{i}\frac{\partial x^{j}}{\partial
x^{j%
{\acute{}}%
}}\circ \pi \cdot M_{c%
{\acute{}}%
}^{c}\circ \pi ,\vspace*{3mm} \\ 
V_{b%
{\acute{}}%
c%
{\acute{}}%
}^{a%
{\acute{}}%
} & =M_{a}^{a%
{\acute{}}%
}\circ \pi \cdot V_{bc}^{a}\cdot M_{b%
{\acute{}}%
}^{b}\circ \pi \cdot M_{c%
{\acute{}}%
}^{c}\circ \pi .%
\end{array}
\label{eq64}
\end{equation}

\begin{example}
\label{e47}If $\left( E,\pi ,M\right) =\left( F,\nu ,N\right) $ is a vector
bundle endowed with the $\left( \rho ,\eta \right) $-connection $\left( \rho
,\eta \right) \Gamma $, then the local real functions%
\begin{equation}
\begin{array}[b]{c}
\left( \frac{\partial \left( \rho ,\eta \right) \Gamma _{\gamma }^{a}}{%
\partial y^{b}},\frac{\partial \left( \rho ,\eta \right) \Gamma _{\gamma
}^{a}}{\partial y^{b}},0,0\right)%
\end{array}
\label{eq65}
\end{equation}%
are the components of a distinguished linear $\left( \rho ,\eta \right) $%
\textit{-}connection for the generalized tangent bundle $\left( \left( \rho
,\eta \right) TE,\left( \rho ,\eta \right) \tau _{E},E\right) ,$ which will
by called the \emph{Berwald linear }$\left( \rho ,\eta \right) $\emph{%
-connection.}
\end{example}

The Berwald linear $(Id_{TM},Id_{M})$-connection are the usual \emph{%
Ber\-wald linear connection.}

\begin{theorem}
\label{t48}\emph{If the generalized tangent bundle} $\!(\!(\rho ,\!\eta
)T\!E,\!(\rho ,\!\eta )\tau _{E},\!E\!)$ \emph{is endowed with a
distinguished linear} $\!(\rho ,\!\eta )$\emph{-connection} $((\rho ,\eta
)H,(\rho ,\eta )V),$ \emph{then for any} 
\begin{equation*}
\begin{array}[b]{c}
X=Z^{\alpha }\tilde{\delta}_{\alpha }+Y^{a}\overset{\cdot }{\tilde{\partial}}%
_{a}\in \Gamma (\!(\rho ,\eta )TE,\!(\rho ,\!\eta )\tau _{E},\!E)%
\end{array}%
\end{equation*}%
\emph{and for any} 
\begin{equation*}
T\in \mathcal{T}_{qs}^{pr}\!(\!(\rho ,\eta )TE,\!(\rho ,\eta )\tau _{E},\!E),
\end{equation*}%
\emph{we obtain the formula:}%
\begin{equation}
\begin{array}{l}
\left( \rho ,\eta \right) D_{X}\left( T_{\beta _{1}...\beta
_{q}b_{1}...b_{s}}^{\alpha _{1}...\alpha _{p}a_{1}...a_{r}}\tilde{\delta}%
_{\alpha _{1}}\otimes ...\otimes \tilde{\delta}_{\alpha _{p}}\otimes d\tilde{%
z}^{\beta _{1}}\otimes ...\otimes \right. \vspace*{1mm} \\ 
\hspace*{9mm}\left. \otimes d\tilde{z}^{\beta _{q}}\otimes \overset{\cdot }{%
\tilde{\partial}}_{a_{1}}\otimes ...\otimes \overset{\cdot }{\tilde{\partial}%
}_{a_{r}}\otimes \delta \tilde{y}^{b_{1}}\otimes ...\otimes \delta \tilde{y}%
^{b_{s}}\right) =\vspace*{1mm} \\ 
\hspace*{9mm}=Z^{\gamma }T_{\beta _{1}...\beta _{q}b_{1}...b_{s}\mid \gamma
}^{\alpha _{1}...\alpha _{p}a_{1}...a_{r}}\tilde{\delta}_{\alpha
_{1}}\otimes ...\otimes \tilde{\delta}_{\alpha _{p}}\otimes d\tilde{z}%
^{\beta _{1}}\otimes ...\otimes d\tilde{z}^{\beta _{q}}\otimes \overset{%
\cdot }{\tilde{\partial}}_{a_{1}}\otimes ...\otimes \vspace*{1mm} \\ 
\hspace*{9mm}\otimes \overset{\cdot }{\tilde{\partial}}_{a_{r}}\otimes
\delta \tilde{y}^{b_{1}}\otimes ...\otimes \delta \tilde{y}%
^{b_{s}}+Y^{c}T_{\beta _{1}...\beta _{q}b_{1}...b_{s}}^{\alpha _{1}...\alpha
_{p}a_{1}...a_{r}}\mid _{c}\tilde{\delta}_{\alpha _{1}}\otimes ...\otimes 
\vspace*{1mm} \\ 
\hspace*{9mm}\otimes \tilde{\delta}_{\alpha _{p}}\otimes d\tilde{z}^{\beta
_{1}}\otimes ...\otimes d\tilde{z}^{\beta _{q}}\otimes \overset{\cdot }{%
\tilde{\partial}}_{a_{1}}\otimes ...\otimes \overset{\cdot }{\tilde{\partial}%
}_{a_{r}}\otimes \delta \tilde{y}^{b_{1}}\otimes ...\otimes \delta \tilde{y}%
^{b_{s}},%
\end{array}
\label{eq66}
\end{equation}%
\emph{where}%
\begin{equation}
\begin{array}{l}
T_{\beta _{1}...\beta _{q}b_{1}...b_{s}\mid \gamma }^{\alpha _{1}...\alpha
_{p}a_{1}...a_{r}}=\vspace*{2mm}\Gamma \left( \tilde{\rho},Id_{E}\right)
\left( \tilde{\delta}_{\gamma }\right) T_{\beta _{1}...\beta
_{q}b_{1}...b_{s}}^{\alpha _{1}...\alpha _{p}a_{1}...a_{r}} \\ 
\hspace*{8mm}+\left( \rho ,\eta \right) H_{\alpha \gamma }^{\alpha
_{1}}T_{\beta _{1}...\beta _{q}b_{1}...b_{s}}^{\alpha \alpha _{2}...\alpha
_{p}a_{1}...a_{r}}+...+\vspace*{2mm}\left( \rho ,\eta \right) H_{\alpha
\gamma }^{\alpha _{p}}T_{\beta _{1}...\beta _{q}b_{1}...b_{s}}^{\alpha
_{1}...\alpha _{p-1}\alpha a_{1}...a_{r}} \\ 
\hspace*{8mm}-\left( \rho ,\eta \right) H_{\beta _{1}\gamma }^{\beta
}T_{\beta \beta _{2}...\beta _{q}b_{1}...b_{s}}^{\alpha _{1}...\alpha
_{p}a_{1}...a_{r}}-...-\vspace*{2mm}\left( \rho ,\eta \right) H_{\beta
_{q}\gamma }^{\beta }T_{\beta _{1}...\beta _{q-1}\beta
b_{1}...b_{s}}^{\alpha _{1}...\alpha _{p}a_{1}...a_{r}} \\ 
\hspace*{8mm}+\left( \rho ,\eta \right) H_{a\gamma }^{a_{1}}T_{\beta
_{1}...\beta _{q}b_{1}...b_{s}}^{\alpha _{1}...\alpha
_{p}aa_{2}...a_{r}}+...+\vspace*{2mm}\left( \rho ,\eta \right) H_{a\gamma
}^{a_{r}}T_{\beta _{1}...\beta _{q}b_{1}...b_{s}}^{\alpha _{1}...\alpha
_{p}a_{1}...a_{r-1}a} \\ 
\hspace*{8mm}-\left( \rho ,\eta \right) H_{b_{1}\gamma }^{b}T_{\beta
_{1}...\beta _{q}bb_{2}...b_{s}}^{\alpha _{1}...\alpha _{p}a_{1}...a_{r}}-%
\vspace*{2mm}...-\left( \rho ,\eta \right) H_{b_{s}\gamma }^{b}T_{\beta
_{1}...\beta _{q}b_{1}...b_{s-1}b}^{\alpha _{1}...\alpha _{p}a_{1}...a_{r}}%
\end{array}
\label{eq67}
\end{equation}%
\emph{\ and}%
\begin{equation}
\begin{array}{l}
T_{\beta _{1}...\beta _{q}b_{1}...b_{s}}^{\alpha _{1}...\alpha
_{p}a_{1}...a_{r}}\mid _{c}=\Gamma \left( \tilde{\rho},Id_{E}\right) \left( 
\overset{\cdot }{\tilde{\partial}}_{c}\right) T_{\beta _{1}...\beta
_{q}b_{1}...b_{s}}^{\alpha _{1}...\alpha _{p}a_{1}...a_{r}} \\ 
\hspace*{8mm}+\left( \rho ,\eta \right) V_{\alpha c}^{\alpha _{1}}T_{\beta
_{1}...\beta _{q}b_{1}...b_{s}}^{\alpha \alpha _{2}...\alpha
_{p}a_{1}...a_{r}}+...+\left( \rho ,\eta \right) V_{\alpha c}^{\alpha
_{p}}T_{\beta _{1}...\beta _{q}b_{1}...b_{s}}^{\alpha _{1}...\alpha
_{p-1}\alpha a_{1}...a_{r}}\vspace*{2mm} \\ 
\hspace*{8mm}-\left( \rho ,\eta \right) V_{\beta _{1}c}^{\beta }T_{\beta
\beta _{2}...\beta _{q}b_{1}...b_{s}}^{\alpha _{1}...\alpha
_{p}a_{1}...a_{r}}-...-\left( \rho ,\eta \right) V_{\beta _{q}c}^{\beta
}T_{\beta _{1}...\beta _{q-1}\beta b_{1}...b_{s}}^{\alpha _{1}...\alpha
_{p}a_{1}...a_{r}}\vspace*{2mm} \\ 
\hspace*{8mm}+\left( \rho ,\eta \right) V_{ac}^{a_{1}}T_{\beta _{1}...\beta
_{q}b_{1}...b_{s}}^{\alpha _{1}...\alpha _{p}aa_{2}...a_{r}}+...+\left( \rho
,\eta \right) V_{ac}^{a_{r}}T_{\beta _{1}...\beta _{q}b_{1}...b_{s}}^{\alpha
_{1}...\alpha _{p}a_{1}...a_{r-1}a} \\ 
\hspace*{8mm}-\left( \rho ,\eta \right) V_{b_{1}c}^{b}T_{\beta _{1}...\beta
_{q}bb_{2}...b_{s}}^{\alpha _{1}...\alpha _{p}a_{1}...a_{r}}-...-\left( \rho
,\eta \right) V_{b_{s}c}^{b}T_{\beta _{1}...\beta
_{q}b_{1}...b_{s-1}b}^{\alpha _{1}...\alpha _{p}a_{1}...a_{r}}.%
\end{array}
\label{eq68}
\end{equation}
\end{theorem}

\begin{corollary}
\label{c49}\emph{In the particular case of Lie algebroids (see \cite{3}), }$%
\left( \eta ,h\right) =\left( Id_{M},Id_{M}\right) ,$\emph{\ we obtain}%
\begin{equation}
\begin{array}{l}
T_{\beta _{1}...\beta _{q}b_{1}...b_{s}\mid \gamma }^{\alpha _{1}...\alpha
_{p}a_{1}...a_{r}}=\vspace*{2mm}\Gamma \left( \tilde{\rho},Id_{E}\right)
\left( \tilde{\delta}_{\gamma }\right) T_{\beta _{1}...\beta
_{q}b_{1}...b_{s}}^{\alpha _{1}...\alpha _{p}a_{1}...a_{r}} \\ 
\hspace*{8mm}+\rho H_{\alpha \gamma }^{\alpha _{1}}T_{\beta _{1}...\beta
_{q}b_{1}...b_{s}}^{\alpha \alpha _{2}...\alpha _{p}a_{1}...a_{r}}+...+%
\vspace*{2mm}\rho H_{\alpha \gamma }^{\alpha _{p}}T_{\beta _{1}...\beta
_{q}b_{1}...b_{s}}^{\alpha _{1}...\alpha _{p-1}\alpha a_{1}...a_{r}} \\ 
\hspace*{8mm}-\rho H_{\beta _{1}\gamma }^{\beta }T_{\beta \beta _{2}...\beta
_{q}b_{1}...b_{s}}^{\alpha _{1}...\alpha _{p}a_{1}...a_{r}}-...-\vspace*{2mm}%
\rho H_{\beta _{q}\gamma }^{\beta }T_{\beta _{1}...\beta _{q-1}\beta
b_{1}...b_{s}}^{\alpha _{1}...\alpha _{p}a_{1}...a_{r}} \\ 
\hspace*{8mm}+\rho H_{a\gamma }^{a_{1}}T_{\beta _{1}...\beta
_{q}b_{1}...b_{s}}^{\alpha _{1}...\alpha _{p}aa_{2}...a_{r}}+...+\vspace*{2mm%
}\rho H_{a\gamma }^{a_{r}}T_{\beta _{1}...\beta _{q}b_{1}...b_{s}}^{\alpha
_{1}...\alpha _{p}a_{1}...a_{r-1}a} \\ 
\hspace*{8mm}-\rho H_{b_{1}\gamma }^{b}T_{\beta _{1}...\beta
_{q}bb_{2}...b_{s}}^{\alpha _{1}...\alpha _{p}a_{1}...a_{r}}-\vspace*{2mm}%
...-\rho H_{b_{s}\gamma }^{b}T_{\beta _{1}...\beta
_{q}b_{1}...b_{s-1}b}^{\alpha _{1}...\alpha _{p}a_{1}...a_{r}}%
\end{array}
\label{eq69}
\end{equation}%
\emph{\ and}%
\begin{equation}
\begin{array}{l}
T_{\beta _{1}...\beta _{q}b_{1}...b_{s}}^{\alpha _{1}...\alpha
_{p}a_{1}...a_{r}}\mid _{c}=\Gamma \left( \tilde{\rho},Id_{E}\right) \left( 
\overset{\cdot }{\tilde{\partial}}_{c}\right) T_{\beta _{1}...\beta
_{q}b_{1}...b_{s}}^{\alpha _{1}...\alpha _{p}a_{1}...a_{r}} \\ 
\hspace*{8mm}+\rho V_{\alpha c}^{\alpha _{1}}T_{\beta _{1}...\beta
_{q}b_{1}...b_{s}}^{\alpha \alpha _{2}...\alpha _{p}a_{1}...a_{r}}+...+\rho
V_{\alpha c}^{\alpha _{p}}T_{\beta _{1}...\beta _{q}b_{1}...b_{s}}^{\alpha
_{1}...\alpha _{p-1}\alpha a_{1}...a_{r}} \\ 
\hspace*{8mm}-\rho V_{\beta _{1}c}^{\beta }T_{\beta \beta _{2}...\beta
_{q}b_{1}...b_{s}}^{\alpha _{1}...\alpha _{p}a_{1}...a_{r}}-...-\rho
V_{\beta _{q}c}^{\beta }T_{\beta _{1}...\beta _{q-1}\beta
b_{1}...b_{s}}^{\alpha _{1}...\alpha _{p}a_{1}...a_{r}} \\ 
\hspace*{8mm}+\rho V_{ac}^{a_{1}}T_{\beta _{1}...\beta
_{q}b_{1}...b_{s}}^{\alpha _{1}...\alpha _{p}aa_{2}...a_{r}}+...+\rho
V_{ac}^{a_{r}}T_{\beta _{1}...\beta _{q}b_{1}...b_{s}}^{\alpha _{1}...\alpha
_{p}a_{1}...a_{r-1}a}\vspace*{2mm} \\ 
\hspace*{8mm}-\rho V_{b_{1}c}^{b}T_{\beta _{1}...\beta
_{q}bb_{2}...b_{s}}^{\alpha _{1}...\alpha _{p}a_{1}...a_{r}}-...-\rho
V_{b_{s}c}^{b}T_{\beta _{1}...\beta _{q}b_{1}...b_{s-1}b}^{\alpha
_{1}...\alpha _{p}a_{1}...a_{r}}.%
\end{array}
\label{eq70}
\end{equation}
\end{corollary}

\emph{In the classical case, }$\left( \rho ,\eta ,h\right) =\left(
Id_{TE},Id_{M},Id_{M}\right) ,$\emph{\ we obtain}%
\begin{equation}
\begin{array}{l}
T_{j_{1}...j_{q}b_{1}...b_{s}\mid k}^{i_{1}...i_{p}a_{1}...a_{r}}=\vspace*{%
2mm}\delta _{k}\left(
T_{j_{1}...j_{q}b_{1}...b_{s}}^{i_{1}...i_{p}a_{1}...a_{r}}\right) \\ 
\hspace*{8mm}%
+H_{ik}^{i_{1}}T_{j_{1}...j_{q}b_{1}...b_{s}}^{ii_{2}...i_{p}a_{1}...a_{r}}+...+%
\vspace*{2mm}H_{ik}^{i_{p}}T_{\beta _{1}...\beta
_{q}b_{1}...b_{s}}^{i_{1}...i_{p-1}ia_{1}...a_{r}} \\ 
\hspace*{8mm}%
-H_{j_{1}k}^{j}T_{jj_{2}...j_{q}b_{1}...b_{s}}^{i_{1}...i_{p}a_{1}...a_{r}}-...-%
\vspace*{2mm}H_{j_{q}k}^{j}T_{j_{1}...j_{q-1}jb_{1}...b_{s}}^{\alpha
_{1}...\alpha _{p}a_{1}...a_{r}} \\ 
\hspace*{8mm}+H_{ak}^{a_{1}}T_{\beta _{1}...\beta _{q}b_{1}...b_{s}}^{\alpha
_{1}...\alpha _{p}aa_{2}...a_{r}}+...+\vspace*{2mm}H_{ak}^{a_{r}}T_{\beta
_{1}...\beta _{q}b_{1}...b_{s}}^{\alpha _{1}...\alpha _{p}a_{1}...a_{r-1}a}
\\ 
\hspace*{8mm}-H_{b_{1}k}^{b}T_{\beta _{1}...\beta
_{q}bb_{2}...b_{s}}^{\alpha _{1}...\alpha _{p}a_{1}...a_{r}}-\vspace*{2mm}%
...-H_{b_{s}k}^{b}T_{\beta _{1}...\beta _{q}b_{1}...b_{s-1}b}^{\alpha
_{1}...\alpha _{p}a_{1}...a_{r}}%
\end{array}
\label{eq71}
\end{equation}%
\emph{\ and}%
\begin{equation}
\begin{array}{l}
T_{j_{1}...j_{q}b_{1}...b_{s}}^{i_{1}...i_{p}a_{1}...a_{r}}\mid _{c}=\dot{%
\partial}_{c}\left( T_{\beta _{1}...\beta _{q}b_{1}...b_{s}}^{\alpha
_{1}...\alpha _{p}a_{1}...a_{r}}\right) \vspace*{2mm} \\ 
\hspace*{8mm}%
+V_{ic}^{i_{1}}T_{j_{1}...j_{q}b_{1}...b_{s}}^{ii_{2}...i_{p}a_{1}...a_{r}}+...+V_{ic}^{i_{p}}T_{\beta _{1}...\beta _{q}b_{1}...b_{s}}^{i_{1}...i_{p-1}ia_{1}...a_{r}}%
\vspace*{2mm} \\ 
\hspace*{8mm}%
-V_{j_{1}c}^{j}T_{jj_{2}...j_{q}b_{1}...b_{s}}^{i_{1}...i_{p}a_{1}...a_{r}}-...-V_{j_{q}c}^{j}T_{j_{1}...j_{q-1}jb_{1}...b_{s}}^{i_{1}...i_{p}a_{1}...a_{r}}%
\vspace*{2mm} \\ 
\hspace*{8mm}%
+V_{ac}^{a_{1}}T_{j_{1}...j_{q}b_{1}...b_{s}}^{i_{1}...i_{p}aa_{2}...a_{r}}+...+V_{ac}^{a_{r}}T_{j_{1}...j_{q}b_{1}...b_{s}}^{i_{1}...i_{p}a_{1}...a_{r-1}a}%
\vspace*{2mm} \\ 
\hspace*{8mm}%
-V_{b_{1}c}^{b}T_{j_{1}...j_{q}bb_{2}...b_{s}}^{i_{1}...i_{p}a_{1}...a_{r}}-...-V_{b_{s}c}^{b}T_{j_{1}...j_{q}b_{1}...b_{s-1}b}^{i_{1}...i_{p}a_{1}...a_{r}}.%
\end{array}
\label{eq72}
\end{equation}

\begin{definition}
\label{d50}If $\left( E,\pi ,M\right) =\left( F,\nu ,N\right) ,$ $\left(
\rho ,\eta \right) \Gamma $ is a $\left( \rho ,\eta \right) $-connection for
the vector bundle $\left( E,\pi ,M\right) $ and 
\begin{equation*}
\left( \left( \rho ,\eta \right) H_{bc}^{a},\left( \rho ,\eta \right) \tilde{%
H}_{bc}^{a},\left( \rho ,\eta \right) V_{bc}^{a},\left( \rho ,\eta \right) 
\tilde{V}_{bc}^{a}\right)
\end{equation*}%
are the components of a distinguished linear $\left( \rho ,\eta \right) $%
\textit{-}connection for the generalized tangent bundle $\left( \left( \rho
,\eta \right) TE,\left( \rho ,\eta \right) \tau _{E},E\right) $ such that 
\begin{equation*}
\left( \rho ,\eta \right) H_{bc}^{a}=\left( \rho ,\eta \right) \tilde{H}%
_{bc}^{a}\mbox{ and }\left( \rho ,\eta \right) V_{bc}^{a}=\left( \rho ,\eta
\right) \tilde{V}_{bc}^{a},
\end{equation*}%
then we will say that \emph{the generalized tangent bundle }$\!(\!(\rho
,\!\eta )TE,(\rho ,\!\eta )\tau _{E},\!E)$ \emph{is endowed with a normal
distinguished linear }$\left( \rho ,\eta \right) $\emph{-connection }$\left(
\left( \rho ,\eta \right) H,\left( \rho ,\eta \right) V\right) $\emph{\ on
components }$\left( \left( \rho ,\eta \right) H_{bc}^{a},\left( \rho ,\eta
\right) V_{bc}^{a}\right) $.
\end{definition}

In the particular case of Lie algebroids, $\left( \eta ,h\right) =\left(
Id_{M},Id_{M}\right) ,$\emph{\ }the components of a normal distinguished
linear $\left( \rho ,Id_{M}\right) $-connection $\left( \rho H,\rho V\right) 
$ will be denoted $\left( \rho H_{bc}^{a},\rho V_{bc}^{a}\right) $.

In the classical case, $\left( \rho ,\eta ,h\right) =\left(
Id_{TE},Id_{M},Id_{M}\right) ,$\emph{\ }the components of a normal
distinguished linear $\left( Id_{TM},Id_{M}\right) $-connection $\left(
H,V\right) $ will be denoted $\left( H_{jk}^{i},V_{jk}^{i}\right) $.

\section{Mechanical systems}

\ \ 

We consider the following diagram:%
\begin{equation}
\begin{array}{rcl}
E &  & \left( E,\left[ ,\right] _{E,h},\left( \rho ,\eta \right) \right) \\ 
\pi \downarrow &  & ~\downarrow \pi \\ 
M & ^{\underrightarrow{~\ \ \ \ h~\ \ \ \ }} & ~\ M%
\end{array}
\label{eq73}
\end{equation}
where $\left( \left( E,\pi ,M\right) ,\left[ ,\right] _{E,h},\left( \rho
,\eta \right) \right) $ is a generalized Lie algebroid.

\begin{definition}
\label{d51}The triple%
\begin{equation}
\begin{array}{c}
\left( \left( E,\pi ,M\right) ,F_{e},\left( \rho ,\eta \right) \Gamma
\right) ,%
\end{array}
\label{eq74}
\end{equation}
where%
\begin{equation}
\begin{array}[t]{l}
F_{e}=F^{a}\frac{\partial }{\partial \tilde{y}^{a}}\in \Gamma \left( V\left(
\rho ,\eta \right) TE,\left( \rho ,\eta \right) \tau _{E},E\right)%
\end{array}
\label{eq75}
\end{equation}%
is an external force and $\left( \rho ,\eta \right) \Gamma $ is a $\left(
\rho ,\eta \right) $-connection for the vector bundle $\left( E,\pi
,M\right) $, will be called \emph{mechanical }$\left( \rho ,\eta \right) $%
\emph{-system.}
\end{definition}

\begin{definition}
\label{d52}A smooth \emph{Lagrange fundamental function} on the vector
bundle\break $\left( E,\pi ,M\right) $ is a mapping $E~\ ^{\underrightarrow{%
\ \ L\ \ }}~\ \mathbb{R}$ which satisfies the following conditions:\medskip
\end{definition}

1. $L\circ u\in C^{\infty }\left( M\right) $, for any $u\in \Gamma \left(
E,\pi ,M\right) \setminus \left\{ 0\right\} $;\smallskip

2. $L\circ 0\in C^{0}\left( M\right) $, where $0$ means the null section of $%
\left( E,\pi ,M\right) .$\medskip

Let $L$ be a Lagrangian defined on the total space of the vector bundle $%
\left( E,\pi ,M\right) .$

If $\left( U,s_{U}\right) $ is a local vector $\left( m+r\right) $-chart for 
$\left( E,\pi ,M\right) $, then we obtain the following real functions
defined on $\pi ^{-1}\left( U\right) $:%
\begin{equation}
\begin{array}{cc}
L_{i}\overset{put}{=}\displaystyle\frac{\partial L}{\partial x^{i}}\overset{%
put}{=}\frac{\partial }{\partial x^{i}}\left( L\right) & L_{ib}\overset{put}{%
=}\displaystyle\frac{\partial ^{2}L}{\partial x^{i}\partial y^{b}}\vspace*{%
2mm}\overset{put}{=}\frac{\partial }{\partial x^{i}}\left( \frac{\partial }{%
\partial y^{b}}\left( L\right) \right) \\ 
L_{a}\overset{put}{=}\displaystyle\frac{\partial L}{\partial y^{a}}\overset{%
put}{=}\frac{\partial }{\partial y^{a}}\left( L\right) & L_{ab}\overset{put}{%
=}\displaystyle\frac{\partial ^{2}L}{\partial y^{a}\partial y^{b}}\overset{%
put}{=}\frac{\partial }{\partial y^{a}}\left( \frac{\partial }{\partial y^{b}%
}\left( L\right) \right) .%
\end{array}
\label{eq76}
\end{equation}

\begin{definition}
\label{d53}If for any vector local $m+r$-chart $\left( U,s_{U}\right) $ of $%
\left( E,\pi ,M\right) ,$ we have:%
\begin{equation}
\begin{array}{c}
rank\left\Vert L_{ab}\left( u_{x}\right) \right\Vert =r,%
\end{array}
\label{eq77}
\end{equation}
for any $u_{x}\in \pi ^{-1}\left( U\right) \backslash \left\{ 0_{x}\right\} $%
, then we will say that \emph{the Lagrangian }$L$\emph{\ is regular.}
\end{definition}

\begin{proposition}
\label{p54}\emph{If the Lagrangian }$L$\emph{\ is regular, then for any
vector local }$m+r$\emph{-chart }$\left( U,s_{U}\right) $\emph{\ of }$\left(
E,\pi ,M\right) ,$\emph{\ we obtain the real functions }$\tilde{L}^{ab}$%
\emph{\ locally defined by}%
\begin{equation}
\begin{array}{ccc}
\pi ^{-1}\left( U\right) & ^{\underrightarrow{\ \ \tilde{L}^{ab}\ \ }} & 
\mathbb{R} \\ 
u_{x} & \longmapsto & \tilde{L}^{ab}\left( u_{x}\right)%
\end{array}
\label{eq78}
\end{equation}%
\emph{where }$\left\Vert \tilde{L}^{ab}\left( u_{x}\right) \right\Vert
=\left\Vert L_{ba}\left( u_{x}\right) \right\Vert ^{-1}$\emph{, for any }$%
u_{x}\in \pi ^{-1}\left( U\right) \backslash \left\{ 0_{x}\right\} .$
\end{proposition}

\begin{definition}
\label{d55}A smooth \emph{Finsler fundamental function} on the vector bundle 
$\left( E,\pi ,M\right) $ is a mapping $%
\begin{array}[b]{ccc}
E & ^{\underrightarrow{\ F\ }} & \mathbb{R}_{+}%
\end{array}%
$ which satisfies the following conditions:\medskip
\end{definition}

1. $F\circ u\in C^{\infty }\left( M\right) $, for any $u\in \Gamma \left(
E,\pi ,M\right) \setminus \left\{ 0\right\} $;\smallskip

2. $F\circ 0\in C^{0}\left( M\right) $, where $0$ means the null section of $%
\left( E,\pi ,M\right) $;\smallskip

3. $F$ is positively $1$-homogenous on the fibres of vector bundle $\left(
E,\pi ,M\right) ;$\smallskip

4. For any vector local $m+r$-chart $\left( U,s_{U}\right) $ of $\left(
E,\pi ,M\right) ,$ the hessian:%
\begin{equation}
\left\Vert F_{~ab}^{2}\left( u_{x}\right) \right\Vert  \label{eq79}
\end{equation}%
is positively define for any $u_{x}\in \pi ^{-1}\left( U\right) \backslash
\left\{ 0_{x}\right\} $.

\begin{definition}
\label{d56}If $L$ respectively $F$ is a smooth Lagrange respectively Finsler
function, then the triple%
\begin{equation}
\begin{array}[t]{l}
\left( \left( E,\pi ,M\right) ,F_{e},L\right) 
\end{array}
\label{eq80}
\end{equation}%
respectively%
\begin{equation}
\begin{array}[t]{l}
\left( \left( E,\pi ,M\right) ,F_{e},F\right) 
\end{array}
\label{eq81}
\end{equation}%
where $F_{e}=F^{a}\displaystyle\frac{\partial }{\partial \tilde{y}^{a}}\in
\Gamma \left( V\left( \rho ,\eta \right) TE,\left( \rho ,\eta \right) \tau
_{E},E\right) $ is an external force, is called \emph{Lagrange mechanical }$%
\left( \rho ,\eta \right) $\emph{-system }and \emph{Finsler mechanical }$%
\left( \rho ,\eta \right) $\emph{-system, respectively}.
\end{definition}

\begin{definition}
Any Lagrange mechanical $\left( \rho ,Id_{M}\right) $-system and any Finsler
mechanical\break $\left( \rho ,Id_{M}\right) $-system will be called \emph{%
Lagrange mechanical }$\rho $-\emph{system }and \emph{Finsler mechanical }$%
\rho $-\emph{system}, respectively.
\end{definition}

Any Lagrange mechanical $\left( Id_{TM},Id_{M}\right) $-system and any
Finsler mechanical\break $\left( Id_{TM},Id_{M}\right) $-system will be
called \emph{Lagrange mechanical system }and \emph{Finsler mechanical system}%
, respectively.

\section{ $(\protect\rho ,\protect\eta )$-semisprays and $(\protect\rho ,%
\protect\eta )$-sprays for mechanical $(\protect\rho ,\protect\eta )$-systems%
}

\ \ \ 

Let $(\left( E,\pi ,M\right) ,F_{e},(\rho ,\eta )\Gamma )$ be an arbitrary
mechanical $\left( \rho ,\eta \right) $-system.

\begin{definition}
\label{d57}The\textit{\ }vertical section $\mathbb{C}\mathbf{=}y^{a}\overset{%
\cdot }{\tilde{\partial}}_{a}$will be called the\textit{\ }\emph{Liouville
section.}
\end{definition}

A section $S\in \Gamma \left( \left( \rho ,\eta \right) TE,\left( \rho ,\eta
\right) \tau _{E},E\right) $\ will be called $\left( \rho ,\eta \right) $%
\emph{-semispray}\ if there exists an almost tangent structure $e$ such that 
$e\left( S\right) =\mathbb{C}.$

Let $g\in \mathbf{Man}\left( E,E\right) $ be such that $\left( g,h\right) $
is a locally invertible $\mathbf{B}^{\mathbf{v}}$-morphism of $\left( E,\pi
,M\right) $\ source and\ target.

\begin{theorem}
\label{t58}\emph{The section}%
\begin{equation}
\begin{array}[t]{l}
S=\left( g_{b}^{a}\circ h\circ \pi \right) y^{b}\tilde{\partial}_{a}-2\left(
G^{a}-\frac{1}{4}F^{a}\right) \overset{\cdot }{\tilde{\partial}}_{a}%
\end{array}
\label{eq82}
\end{equation}%
\emph{\ is a }$\left( \rho ,\eta \right) $\emph{-semispray such that the
real local functions }$G^{a},\ a\in \overline{1,n},$\emph{\ satisfy the
following conditions}%
\begin{equation}
\begin{array}{cl}
\left( \rho ,\eta \right) \Gamma _{c}^{a} & =\left( \tilde{g}_{c}^{b}\circ
h\circ \pi \right) \frac{\partial \left( G^{a}-\frac{1}{4}F^{a}\right) }{%
\partial y^{b}} \\ 
& -\frac{1}{2}\left( g_{e}^{d}\circ h\circ \pi \right) y^{e}L_{dc}^{f}\left( 
\tilde{g}_{f}^{a}\circ h\circ \pi \right) \\ 
& +\frac{1}{2}\left( \rho _{c}^{j}\circ h\circ \pi \right) \frac{\partial
\left( g_{e}^{b}\circ h\circ \pi \right) }{\partial x^{j}}y^{e}\left( \tilde{%
g}_{b}^{a}\circ h\circ \pi \right) \\ 
& -\frac{1}{2}\left( g_{e}^{b}\circ h\circ \pi \right) y^{e}\left( \rho
_{b}^{i}\circ h\circ \pi \right) \frac{\partial \left( \tilde{g}%
_{c}^{a}\circ h\circ \pi \right) }{\partial x^{i}}%
\end{array}
\label{eq83}
\end{equation}%
\emph{In addition, we remark that the local real functions}%
\begin{equation}
\begin{array}{cl}
\left( \rho ,\eta \right) \mathring{\Gamma}_{c}^{a} & =\left( \tilde{g}%
_{c}^{b}\circ h\circ \pi \right) \frac{\partial G^{a}}{\partial y^{b}} \\ 
& -\frac{1}{2}\left( g_{e}^{d}\circ h\circ \pi \right) y^{e}L_{dc}^{f}\left( 
\tilde{g}_{f}^{a}\circ h\circ \pi \right) \\ 
& +\frac{1}{2}\left( \rho _{c}^{j}\circ h\circ \pi \right) \frac{\partial
\left( g_{e}^{b}\circ h\circ \pi \right) }{\partial x^{j}}y^{e}\left( \tilde{%
g}_{b}^{a}\circ h\circ \pi \right) \\ 
& -\frac{1}{2}\left( g_{e}^{b}\circ h\circ \pi \right) y^{e}\left( \rho
_{b}^{i}\circ h\circ \pi \right) \frac{\partial \left( \tilde{g}%
_{c}^{a}\circ h\circ \pi \right) }{\partial x^{i}}%
\end{array}
\label{eq84}
\end{equation}%
\emph{are the components of a }$\left( \rho ,\eta \right) $\emph{-connection 
}$\left( \rho ,\eta \right) \mathring{\Gamma}$\emph{\ for the vector bundle }%
$\left( E,\pi ,M\right) .$

The $\left( \rho ,\eta \right) $-semispray $S$\ will be called \emph{the\
canonical }$\left( \rho ,\eta \right) $\emph{-semispray associated to
mechanical }$\left( \rho ,\eta \right) $\emph{-system }$\left( \left( E,\pi
,M\right) ,F_{e},\left( \rho ,\eta \right) \Gamma \right) $\emph{\ and from
locally invertible }$\mathbf{B}^{\mathbf{v}}$\emph{-morphism }$\left(
g,h\right) .$
\end{theorem}

\begin{proof}
We consider the $\mathbf{Mod}$-endomorphism%
\begin{equation*}
\begin{array}{rcl}
\Gamma \left( \left( \rho ,\eta \right) TE,\left( \rho ,\eta \right) \tau
_{E},E\right) & ^{\underrightarrow{\ \ \mathbb{P}\ \ }} & \Gamma \left(
\left( \rho ,\eta \right) TE,\left( \rho ,\eta \right) \tau _{E},E\right) 
\vspace*{1mm} \\ 
X & \longmapsto & \mathcal{J}_{\left( g,h\right) }\left[ S,X\right] _{\left(
\rho ,\eta \right) TE}-\left[ S,\mathcal{J}_{\left( g,h\right) }X\right]
_{\left( \rho ,\eta \right) TE}.%
\end{array}%
\end{equation*}

Let $X=Z^{a}\tilde{\partial}_{a}+Y^{a}\overset{\cdot }{\tilde{\partial}}_{a}$
be an arbitrary section. Since 
\begin{equation*}
\begin{array}{cl}
\left[ S,X\right] _{\left( \rho ,\eta \right) TE} & =\displaystyle\left[
\left( g_{e}^{a}\circ h\circ \pi \cdot y^{e}\right) \tilde{\partial}%
_{a},Z^{b}\tilde{\partial}_{b}\right] _{\left( \rho ,\eta \right) TE}+\left[
\left( g_{e}^{a}\circ h\circ \pi \cdot y^{e}\right) \tilde{\partial}%
_{a},Y^{b}\overset{\cdot }{\tilde{\partial}}_{b}\right] _{\left( \rho ,\eta
\right) TE}\vspace*{2mm} \\ 
& \displaystyle-\left[ 2\left( G^{a}-\frac{1}{4}F^{a}\right) \overset{\cdot }%
{\tilde{\partial}}_{a},Z^{b}\tilde{\partial}_{b}\right] _{\left( \rho ,\eta
\right) TE}-\left[ 2\left( G^{a}-\frac{1}{4}F^{a}\right) \overset{\cdot }{%
\tilde{\partial}}_{a},Y^{b}\overset{\cdot }{\tilde{\partial}}_{b}\right]
_{\left( \rho ,\eta \right) TE}%
\end{array}%
\end{equation*}%
and 
\begin{equation*}
\begin{array}{cl}
\left[ \left( g_{e}^{a}\circ h\circ \pi \cdot y^{e}\right) \tilde{\partial}%
_{a},Z^{b}\tilde{\partial}_{b}\right] _{\left( \rho ,\eta \right) TE} & =%
\displaystyle\left( g_{e}^{a}\circ h\circ \pi \right) y^{e}\left( \rho
_{a}^{i}\circ h\circ \pi \right) \frac{\partial Z^{c}}{\partial x^{i}}\tilde{%
\partial}_{c}\vspace*{2mm} \\ 
& -\displaystyle Z^{b}\left( \rho _{b}^{j}\circ h\circ \pi \right) \frac{%
\partial \left( g_{e}^{c}\circ h\circ \pi \right) }{\partial x^{j}}y^{e}%
\tilde{\partial}_{c}\vspace*{2mm} \\ 
& \displaystyle+\left( g_{e}^{a}\circ h\circ \pi \right) y^{e}Z^{b}L_{ab}^{c}%
\tilde{\partial}_{c},%
\end{array}%
\end{equation*}%
\begin{equation*}
\begin{array}{cl}
\left[ \left( g_{e}^{a}\circ h\circ \pi \cdot y^{e}\right) \tilde{\partial}%
_{a},Y^{b}\overset{\cdot }{\tilde{\partial}}_{b}\right] _{\left( \rho ,\eta
\right) TE} & =\displaystyle\left( g_{e}^{a}\circ h\circ \pi \right)
y^{e}\left( \rho _{a}^{i}\circ h\circ \pi \right) \frac{\partial Y^{c}}{%
\partial x^{i}}\overset{\cdot }{\tilde{\partial}}_{c}\vspace*{2mm} \\ 
& \displaystyle-Y^{b}g_{b}^{c}\tilde{\partial}_{c},%
\end{array}%
\end{equation*}%
\begin{equation*}
\begin{array}{cl}
\displaystyle\left[ 2\left( G^{a}-\frac{1}{4}F^{a}\right) \overset{\cdot }{%
\tilde{\partial}}_{a},Z^{b}\tilde{\partial}_{b}\right] _{\left( \rho ,\eta
\right) TE} & \displaystyle=2\left( G^{a}-\frac{1}{4}F^{a}\right) \frac{%
\partial Z^{c}}{\partial y^{a}}\tilde{\partial}_{c}\vspace*{2mm} \\ 
& \displaystyle-2Z^{b}\rho _{b}^{j}\circ h\circ \pi \frac{\partial \left(
G^{c}-\frac{1}{4}F^{c}\right) }{\partial x^{j}}\overset{\cdot }{\tilde{%
\partial}}_{c},%
\end{array}%
\end{equation*}%
\begin{equation*}
\begin{array}{c}
\left[ 2\left( G^{a}-\frac{1}{4}F^{a}\right) \overset{\cdot }{\tilde{\partial%
}}_{a},Y^{b}\overset{\cdot }{\tilde{\partial}}_{b}\right] _{\left( \rho
,\eta \right) TE}=2\left( G^{a}-\frac{1}{4}F^{a}\right) \frac{\partial Y^{c}%
}{\partial y^{a}}\overset{\cdot }{\tilde{\partial}}_{c}-2Y^{b}\displaystyle%
\frac{\partial \left( G^{c}-\displaystyle\frac{1}{4}F^{c}\right) }{\partial
y^{b}}\overset{\cdot }{\tilde{\partial}}_{c},%
\end{array}%
\end{equation*}%
it results that 
\begin{equation*}
\begin{array}{cl}
\mathcal{J}_{\left( g,h\right) }\left[ S,X\right] _{\left( \rho ,\eta
\right) TE} & \displaystyle=\left( g_{e}^{a}\circ h\circ \pi \right)
y^{e}\left( \rho _{a}^{i}\circ h\circ \pi \right) \frac{\partial Z^{c}}{%
\partial x^{i}}\left( \tilde{g}_{c}^{d}\circ h\circ \pi \right) \overset{%
\cdot }{\tilde{\partial}}_{d}\vspace*{2mm} \\ 
& \displaystyle-Z^{b}\left( \rho _{b}^{j}\circ h\circ \pi \right) \frac{%
\partial \left( g_{e}^{c}\circ h\circ \pi \right) }{\partial x^{j}}%
y^{e}\left( \tilde{g}_{c}^{d}\circ h\circ \pi \right) \overset{\cdot }{%
\tilde{\partial}}_{d}\vspace*{2mm} \\ 
& \displaystyle+\left( g_{e}^{a}\circ h\circ \pi \right)
y^{e}Z^{b}L_{ab}^{c}\left( \tilde{g}_{c}^{d}\circ h\circ \pi \right) \overset%
{\cdot }{\tilde{\partial}}_{d}\vspace*{2mm}\displaystyle-Y^{d}\overset{\cdot 
}{\tilde{\partial}}_{d}\vspace*{2mm} \\ 
& \displaystyle-2\left( G^{a}-\frac{1}{4}F^{a}\right) \frac{\partial Z^{c}}{%
\partial y^{a}}\left( \tilde{g}_{c}^{d}\circ h\circ \pi \right) \overset{%
\cdot }{\tilde{\partial}}_{d}.%
\end{array}%
\eqno(P2)
\end{equation*}

Since 
\begin{equation*}
\begin{array}{cl}
\left[ S,\mathcal{J}_{\left( g,h\right) }X\right] _{\left( \rho ,\eta
\right) TE} & \displaystyle=\left[ \left( g_{e}^{a}\circ h\circ \pi \right)
y^{e}\tilde{\partial}_{a},Z^{b}\left( \tilde{g}_{b}^{c}\circ h\circ \pi
\right) \overset{\cdot }{\tilde{\partial}}_{c}\right] _{\left( \rho ,\eta
\right) TE}\vspace*{2mm} \\ 
& \displaystyle-\left[ 2\left( G^{a}-\frac{1}{4}F^{a}\right) \overset{\cdot }%
{\tilde{\partial}}_{a},Z^{b}\left( \tilde{g}_{b}^{c}\circ h\circ \pi \right) 
\overset{\cdot }{\tilde{\partial}}_{c}\right] _{\left( \rho ,\eta \right) TE}%
\end{array}%
\end{equation*}%
and 
\begin{equation*}
\begin{array}{cl}
\left[ \left( g_{e}^{a}\circ h\circ \pi \right) y^{e}\tilde{\partial}%
_{a},Z^{b}\left( \tilde{g}_{b}^{c}\circ h\circ \pi \right) \overset{\cdot }{%
\tilde{\partial}}_{c}\right] _{\left( \rho ,\eta \right) TE} & \displaystyle%
=-Z^{d}\tilde{\partial}_{d}+\left( g_{e}^{a}\circ h\circ \pi \right)
y^{e}\left( \rho _{a}^{i}\circ h\circ \pi \right) \frac{\partial Z^{b}}{%
\partial x^{i}}\left( \tilde{g}_{b}^{d}\circ h\circ \pi \right) \overset{%
\cdot }{\tilde{\partial}}_{d}\vspace*{2mm} \\ 
& \displaystyle-\left( g_{e}^{a}\circ h\circ \pi \right) y^{e}\left( \rho
_{a}^{i}\circ h\circ \pi \right) Z^{b}\frac{\partial \left( \tilde{g}%
_{b}^{d}\circ h\circ \pi \right) }{\partial x^{i}}\overset{\cdot }{\tilde{%
\partial}}_{d},%
\end{array}%
\end{equation*}%
\begin{equation*}
\begin{array}{cl}
\displaystyle\left[ 2\left( G^{a}-\frac{1}{4}F^{a}\right) \overset{\cdot }{%
\tilde{\partial}}_{a},Z^{b}\left( \tilde{g}_{b}^{c}\circ h\circ \pi \right) 
\overset{\cdot }{\tilde{\partial}}_{c}\right] _{\left( \rho ,\eta \right) TE}
& \displaystyle=2\left( G^{a}-\frac{1}{4}F^{a}\right) \frac{\partial Z^{b}}{%
\partial y^{a}}\left( \tilde{g}_{b}^{d}\circ h\circ \pi \right) \overset{%
\cdot }{\tilde{\partial}}_{d}\vspace*{2mm} \\ 
& \displaystyle-Z^{b}\left( \tilde{g}_{b}^{c}\circ h\circ \pi \right) \frac{%
\partial 2\left( G^{d}-\frac{1}{4}F^{d}\right) }{\partial y^{c}}\overset{%
\cdot }{\tilde{\partial}}_{d}%
\end{array}%
\end{equation*}%
it results that%
\begin{equation*}
\begin{array}{cl}
\left[ S,\mathcal{J}_{\left( g,h\right) }X\right] _{\left( \rho ,\eta
\right) TE} & \displaystyle=-Z^{d}\tilde{\partial}_{d}+\left( g_{e}^{a}\circ
h\circ \pi \right) y^{e}\left( \rho _{a}^{i}\circ h\circ \pi \right) \frac{%
\partial Z^{b}}{\partial x^{i}}\left( \tilde{g}_{b}^{d}\circ h\circ \pi
\right) \overset{\cdot }{\tilde{\partial}}_{d} \\ 
& \displaystyle-\left( g_{e}^{a}\circ h\circ \pi \right) y^{e}\left( \rho
_{a}^{i}\circ h\circ \pi \right) Z^{b}\frac{\partial \left( \tilde{g}%
_{b}^{d}\circ h\circ \pi \right) }{\partial x^{i}}\overset{\cdot }{\tilde{%
\partial}}_{d} \\ 
& \displaystyle-2\left( G^{a}-\frac{1}{4}F^{a}\right) \frac{\partial Z^{b}}{%
\partial y^{a}}\left( \tilde{g}_{b}^{d}\circ h\circ \pi \right) \overset{%
\cdot }{\tilde{\partial}}_{d}\vspace*{2mm} \\ 
& \displaystyle+Z^{b}\left( \tilde{g}_{b}^{c}\circ h\circ \pi \right) \frac{%
\partial 2\left( G^{d}-\frac{1}{4}F^{d}\right) }{\partial y^{c}}\overset{%
\cdot }{\tilde{\partial}}_{d}.%
\end{array}%
\eqno(P2)
\end{equation*}

Using equalities $\left( P_{1}\right) $ and $\left( P_{2}\right) $, we
obtain:%
\begin{equation*}
\begin{array}[b]{cl}
\mathbb{P}\left( Z^{a}\tilde{\partial}_{a}+Y^{a}\overset{\cdot }{\tilde{%
\partial}}_{a}\right) & =Z^{a}\tilde{\partial}_{a}-Y^{d}\overset{\cdot }{%
\tilde{\partial}}_{d}+\left( g_{e}^{a}\circ h\circ \pi \right)
y^{e}Z^{b}\left( L_{ab}^{c}\circ h\circ \pi \right) \left( \tilde{g}%
_{c}^{d}\circ h\circ \pi \right) \overset{\cdot }{\tilde{\partial}}_{d} \\ 
& -Z^{b}\left( \rho _{b}^{j}\circ h\circ \pi \right) \frac{\partial \left(
g_{e}^{c}\circ h\circ \pi \right) }{\partial x^{j}}y^{e}\left( \tilde{g}%
_{c}^{d}\circ h\circ \pi \right) \overset{\cdot }{\tilde{\partial}}_{d} \\ 
& +\left( g_{e}^{a}\circ h\circ \pi \right) y^{e}\left( \rho _{a}^{i}\circ
h\circ \pi \right) Z^{b}\frac{\partial \left( \tilde{g}_{b}^{d}\circ h\circ
\pi \right) }{\partial x^{i}}\overset{\cdot }{\tilde{\partial}}_{d} \\ 
& -Z^{b}\left( \tilde{g}_{b}^{c}\circ h\circ \pi \right) \frac{\partial
2\left( G^{d}-\frac{1}{4}F^{d}\right) }{\partial y^{c}}\overset{\cdot }{%
\tilde{\partial}}_{d}%
\end{array}%
\end{equation*}

After some calculations, it results that $\mathbb{P}$ is an almost product
structure.

Using the equalities \eqref{eq50} and \eqref{eq56} it results that 
\begin{equation*}
\mathbb{P}\left( Z^{a}\tilde{\partial}_{a}+Y^{a}\overset{\cdot }{\tilde{%
\partial}}_{a}\right) =\left( Id-2\left( \rho ,\eta \right) \Gamma \right)
\left( Z^{a}\tilde{\partial}_{a}+Y^{a}\overset{\cdot }{\tilde{\partial}}%
_{a}\right) ,
\end{equation*}%
for any $Z^{a}\tilde{\partial}_{a}+Y^{a}\overset{\cdot }{\tilde{\partial}}%
_{a}\in \Gamma \left( \left( \rho ,\eta \right) TE,\left( \rho ,\eta \right)
\tau _{E},E\right) $ and we obtain%
\begin{equation*}
\begin{array}[b]{cl}
\left( \rho ,\eta \right) \Gamma \left( Z^{a}\tilde{\partial}_{a}+Y^{a}%
\overset{\cdot }{\tilde{\partial}}_{a}\right) & =Y^{d}\overset{\cdot }{%
\tilde{\partial}}_{d}-\frac{1}{2}\left( g_{e}^{a}\circ h\circ \pi \right)
y^{e}Z^{b}\left( L_{ab}^{c}\circ h\circ \pi \right) \left( \tilde{g}%
_{c}^{d}\circ h\circ \pi \right) \overset{\cdot }{\tilde{\partial}}_{d} \\ 
& +\frac{1}{2}Z^{b}\left( \rho _{b}^{j}\circ h\circ \pi \right) \frac{%
\partial \left( g_{e}^{c}\circ h\circ \pi \right) }{\partial x^{j}}%
y^{e}\left( \tilde{g}_{c}^{d}\circ h\circ \pi \right) \overset{\cdot }{%
\tilde{\partial}}_{d} \\ 
& -\frac{1}{2}\left( g_{e}^{a}\circ h\circ \pi \right) y^{e}\left( \rho
_{a}^{i}\circ h\circ \pi \right) Z^{b}\frac{\partial \left( \tilde{g}%
_{b}^{d}\circ h\circ \pi \right) }{\partial x^{i}}\overset{\cdot }{\tilde{%
\partial}}_{d} \\ 
& +Z^{b}\left( \tilde{g}_{b}^{c}\circ h\circ \pi \right) \frac{\partial
\left( G^{d}-\frac{1}{4}F^{d}\right) }{\partial y^{c}}\overset{\cdot }{%
\tilde{\partial}}_{d}.%
\end{array}%
\end{equation*}%
Since 
\begin{equation*}
\begin{array}{c}
\left( \rho ,\eta \right) \Gamma \left( Z^{a}\tilde{\partial}_{a}+Y^{a}%
\overset{\cdot }{\tilde{\partial}}_{a}\right) =\left( Y^{d}+\left( \rho
,\eta \right) \Gamma _{b}^{d}Z^{b}\right) \overset{\cdot }{\tilde{\partial}}%
_{d}%
\end{array}%
\end{equation*}%
it results the relations \eqref{eq83}. In addition, since 
\begin{equation*}
\left( \rho ,\eta \right) \mathring{\Gamma}_{c}^{a}=\left( \rho ,\eta
\right) \Gamma _{c}^{a}+\frac{1}{4}\tilde{g}_{c}^{d}\circ h\circ \pi \frac{%
\partial F^{a}}{\partial y^{d}}
\end{equation*}%
and 
\begin{equation*}
\begin{array}{ll}
\left( \rho ,\eta \right) \mathring{\Gamma}_{c%
{\acute{}}%
}^{a%
{\acute{}}%
} & \displaystyle=\left( \rho ,\eta \right) \Gamma _{c%
{\acute{}}%
}^{a%
{\acute{}}%
}+\frac{1}{2}\tilde{g}_{c%
{\acute{}}%
}^{b%
{\acute{}}%
}\circ h\circ \pi \displaystyle\frac{\partial F^{a%
{\acute{}}%
}}{\partial y^{b%
{\acute{}}%
}}\vspace*{2mm} \\ 
& \displaystyle=M_{a}^{a%
{\acute{}}%
}\circ \pi \left( \rho _{c}^{i}\circ h\circ \pi \cdot \frac{\partial M_{b%
{\acute{}}%
}^{a}}{\partial x^{i}}y^{b%
{\acute{}}%
}+\left( \rho ,\eta \right) \Gamma _{c}^{a}\right) M_{c%
{\acute{}}%
}^{c}\circ h\circ \pi \vspace*{2mm} \\ 
& \displaystyle+M_{a}^{a%
{\acute{}}%
}\circ \pi \left( \frac{1}{4}\tilde{g}_{c}^{b}\circ h\circ \pi \cdot \frac{%
\partial F^{a}}{\partial y^{b}}\right) M_{c%
{\acute{}}%
}^{c}\circ h\circ \pi \vspace*{2mm} \\ 
& \displaystyle=M_{a}^{a%
{\acute{}}%
}{\circ }\pi \left( \rho _{c}^{i}{\circ }h{\circ }\pi \cdot \frac{\partial
M_{b%
{\acute{}}%
}^{a}}{\partial x^{i}}y^{b%
{\acute{}}%
}+\left( \left( \rho ,\eta \right) \Gamma _{c}^{a}+\frac{1}{4}\tilde{g}%
_{c}^{b}{\circ h\circ }\pi \cdot \frac{\partial F^{a}}{\partial y^{b}}%
\right) \right) M_{c%
{\acute{}}%
}^{c}{\circ }h{\circ }\pi \vspace*{2mm} \\ 
& \displaystyle=M_{a}^{a%
{\acute{}}%
}{\circ }\pi \left( \rho _{c}^{i}{\circ }h{\circ }\pi \cdot \frac{\partial
M_{b%
{\acute{}}%
}^{a}}{\partial x^{i}}y^{b%
{\acute{}}%
}+\left( \rho ,\eta \right) \mathring{\Gamma}_{c}^{a}\right) M_{c%
{\acute{}}%
}^{c}{\circ }h{\circ }\pi%
\end{array}%
\end{equation*}%
it results the conclusion of the theorem.
\end{proof}

\begin{remark}
\label{r59}I\textrm{f }$\eta =Id_{M}$\textrm{, }$\left( g,h\right) =\left(
Id_{E},Id_{M}\right) $\textrm{\ and }$F_{e}=0$\textrm{, then we obtain\ the
canonical semispray associated to }$\rho $\textrm{-connection }$\rho \Gamma $%
\textrm{\ presented in \cite{3}, pp131-132 and \cite{4}.}

\textrm{If }$\left( \rho ,\eta \right) =\left( Id_{TM},Id_{M}\right) $%
\textrm{, }$\left( g,h\right) =\left( Id_{E},Id_{M}\right) $\textrm{\ and }$%
F_{e}\neq 0$\textrm{, then we obtain\ the canonical semispray associated to
connection }$\Gamma $\textrm{\ presented in \cite{5}.}

\textrm{In particular, if }$\left( \rho ,\eta \right) =\left(
Id_{TM},Id_{M}\right) $\textrm{, }$\left( g,h\right) =\left(
Id_{E},Id_{M}\right) $\textrm{, and }$F_{e}=0$\textrm{, then we obtain the
classical canonical semispray associated to connection }$\Gamma $\textrm{.}
\end{remark}

Using Theorem \textrm{\ref{t58}}, we obtain the following:

\begin{theorem}
\label{t60}\emph{The following properties hold good:}\medskip

$1^{\circ }$\ \emph{Since } $\overset{\circ }{\tilde{\delta}}_{c}=\tilde{%
\partial}_{c}-\left( \rho ,\eta \right) \mathring{\Gamma}_{c}^{a}\overset{%
\cdot }{\tilde{\partial}}_{a},~c\in \overline{1,r},$ \emph{it results that}%
\begin{equation}
\begin{array}[t]{l}
\overset{\circ }{\tilde{\delta}}_{c}=\tilde{\delta}_{c}-\frac{1}{4}\tilde{g}%
_{c}^{b}\circ h\circ \pi \cdot \frac{\partial F^{a}}{\partial y^{b}}\overset{%
\cdot }{\tilde{\partial}}_{a},~c\in \overline{1,r}.%
\end{array}
\label{eq85}
\end{equation}%
\emph{\ }

$2^{\circ }\ $\emph{Since} $\mathring{\delta}\tilde{y}^{a}=\left( \rho ,\eta
\right) \mathring{\Gamma}_{c}^{a}d\tilde{z}^{c}+d\tilde{y}^{a},$ \emph{it
results that}%
\begin{equation}
\begin{array}[t]{l}
\mathring{\delta}\tilde{y}^{a}=\delta \tilde{y}^{a}+\frac{1}{4}\tilde{g}%
_{c}^{b}\circ h\circ \pi \frac{\partial F^{a}}{\partial y^{b}}d\tilde{z}%
^{c},~a\in \overline{1,r}.%
\end{array}
\label{eq86}
\end{equation}
\end{theorem}

\emph{\ \ \ }

\begin{theorem}
\label{t61}\emph{The real local functions}%
\begin{equation}
\begin{array}[t]{l}
\left( \frac{\partial \left( \rho ,\eta \right) \Gamma _{c}^{a}}{\partial
y^{b}},\frac{\partial \left( \rho ,\eta \right) \Gamma _{c}^{a}}{\partial
y^{b}},0,~0\right) ,~a,b,c\in \overline{1,r},%
\end{array}
\label{eq87}
\end{equation}%
\emph{and}%
\begin{equation}
\begin{array}[t]{l}
\left( \frac{\partial \left( \rho ,\eta \right) \mathring{\Gamma}_{c}^{a}}{%
\partial y^{b}},\frac{\partial \left( \rho ,\eta \right) \mathring{\Gamma}%
_{c}^{a}}{\partial y^{b}},0,~0\right) ,~a,b,c\in \overline{1,r},%
\end{array}
\label{eq88}
\end{equation}%
\emph{respectively, are the coefficients to a Berwald linear }$\left( \rho
,\eta \right) $\emph{-connection for the generalized tangent bundle }$\left(
\left( \rho ,\eta \right) TE,\left( \rho ,\eta \right) \tau _{E},E\right) $.
\end{theorem}

\begin{theorem}
\label{t62}\emph{The tensor of integrability of the }$\left( \rho ,\eta
\right) $\emph{-connection }$\left( \rho ,\eta \right) \mathring{\Gamma}$%
\emph{\ is as follows:}%
\begin{equation}
\begin{array}{l}
\displaystyle\left( \rho ,\eta ,h\right) \mathbb{\mathring{R}}%
_{cd}^{a}=\left( \rho ,\eta ,h\right) \mathbb{R}_{cd}^{a}+\frac{1}{4}\left( 
\tilde{g}_{d}^{e}\circ h\circ \pi \frac{\partial F^{a}}{\partial y^{e}}_{|c}-%
\tilde{g}_{c}^{e}\circ h\circ \pi \frac{\partial F^{a}}{\partial y^{e}}%
_{|d}\right) \vspace*{1mm} \\ 
\displaystyle+\frac{1}{16}\left( \tilde{g}_{d}^{e}\circ h\circ \pi \frac{%
\partial F^{b}}{\partial y^{e}}\tilde{g}_{c}^{f}\circ h\circ \pi \frac{%
\partial ^{2}F^{a}}{\partial y^{b}\partial y^{f}}-\tilde{g}_{c}^{f}\circ
h\circ \pi \frac{\partial F^{b}}{\partial y^{f}}\tilde{g}_{d}^{e}\circ
h\circ \pi \frac{\partial ^{2}F^{a}}{\partial y^{b}\partial y^{e}}\right) 
\vspace*{1mm} \\ 
\displaystyle+\frac{1}{4}\left( L_{cd}^{f}\circ h\circ \pi \right) \left( 
\tilde{g}_{f}^{e}\circ h\circ \pi \right) \frac{\partial F^{a}}{\partial
y^{e}},%
\end{array}
\label{eq89}
\end{equation}%
\emph{where }$_{|c}$\emph{\ is the }$h$\emph{-covariant derivation with
respect to the normal Berwald linear }$\left( \rho ,\eta \right) $\emph{%
-connection }\eqref{eq87}\emph{.}
\end{theorem}

\begin{proof}
Since 
\begin{equation*}
\begin{array}{cl}
\left( \rho ,\eta ,h\right) \mathbb{\mathring{R}}_{cd}^{a}= & \Gamma \left( 
\tilde{\rho},Id_{E}\right) \left( \overset{\circ }{\tilde{\delta}}%
_{c}\right) \left( \left( \rho ,\eta \right) \mathring{\Gamma}%
_{d}^{a}\right) -\Gamma \left( \tilde{\rho},Id_{E}\right) \left( \overset{%
\circ }{\tilde{\delta}}_{d}\right) \left( \left( \rho ,\eta \right) 
\mathring{\Gamma}_{c}^{a}\right) \\ 
& +L_{cd}^{e}\circ h\circ \pi \left( \rho ,\eta \right) \mathring{\Gamma}%
_{e}^{a},%
\end{array}%
\end{equation*}%
and%
\begin{equation*}
\begin{array}{cl}
\Gamma \left( \tilde{\rho},Id_{E}\right) \left( \overset{\circ }{\tilde{%
\delta}}_{c}\right) \left( \left( \rho ,\eta \right) \mathring{\Gamma}%
_{d}^{a}\right) & \displaystyle=\Gamma \left( \tilde{\rho},Id_{E}\right)
\left( \tilde{\delta}_{c}\right) \left( \left( \rho ,\eta \right) \Gamma
_{d}^{a}\right) \vspace*{1mm} \\ 
& \displaystyle+\frac{1}{4}\Gamma \left( \tilde{\rho},Id_{E}\right) \left( 
\tilde{\delta}_{c}\right) \left( \tilde{g}_{d}^{e}\circ h\circ \pi \frac{%
\partial F^{a}}{\partial y^{e}}\right) \vspace*{1mm} \\ 
& \displaystyle-\frac{1}{4}\tilde{g}_{c}^{e}\circ h\circ \pi \frac{\partial
F^{f}}{\partial y^{e}}\frac{\partial }{\partial y^{f}}\left( \left( \rho
,\eta \right) \Gamma _{d}^{a}\right) \vspace*{1mm} \\ 
& \displaystyle-\frac{1}{16}\tilde{g}_{c}^{e}\circ h\circ \pi \frac{\partial
F^{f}}{\partial y^{e}}\frac{\partial }{\partial y^{f}}\left( \tilde{g}%
_{d}^{e}\circ h\circ \pi \frac{\partial F^{a}}{\partial y^{e}}\right) ,%
\end{array}%
\end{equation*}%
\begin{equation*}
\begin{array}{cl}
\Gamma \left( \tilde{\rho},Id_{E}\right) \left( \overset{\circ }{\tilde{%
\delta}}_{d}\right) \left( \left( \rho ,\eta \right) \mathring{\Gamma}%
_{c}^{a}\right) & \displaystyle=\Gamma \left( \tilde{\rho},Id_{E}\right)
\left( \tilde{\delta}_{d}\right) \left( \left( \rho ,\eta \right) \Gamma
_{c}^{a}\right) \vspace*{1mm} \\ 
& \displaystyle+\frac{1}{4}\Gamma \left( \tilde{\rho},Id_{E}\right) \left( 
\tilde{\delta}_{d}\right) \left( \tilde{g}_{c}^{e}\circ h\circ \pi \frac{%
\partial F^{a}}{\partial y^{e}}\right) \vspace*{1mm} \\ 
& \displaystyle-\frac{1}{4}\tilde{g}_{d}^{e}\circ h\circ \pi \frac{\partial
F^{f}}{\partial y^{e}}\frac{\partial }{\partial y^{f}}\left( \left( \rho
,\eta \right) \Gamma _{c}^{a}\right) \vspace*{1mm} \\ 
& \displaystyle-\frac{1}{16}\tilde{g}_{d}^{e}\circ h\circ \pi \frac{\partial
F^{f}}{\partial y^{e}}\frac{\partial }{\partial y^{f}}\left( \tilde{g}%
_{c}^{e}\circ h\circ \pi \frac{\partial F^{a}}{\partial y^{e}}\right) ,%
\end{array}%
\end{equation*}%
\begin{equation*}
\begin{array}{cl}
L_{cd}^{e}\circ h\circ \pi \cdot \left( \rho ,\eta \right) \mathring{\Gamma}%
_{e}^{a} & =L_{cd}^{e}\circ h\circ \pi \cdot \left( \rho ,\eta \right)
\Gamma _{e}^{a}\vspace*{1mm} \\ 
& \displaystyle+L_{cd}^{e}\circ h\circ \pi \cdot \left( \tilde{g}%
_{e}^{f}\circ h\circ \pi \frac{\partial F^{a}}{\partial y^{f}}\right)%
\end{array}%
\end{equation*}%
it results the conclusion of the theorem. 
\end{proof}

\begin{proposition}
\label{p63}\emph{If }$S$\emph{\ is the canonical }$\left( \rho ,\eta \right) 
$\emph{-semispray asso\-cia\-ted to the mechanical }$\left( \rho ,\eta
\right) $\emph{-system }$\left( \left( E,\pi ,M\right) ,F_{e},\left( \rho
,\eta \right) \Gamma \right) $\emph{\ and from }$\mathbf{B}^{\mathbf{v}}$%
\emph{-mor\-phism }$\left( g,h\right) $\emph{,\ then}%
\begin{equation}
\begin{array}[t]{l}
2G^{a%
{\acute{}}%
}=2G^{a}M_{a}^{a%
{\acute{}}%
}\circ h\circ \pi -\left( g_{b}^{a}\circ h\circ \pi \right) y^{b}\left( \rho
_{a}^{i}\circ h\circ \pi \right) \frac{\partial y^{a%
{\acute{}}%
}}{\partial x^{i}}.%
\end{array}
\label{eq90}
\end{equation}%
\emph{\ }
\end{proposition}

\begin{proof}
Since the Jacobian matrix of coordinates transformation is 
\begin{equation*}
\left\Vert 
\begin{array}{ll}
\,\ \ \ \ \ \ \ M_{a}^{a%
{\acute{}}%
}\circ h\circ \pi & \,\ 0\vspace*{1mm} \\ 
\rho _{a}^{i}\circ \left( h\circ \pi \right) \displaystyle\frac{\partial
M_{a}^{a%
{\acute{}}%
}\circ \pi }{\partial x^{i}}y^{a} & M_{a}^{a%
{\acute{}}%
}\circ \pi%
\end{array}%
\right\Vert =\left\Vert 
\begin{array}{ll}
\,\ \ \ \ \ \ \ M_{a}^{a%
{\acute{}}%
}\circ h\circ \pi & \,\ 0\vspace*{1mm} \\ 
\rho _{a}^{i}\circ \left( h\circ \pi \right) \displaystyle\frac{\partial y^{a%
{\acute{}}%
}}{\partial x^{i}} & M_{a}^{a%
{\acute{}}%
}\circ \pi%
\end{array}%
\right\Vert
\end{equation*}%
and 
\begin{equation*}
\begin{array}{c}
\left\Vert 
\begin{array}{ll}
\,\ \ \ \ \ \ \ M_{a}^{a%
{\acute{}}%
}\circ h\circ \pi & \,\ 0\vspace*{1mm} \\ 
\rho _{a}^{i}\circ \left( h\circ \pi \right) \displaystyle\frac{\partial y^{a%
{\acute{}}%
}}{\partial x^{i}} & M_{a}^{a%
{\acute{}}%
}\circ \pi%
\end{array}%
\right\Vert \cdot \left( 
\begin{array}{l}
\,\ \ \ \left( g_{b}^{a}\circ h\circ \pi \right) y^{b}\vspace*{1mm} \\ 
-2\left( G^{a}-\displaystyle\frac{1}{4}F^{a}\right)%
\end{array}%
\right) =\left( 
\begin{array}{l}
\,\ \ \left( g_{b%
{\acute{}}%
}^{a%
{\acute{}}%
}\circ h\circ \pi \right) y^{b%
{\acute{}}%
}\vspace*{1mm} \\ 
-2\left( G^{a%
{\acute{}}%
}-\displaystyle\frac{1}{4}F^{a%
{\acute{}}%
}\right)%
\end{array}%
\right) ,%
\end{array}%
\end{equation*}%
the conclusion results immediately. 
\end{proof}

In the following, we consider a differentiable curve $%
\begin{array}[b]{ccc}
I & ^{\underrightarrow{~c~}} & M%
\end{array}%
$ and its $\left( g,h\right) $-lift $\dot{c}.$


\begin{definition}
\label{d64}If it is verifies the following equality:%
\begin{equation}
\begin{array}[t]{l}
\frac{d\dot{c}\left( t\right) }{dt}=\Gamma \left( \tilde{\rho},Id_{E}\right)
S\left( \dot{c}\left( t\right) \right) ,%
\end{array}
\label{eq91}
\end{equation}%
\textit{\ }then we say that \emph{the curve }$\dot{c}$\emph{\ is an integral
curve of the }$\left( \rho ,\eta \right) $\emph{-semispray }$S$\emph{\ of
the mechanical }$\left( \rho ,\eta \right) $\emph{-system }$\left( \left(
E,\pi ,M\right) ,F_{e},\left( \rho ,\eta \right) \Gamma \right) $.
\end{definition}


\begin{theorem}
\label{t65}\emph{All }$\left( g,h\right) $\emph{-lifts solutions of the
equations:}%
\begin{equation}
\begin{array}[t]{l}
\frac{dy^{a}\left( t\right) }{dt}+2G^{a}\!\circ u\left( c,\dot{c}\right)
\left( x\left( t\right) \right) {=}\frac{1}{2}F^{a}\!\circ u\left( c,\dot{c}%
\right) \left( x\left( t\right) \right) \!,\,a{\in }\overline{1,\!r},%
\end{array}
\label{eq92}
\end{equation}%
\emph{where }$x\left( t\right) =\left( \eta \circ h\circ c\right) \left(
t\right) ,$ \emph{are } \emph{integral curves of the canonical }$\left( \rho
,\eta \right) $\emph{-semispray asso\-cia\-ted to mechanical }$\left( \rho
,\eta \right) $\emph{-system }$\left( \left( E,\pi ,M\right) ,F_{e},\left(
\rho ,\eta \right) \Gamma \right) $\emph{\ and from locally invertible }$%
\mathbf{B}^{\mathbf{v}}$\emph{-mor\-phism }$\left( g,h\right) .$
\end{theorem}

\begin{proof}
Since the equality 
\begin{equation*}
\begin{array}[t]{l}
\frac{d\dot{c}\left( t\right) }{dt}=\Gamma \left( \tilde{\rho},Id_{E}\right)
S\left( \dot{c}\left( t\right) \right)%
\end{array}%
\end{equation*}%
is equivalent to 
\begin{equation*}
\begin{array}{l}
\displaystyle\frac{d}{dt}((\eta \circ h\circ c)^{i}(t),y^{a}(t))\vspace*{2mm}
\\ 
\qquad =\displaystyle\left( \rho _{a}^{i}\circ \eta \circ h\circ
c(t)g_{b}^{a}\circ h\circ c(t)y^{b}(t),-2\left( G^{a}-\frac{1}{4}%
F^{a}\right) ((\eta \circ h\circ c)^{i}(t),y^{a}(t))\right) ,%
\end{array}%
\end{equation*}%
it results 
\begin{equation*}
\begin{array}{l}
\displaystyle\frac{dy^{a}\left( t\right) }{dt}+2G^{a}\left( x^{i}\left(
t\right) ,y^{a}\left( t\right) \right) =\frac{1}{2}F^{a}\left( x^{i}\left(
t\right) ,y^{a}\left( t\right) \right) ,\ \ a\in \overline{1,n},\vspace*{2mm}
\\ 
\displaystyle\frac{dx^{i}\left( t\right) }{dt}=\rho _{a}^{i}\circ \eta \circ
h\circ c\left( t\right) g_{b}^{a}\circ h\circ c\left( t\right) y^{b}\left(
t\right) ,%
\end{array}%
\end{equation*}%
where $x^{i}\left( t\right) =\left( \eta \circ h\circ c\right) ^{i}\left(
t\right) $.
\end{proof}


\begin{definition}
\label{d66}\textbf{\ }If $S$\ is a $\left( \rho ,\eta \right) $-semispray,
then the vector field%
\begin{equation}
\begin{array}{l}
\left[ \mathbb{C},S\right] _{\left( \rho ,\eta \right) TE}-S%
\end{array}
\label{eq93}
\end{equation}
will be called the \emph{derivation of }$\left( \rho ,\eta \right) $\emph{%
-semispray }$S.$
\end{definition}

The $\left( \rho ,\eta \right) $-semispray $S$ will be called $\left( \rho
,\eta \right) $\emph{-spray} if the following conditions are
verified:\medskip

1. $S\circ 0$ is differentiable of class $C^{1}$\ where $0$\ is the null
section;\smallskip

2. Its derivation is the null vector field.\medskip

The $\left( \rho ,\eta \right) $-semispray $S$\ will be called \emph{%
quadratic }$\left( \rho ,\eta \right) $\emph{-spray }if there are verified
the following conditions:\medskip

1. $S\circ 0$ is differentiable of class $C^{2},$\ where $0$\ is the null
section;\smallskip

2. Its derivation is the null vector field.\medskip

In particular, \ if $\ \left( \rho ,\eta \right) =\left(
id_{TM},Id_{M}\right) $ and $\left( g,h\right) =\left( Id_{E},Id_{M}\right)
, $ \ then \ we \ obtain \ the \ \emph{spray} \ and the \emph{quadratic
spray }which is similar with the classical spray and quadratic spray.

\begin{theorem}
\label{t66}\emph{If }$S$\emph{\ is the canonical }$\left( \rho ,\eta \right) 
$\emph{-spray associated to mechanical }$\left( \rho ,\eta \right) $\emph{%
-system }$\left( \left( E,\pi ,M\right) ,F_{e},\left( \rho ,\eta \right)
\Gamma \right) $\emph{\ and from locally invertible }$\mathbf{B}^{\mathbf{v}%
} $\emph{-morphism }$\left( g,h\right) $\emph{, then}%
\begin{equation}
\begin{array}{cl}
2\left( G^{a}-\frac{1}{4}F^{a}\right) & =\left( \rho ,\eta \right) \Gamma
_{c}^{a}\left( g_{f}^{c}\circ h\circ \pi \right) y^{f} \\ 
& +\frac{1}{2}\left( g_{e}^{d}\circ h\circ \pi \right) y^{e}\left(
L_{dc}^{b}\circ h\circ \pi \right) \left( \tilde{g}_{b}^{a}\circ h\circ \pi
\right) \left( g_{f}^{c}\circ h\circ \pi \right) y^{f} \\ 
& -\frac{1}{2}\left( \rho _{c}^{j}\circ h\circ \pi \right) \frac{\partial
\left( g_{e}^{b}\circ h\circ \pi \right) }{\partial x^{j}}y^{e}\left( \tilde{%
g}_{b}^{a}\circ h\circ \pi \right) \left( g_{f}^{c}\circ h\circ \pi \right)
y^{f} \\ 
& +\frac{1}{2}\left( g_{e}^{b}\circ h\circ \pi \right) y^{e}\left( \rho
_{b}^{i}\circ h\circ \pi \right) \frac{\partial \left( \tilde{g}%
_{c}^{a}\circ h\circ \pi \right) }{\partial x^{i}}\left( g_{f}^{c}\circ
h\circ \pi \right) y^{f}%
\end{array}
\label{eq94}
\end{equation}

\emph{We obtain the spray}%
\begin{equation}
\begin{array}{cl}
S & =\left( g_{b}^{a}\circ h\circ \pi \right) y^{b}\tilde{\partial}%
_{a}-\left( \rho ,\eta \right) \Gamma _{c}^{a}\left( g_{f}^{c}\circ h\circ
\pi \right) y^{f}\overset{\cdot }{\tilde{\partial}}_{a} \\ 
& -\frac{1}{2}\left( g_{e}^{d}\circ h\circ \pi \right) y^{e}\left(
L_{dc}^{b}\circ h\circ \pi \right) \left( \tilde{g}_{b}^{a}\circ h\circ \pi
\right) \left( g_{f}^{c}\circ h\circ \pi \right) y^{f}\overset{\cdot }{%
\tilde{\partial}}_{a} \\ 
& +\frac{1}{2}\left( \rho _{c}^{j}\circ h\circ \pi \right) \frac{\partial
\left( g_{e}^{b}\circ h\circ \pi \right) }{\partial x^{j}}y^{e}\left( \tilde{%
g}_{b}^{a}\circ h\circ \pi \right) \left( g_{f}^{c}\circ h\circ \pi \right)
y^{f}\overset{\cdot }{\tilde{\partial}}_{a} \\ 
& -\frac{1}{2}\left( g_{e}^{b}\circ h\circ \pi \right) y^{e}\left( \rho
_{b}^{i}\circ h\circ \pi \right) \frac{\partial \left( \tilde{g}%
_{c}^{a}\circ h\circ \pi \right) }{\partial x^{i}}\left( g_{f}^{c}\circ
h\circ \pi \right) y^{f}\overset{\cdot }{\tilde{\partial}}_{a}%
\end{array}
\label{eq95}
\end{equation}

\emph{This spray will be called the canonical }$\left( \rho ,\eta \right) $%
\emph{-spray associated to mechanical system }$\left( \left( E,\pi ,M\right)
,F_{e},\left( \rho ,\eta \right) \Gamma \right) $\emph{\ and from locally
invertible }$\mathbf{B}^{\mathbf{v}}$\emph{-morphism }$(g,h).$

\emph{In particular, if }$\left( \rho ,\eta \right) =\left(
id_{TM},Id_{M}\right) $\emph{\ and }$\left( g,h\right) =\left(
Id_{E},Id_{M}\right) ,$\emph{\ then we get the canonical spray associated to
connection }$\Gamma $\emph{\ which is similar with the classical canonical
spray associated to connection }$\Gamma $.
\end{theorem}

\begin{proof}
Since 
\begin{equation*}
\begin{array}[t]{l}
\left[ \mathbb{C},S\right] _{\left( \rho ,\eta \right) TE}=\left[ y^{a}%
\overset{\cdot }{\tilde{\partial}}_{a},\left( g_{e}^{b}\circ h\circ \pi
\cdot y^{e}\right) \tilde{\partial}_{b}\right] _{\left( \rho ,\eta \right)
TE}-2\left[ y^{a}\overset{\cdot }{\tilde{\partial}}_{a},\left( G^{b}-\frac{1%
}{4}F^{b}\right) \overset{\cdot }{\tilde{\partial}}_{b}\right] _{\left( \rho
,\eta \right) TE},%
\end{array}%
\end{equation*}

\begin{equation*}
\!\!%
\begin{array}{cl}
\left[ y^{a}\overset{\cdot }{\tilde{\partial}}_{a},\left( g_{e}^{b}\circ
h\circ \pi \cdot y^{e}\right) \tilde{\partial}_{b}\right] _{\left( \rho
,\eta \right) TE}\!\!\!\! & \displaystyle=\left( g_{e}^{b}\circ h\circ \pi
\cdot y^{e}\right) \tilde{\partial}_{b}\vspace*{2mm}%
\end{array}%
\end{equation*}%
and 
\begin{equation*}
\begin{array}{cl}
\left[ y^{a}\overset{\cdot }{\tilde{\partial}}_{a},\left( G^{b}-\frac{1}{4}%
F^{b}\right) \overset{\cdot }{\tilde{\partial}}_{b}\right] _{\left( \rho
,\eta \right) TE} & \displaystyle=y^{a}\frac{\partial \left( G^{b}-\frac{1}{4%
}F^{b}\right) }{\partial y^{a}}\overset{\cdot }{\tilde{\partial}}_{b}-\left(
G^{b}-\frac{1}{4}F^{b}\right) \overset{\cdot }{\tilde{\partial}}_{b}\vspace*{%
2mm}%
\end{array}%
\end{equation*}%
it results that 
\begin{equation*}
\begin{array}{cc}
\left[ \mathbb{C},S\right] _{\left( \rho ,\eta \right) TE}-S & \displaystyle%
=2\left( -y^{f}\frac{\partial \left( G^{a}-\frac{1}{4}F^{a}\right) }{y^{f}}%
+2\left( G^{a}-\frac{1}{4}F^{a}\right) \right) \overset{\cdot }{\tilde{%
\partial}}_{a}%
\end{array}%
\eqno(S1)
\end{equation*}
Using equality \eqref{eq83}, it results that%
\begin{equation*}
\begin{array}{cl}
\displaystyle\frac{\partial \left( G^{a}-\frac{1}{4}F^{a}\right) }{y^{f}} & 
=\left( \rho ,\eta \right) \Gamma _{c}^{a}\left( g_{f}^{c}\circ h\circ \pi
\right) \\ 
& +\frac{1}{2}\left( g_{e}^{d}\circ h\circ \pi \right) y^{e}\left(
L_{dc}^{b}\circ h\circ \pi \right) \left( \tilde{g}_{b}^{a}\circ h\circ \pi
\right) \left( g_{f}^{c}\circ h\circ \pi \right) \\ 
& -\frac{1}{2}\left( \rho _{c}^{j}\circ h\circ \pi \right) \frac{\partial
\left( g_{e}^{b}\circ h\circ \pi \right) }{\partial x^{j}}y^{e}\left( \tilde{%
g}_{b}^{a}\circ h\circ \pi \right) \left( g_{f}^{c}\circ h\circ \pi \right)
\\ 
& +\frac{1}{2}\left( g_{e}^{b}\circ h\circ \pi \right) y^{e}\left( \rho
_{b}^{i}\circ h\circ \pi \right) \frac{\partial \left( \tilde{g}%
_{c}^{a}\circ h\circ \pi \right) }{\partial x^{i}}\left( g_{f}^{c}\circ
h\circ \pi \right)%
\end{array}%
\eqno(S2)
\end{equation*}

Using equalities $\left( S_{1}\right) $ and $\left( S_{2}\right) $, it
results the conclusion of the theorem. 
\end{proof}

\begin{theorem}
\label{t67}\textbf{\ }\emph{\ All }$\left( g,h\right) $\emph{-lifts
solutions of the following system of equations:}%
\begin{equation}
\begin{array}{l}
\displaystyle\frac{dy^{a}}{dt}+\left( \rho ,\eta \right) \Gamma
_{c}^{a}\left( g_{f}^{c}\circ h\circ \pi \right) y^{f}\vspace*{2mm} \\ 
\displaystyle+\frac{1}{2}\left( g_{e}^{d}\circ h\circ \pi \right)
y^{e}\left( L_{dc}^{b}\circ h\circ \pi \right) \left( \tilde{g}_{b}^{a}\circ
h\circ \pi \right) \left( g_{f}^{c}\circ h\circ \pi \right) y^{f} \\ 
\displaystyle-\frac{1}{2}\left( \rho _{c}^{j}\circ h\circ \pi \right) \frac{%
\partial \left( g_{e}^{b}\circ h\circ \pi \right) }{\partial x^{j}}%
y^{e}\left( \tilde{g}_{b}^{a}\circ h\circ \pi \right) \left( g_{f}^{c}\circ
h\circ \pi \right) y^{f} \\ 
\displaystyle+\frac{1}{2}\left( g_{e}^{b}\circ h\circ \pi \right)
y^{e}\left( \rho _{b}^{i}\circ h\circ \pi \right) \frac{\partial \left( 
\tilde{g}_{c}^{a}\circ h\circ \pi \right) }{\partial x^{i}}\left(
g_{f}^{c}\circ h\circ \pi \right) y^{f}=0,%
\end{array}
\label{eq96}
\end{equation}%
\emph{are the integral curves of canonical }$\left( \rho ,\eta \right) $%
\emph{-spray associated to mechanical }$\left( \rho ,\eta \right) $\emph{%
-system }$\left( \left( E,\pi ,M\right) ,F_{e},\left( \rho ,\eta \right)
\Gamma \right) $\emph{\ and from locally invertible }$\mathbf{B}^{\mathbf{v}}
$\emph{-morphism\ }$\left( g,h\right) .$
\end{theorem}

\section{A Lagrangian formalism for Lagrange mechanical $\left( \protect\rho %
,\protect\eta \right) $-systems}

\ \ \ \ 

Let $\left( \left( E,\pi ,M\right) ,F_{e},L\right) $ be an arbitrary
Lagrange mechanical $\left( \rho ,\eta \right) $-system.

Let $\left( d\tilde{z}^{\alpha },d\tilde{y}^{a}\right) $ be the \emph{%
natural dual} $\left( \rho ,\eta \right) $\emph{-base} of the \emph{natural }%
$\left( \rho ,\eta \right) $\emph{-base} $\left( \tilde{\partial}_{\alpha },%
\overset{\cdot }{\tilde{\partial}}_{a}\right) .$

It is very important to remark that the $1$-forms $d\tilde{z}^{a},d\tilde{y}%
^{a},~a\in \overline{1,p}$ are not the differentials of coordinates
functions as in the classical case, but we will use the same notations. In
this case 
\begin{equation*}
\left( d\tilde{z}^{a}\right) \neq d^{\left( \rho ,\eta \right) TE}\left( 
\tilde{z}^{a}\right) ,
\end{equation*}%
where $d^{\left( \rho ,\eta \right) TE}$ is the exterior differentiation
operator associated to exterior differential $\mathcal{F}\left( E\right) $%
-algebra 
\begin{equation*}
\left( \Lambda \left( \left( \rho ,\eta \right) TE,\left( \rho ,\eta \right)
\tau _{E},E\right) ,+,\cdot ,\wedge \right) .
\end{equation*}

Let $L$ be a regular Lagrangian and let $\left( g,h\right) $\ be a locally
invertible $\mathbf{B}^{\mathbf{v}}$-morphism of $\left( E,\pi ,M\right) $
source and target.

\begin{definition}
\label{d68}The $1$-form%
\begin{equation}
\begin{array}{c}
\theta _{L}=\left( \tilde{g}_{a}^{e}\circ h\circ \pi \cdot L_{e}\right) d%
\tilde{z}^{a}%
\end{array}
\label{eq97}
\end{equation}%
will be called the $1$\emph{-form of Poincar\'{e}-Cartan type associated to
the Lagrangian }$L$ \emph{and to the locally invertible }$\mathbf{B}^{%
\mathbf{v}}$\emph{-morphism }$\left( g,h\right) $.\medskip
\end{definition}

Easily, we obtain:%
\begin{equation}
\begin{array}[t]{l}
\theta _{L}\left( \tilde{\partial}_{a}\right) =\tilde{g}_{b}^{e}\circ h\circ
\pi \cdot L_{e},\,\,\ \theta _{L}\left( \overset{\cdot }{\tilde{\partial}}%
_{b}\right) =0.%
\end{array}
\label{eq98}
\end{equation}

\begin{definition}
\label{d69}The $2$-form 
\begin{equation*}
\omega _{L}=d^{\left( \rho ,\eta \right) TE}\theta _{L}
\end{equation*}%
will be called the $2$\emph{-form of Poincar\'{e}-Cartan type associated to
the Lagrangian }$L$\emph{\ and to the locally invertible }$\mathbf{B}^{%
\mathbf{v}}$\emph{-morphism }$\left( g,h\right) $.\medskip
\end{definition}

By the definition of $d^{\left( \rho ,\eta \right) TE},$ we obtain:%
\begin{equation}
\begin{array}{ll}
\omega _{L}\left( U,V\right) & \displaystyle=\Gamma \left( \tilde{\rho}%
,Id_{E}\right) \left( U\right) \left( \theta _{L}\left( V\right) \right) 
\vspace*{2mm} \\ 
& \displaystyle-\,\Gamma \left( \tilde{\rho},Id_{E}\right) \left( V\right)
\left( \theta _{L}\left( U\right) \right) -\theta _{L}\left( \left[ U,V%
\right] _{\left( \rho ,\eta \right) TE}\right) ,%
\end{array}
\label{eq99}
\end{equation}
for any $U,V\in \Gamma \left( \left( \rho ,\eta \right) TE,\left( \rho ,\eta
\right) \tau _{E},E\right) $.\smallskip

\begin{definition}
\label{d70}The real function%
\begin{equation}
\begin{array}{c}
\mathcal{E}_{L}=y^{a}L_{a}-L%
\end{array}
\label{eq100}
\end{equation}%
will be called the \emph{energy of regular Lagrangian }$L.$
\end{definition}

\begin{theorem}
\label{t71}\emph{The equation}%
\begin{equation}
\begin{array}{c}
i_{S}\left( \omega _{L}\right) =-d^{\left( \rho ,\eta \right) TE}\left( 
\mathcal{E}_{L}\right) ,\,S\in \Gamma \left( \left( \rho ,\eta \right)
TE,\left( \rho ,\eta \right) \tau _{E},E\right) ,%
\end{array}
\label{eq101}
\end{equation}%
\emph{\ has an unique solution }$S_{L}\left( g,h\right) $\emph{\ of the type:%
}%
\begin{equation}
\begin{array}[t]{l}
\left( g_{e}^{a}\circ h\circ \pi \right) y^{e}\tilde{\partial}_{a}-2\left(
G^{a}-\frac{1}{4}F^{a}\right) \overset{\cdot }{\tilde{\partial}}_{a},%
\end{array}
\label{eq102}
\end{equation}%
\emph{\ where}%
\begin{equation}
\begin{array}[t]{l}
-2\left( G^{a}-\frac{1}{4}F^{a}\right) =E_{b}\left( L,\rho ,g,h\right) 
\tilde{L}^{ae}\left( g_{e}^{b}\circ h\circ \pi \right)%
\end{array}
\label{eq103}
\end{equation}%
\emph{\ and}%
\begin{equation}
\begin{array}{cl}
E_{b}\left( L,\rho ,g,h\right) & =\left( \rho _{b}^{i}{\circ }h{\circ }\pi
\right) L_{i}-\left( \rho _{b}^{i}{\circ }h{\circ }\pi \right) y^{a}L_{ia}
\\ 
& -\left( g_{f}^{d}\circ h\circ \pi \right) y^{f}\left( \rho _{d}^{i}{\circ }%
h{\circ }\pi \right) \frac{\partial \left( \left( \tilde{g}_{b}^{e}\circ
h\circ \pi \right) L_{e}\right) }{\partial x^{i}} \\ 
& +\left( g_{f}^{d}\circ h\circ \pi \right) y^{f}\left( \rho _{b}^{i}{\circ }%
h{\circ }\pi \right) \frac{\partial \left( \left( \tilde{g}_{d}^{e}\circ
h\circ \pi \right) L_{e}\right) }{\partial x^{i}} \\ 
& +\left( g_{f}^{d}\circ h\circ \pi \right) y^{f}\left( L_{db}^{c}{\circ }h{%
\circ }\pi \right) \left( \tilde{g}_{c}^{e}\circ h\circ \pi \right) L_{e}%
\end{array}%
\hspace*{-4mm}  \label{eq104}
\end{equation}%
$S_{L}\left( g,h\right) $\textit{\ }will be called \emph{the\ canonical }$%
\left( \rho ,\eta \right) $\emph{-semispray associated to Lagrange
mechanical }$\left( \rho ,\eta \right) $\emph{-system }$\left( \left( E,\pi
,M\right) ,F_{e},L\right) $\emph{\ and from locally invertible }$\mathbf{B}^{%
\mathbf{v}}$\emph{-morphism }$(g,h).$
\end{theorem}

\begin{proof}
We obtain that 
\begin{equation*}
i_{S}\left( \omega _{L}\right) =-d^{\left( \rho ,\eta \right) TE}\left( 
\mathcal{E}_{L}\right)
\end{equation*}%
if and only if 
\begin{equation*}
\omega _{L}\left( S,X\right) =-\Gamma \left( \tilde{\rho},Id_{E}\right)
\left( X\right) \left( \mathcal{E}_{L}\right) ,
\end{equation*}%
for any $X\in \Gamma \left( \left( \rho ,\eta \right) TE,\left( \rho ,\eta
\right) \tau _{E},E\right) .$ Particularly, we obtain:%
\begin{equation*}
\begin{array}[t]{l}
\omega _{L}\left( S,\tilde{\partial}_{b}\right) =-\Gamma \left( \tilde{\rho}%
,Id_{E}\right) \left( \tilde{\partial}_{b}\right) \left( \mathcal{E}%
_{L}\right) .%
\end{array}%
\end{equation*}

If we expand this equality, we obtain%
\begin{equation*}
\begin{array}{l}
\left( g_{f}^{d}\circ h\circ \pi \right) y^{f}\left[ \left( \rho _{d}^{i}{%
\circ }h{\circ }\pi \right) \frac{\partial \left( \left( \tilde{g}%
_{b}^{e}\circ h\circ \pi \right) L_{e}\right) }{\partial x^{i}}-\left( \rho
_{b}^{i}{\circ }h{\circ }\pi \right) \frac{\partial \left( \left( \tilde{g}%
_{d}^{e}\circ h\circ \pi \right) L_{e}\right) }{\partial x^{i}}\right. \\ 
\displaystyle\left. -\left( L_{ab}^{c}{\circ }h{\circ }\pi \right) \left( 
\tilde{g}_{c}^{e}\circ h\circ \pi \right) L_{e}\right] -2\left( G^{a}-\frac{1%
}{4}F^{a}\right) \left( \tilde{g}_{a}^{e}\circ h\circ \pi \right) \cdot
L_{eb}\vspace*{2mm} \\ 
\qquad \displaystyle=\rho _{b}^{i}{\circ }h{\circ }\pi \cdot L_{i}-\left(
\rho _{b}^{i}{\circ }h{\circ }\pi \right) \frac{\partial \left(
y^{a}L_{a}\right) }{\partial x^{i}}.%
\end{array}%
\end{equation*}

After some calculations, we obtain the conclusion of the theorem.
\end{proof}


\begin{remark}
\label{r72}\textrm{If }$F_{e}=0$\textrm{\ and }$\eta =Id_{M},$\textrm{\ then 
}%
\begin{equation*}
\begin{array}{cl}
E_{b}\left( L,\rho ,Id_{E},Id_{M}\right) & =\left( \rho _{b}^{i}{\circ }\pi
\right) L_{i}-\left( \rho _{b}^{i}{\circ }\pi \right)
y^{d}L_{id}+y^{d}\left( L_{db}^{c}{\circ }\pi \right) L_{c}%
\end{array}%
\hspace*{-4mm}
\end{equation*}%
\textrm{and }$S_{L}\left( Id_{E},Id_{M}\right) $\textrm{\ is the canonical }$%
\rho $\textrm{-semispray associated to regular Lagrangian }$L$\textrm{\
which is similar with the semispray presented in \cite{25} by E. Martinez.
(see also \cite{21,26}) The canonical }$\rho $\textrm{-semispray }$%
S_{L}\left( Id_{E},Id_{M}\right) $\textrm{\ is the same }$\rho $\textrm{%
-semispray presented in \cite{3,4} \ }
\end{remark}

\textrm{In addition, if }$F_{e}\neq 0$\textrm{\ and\ }$\left( \rho ,\eta
\right) =\left( Id_{TM},Id_{M}\right) $\textrm{, then }$S_{L}\left(
Id_{E},Id_{M}\right) $\textrm{\ will be the\ canonical semispray presented
in \cite{5,6} by I. Buc\u{a}taru and R. Miron.}

\begin{theorem}
\label{t73}\emph{If }$S_{L}\left( g,h\right) $\textit{\ }\emph{is the\
canonical }$\left( \rho ,\eta \right) $\emph{-semispray associated to
Lagrange mechanical }$\left( \rho ,\eta \right) $\emph{-system }$\left(
\left( E,\pi ,M\right) ,F_{e},L\right) $\emph{\ and from locally invertible }%
$\mathbf{B}^{\mathbf{v}}$\emph{-morphism }$(g,h),$ \emph{then the\ real
local functions}%
\begin{equation}
\begin{array}{cl}
\left( \rho ,\eta \right) \Gamma _{c}^{a} & =-\frac{1}{2}\left( \tilde{g}%
_{c}^{d}\circ h\circ \pi \right) \frac{\partial \left( E_{b}\left( L,\rho
,g,h\right) \tilde{L}^{ae}\left( g_{e}^{b}\circ h\circ \pi \right) \right) }{%
\partial y^{d}} \\ 
& -\frac{1}{2}\left( g_{e}^{d}\circ h\circ \pi \right) y^{e}\left(
L_{dc}^{f}\circ h\circ \pi \right) \left( \tilde{g}_{f}^{a}\circ h\circ \pi
\right)  \\ 
& +\frac{1}{2}\left( \rho _{c}^{j}\circ h\circ \pi \right) \frac{\partial
\left( g_{e}^{b}\circ h\circ \pi \right) }{\partial x^{j}}y^{e}\left( \tilde{%
g}_{b}^{a}\circ h\circ \pi \right)  \\ 
& -\frac{1}{2}\left( g_{e}^{b}\circ h\circ \pi \right) y^{e}\left( \rho
_{b}^{i}\circ h\circ \pi \right) \frac{\partial \left( \tilde{g}%
_{c}^{a}\circ h\circ \pi \right) }{\partial x^{i}}%
\end{array}
\label{eq105}
\end{equation}%
\emph{\ are the components of a }$\left( \rho ,\eta \right) $\emph{%
-connection }$\left( \rho ,\eta \right) \Gamma $\emph{\ for the vector
bundle }$\left( E,\pi ,M\right) $\emph{\ which will be called the }$\left(
\rho ,\eta \right) $\emph{-connection associated to Lagrange mechanical }$%
\left( \rho ,\eta \right) $\emph{-system }$\left( \left( E,\pi ,M\right)
,F_{e},L\right) $\emph{\ and from locally invertible }$\mathbf{B}^{\mathbf{v}%
}$\emph{-morphism} $(g,h).$
\end{theorem}

\emph{In the particular case of Lie algebroids, }$\eta =h=Id_{M}$\emph{\ and 
}$g=Id_{E},$\emph{\ we obtain}%
\begin{equation}
\begin{array}{cl}
\rho \Gamma _{c}^{a} & =\displaystyle-\frac{1}{2}\frac{\partial \left(
E_{b}\left( L,\rho ,Id_{E},Id_{M}\right) \tilde{L}^{ab}\right) }{\partial
y^{c}}-\frac{1}{2}y^{b}L_{bc}^{a}\circ \pi .%
\end{array}
\label{eq106}
\end{equation}

\begin{theorem}
\label{t74}\emph{The parallel }$\left( g,h\right) $\emph{-lifts with respect
to }$\left( \rho ,\eta \right) $\emph{-connection }$\left( \rho ,\eta
\right) \Gamma $ \emph{are the\ integral curves of the canonical }$\left(
\rho ,\eta \right) $\emph{-semispray associated to\ mechanical }$\left( \rho
,\eta \right) $\emph{-system }$\left( \left( E,\pi ,M\right) ,F_{e},L\right) 
$ \emph{and from locally invertible }$\mathbf{B}^{\mathbf{v}}$\emph{%
-morphism }$\left( g,h\right) .$
\end{theorem}

\begin{definition}
\label{d75} The equations%
\begin{equation}
\begin{array}{c}
\,\dfrac{dy^{a}\left( t\right) }{dt}-\left( E_{b}\left( L,\rho ,g,h\right) 
\tilde{L}^{ae}\left( g_{e}^{b}\circ h\circ \pi \right) \right) \circ u\left(
c,\dot{c}\right) \left( x\left( t\right) \right) =0,%
\end{array}
\label{eq107}
\end{equation}%
where $x\left( t\right) =\eta \circ h\circ c\left( t\right) $, will be
called the \emph{equations of Euler-Lagrange type associated to Lagrange
mechanical }$\left( \rho ,\eta \right) $\emph{-system }$\left( \left( E,\pi
,M\right) ,F_{e},L\right) $\emph{\ and from locally invertible }$\mathbf{B}^{%
\mathbf{v}}$\emph{-morphism }$\left( g,h\right) .$
\end{definition}

The equations%
\begin{equation}
\begin{array}{c}
\dfrac{dy^{a}\left( t\right) }{dt}-\left( E_{b}\left( L,\rho
,Id_{E},Id_{M}\right) \tilde{L}^{ab}\right) \circ u\left( c,\dot{c}\right)
\left( x\left( t\right) \right) =0,%
\end{array}
\label{eq108}
\end{equation}%
where $x\left( t\right) =c\left( t\right) $, will be called the \emph{%
equations of Euler-Lagrange type associated to Lagrange mechanical} $\rho $%
\emph{-system }$\left( \left( E,\pi ,M\right) ,F_{e},L\right) $.

\begin{remark}
\label{r76}\textrm{The\ integral curves of the canonical }$\left( \rho ,\eta
\right) $\textrm{-semispray associated to mechanical }$\left( \rho ,\eta
\right) $\textrm{-system }$\left( \left( E,\pi ,M\right) ,F_{e},L\right) $%
\textrm{\ and from locally invertible\ }$B^{\mathbf{v}}$\textrm{-morphism }$%
\left( g,h\right) $\textrm{\ are the }$\left( g,h\right) $\textrm{-lifts
solutions for the equations of Euler-Lagrange type }\eqref{eq107}.
\end{remark}

Using our theory, we obtain the following

\begin{theorem}
\label{t77}\emph{If }$F$\emph{\ is a Finsler fundamental function, then the\
geodesics on the manifold }$M$\emph{\ are the curves such that the
components of their }$\left( g,h\right) $\emph{-lifts are\ solutions for the
equations of Euler-Lagrange type }\eqref{eq107}.
\end{theorem}

Therefore, it is natural to propose to extend the study of the Finsler
geometry from the usual Lie algebroid $\left( \left( TM,\tau _{M},M\right) ,%
\left[ ,\right] _{TM},\left( Id_{TM},Id_{M}\right) \right) ,$ to an
arbitrary (generalized) Lie algebroid $\left( \left( E,\pi ,M\right) ,\left[
,\right] _{E,h},\left( \rho ,\eta \right) \right) .$

\ \ \emph{\ }

\section*{Acknowledgment}

\addcontentsline{toc}{section}{Acknowledgment}

I would like to thank R\u{a}dine\c{s}ti-Gorj Cultural Scientifique Society
for financiar support. In memory of Prof. Dr. Gheorghe RADU and Acad. Dr.
Doc. Cornelius RADU. Dedicated to Acad. Dr. Doc. Radu MIRON from Iassy
University, Romania, at his $86^{th}$ anniversary.

\bigskip

\hfill

\end{document}